\documentclass{article}

\usepackage{arxiv}

\usepackage[utf8]{inputenc}
\usepackage[T1]{fontenc}
\usepackage{amsmath,amssymb,amsthm}
\usepackage{array,booktabs,tabularx,longtable}
\usepackage{graphicx}
\usepackage[font=small,labelfont=bf]{caption}
\usepackage{microtype}
\usepackage{hyperref}
\hypersetup{colorlinks=true,linkcolor=black,citecolor=black,urlcolor=blue,
            breaklinks=true,
            pdftitle={Cluster-Graph Edit Distance: Optimal Explicit
                      Embeddings, Metric Proxies, and Complexity},
            pdfauthor={JiYe Liu, Wenkai Wang, Qiang Tian, Wenjun Wang}}

\theoremstyle{plain}
\newtheorem{theorem}{Theorem}[section]
\newtheorem{lemma}[theorem]{Lemma}
\newtheorem{proposition}[theorem]{Proposition}
\newtheorem{corollary}[theorem]{Corollary}
\newtheorem{conjecture}[theorem]{Conjecture}
\theoremstyle{definition}
\newtheorem{definition}[theorem]{Definition}
\newtheorem{remark}[theorem]{Remark}
\numberwithin{equation}{section}

\title{Cluster-Graph Edit Distance:\\ Optimal Explicit Embeddings,
       Metric Proxies, and Complexity}

\author{%
  JiYe Liu\thanks{ORCID 0009-0006-4765-0053} \\
  School of Artificial Intelligence\\
  Tianjin University\\
  Tianjin 300354, China\\
  \texttt{jayliu@tju.edu.cn}
  \And
  Wenkai Wang\thanks{ORCID 0000-0003-1702-1871} \\
  School of Artificial Intelligence\\
  Tianjin University\\
  Tianjin 300354, China\\
  \texttt{wkwang@tju.edu.cn}
  \And
  Qiang Tian\thanks{ORCID 0000-0002-6641-3639} \\
  School of Computer and Information Engineering\\
  Tianjin Normal University\\
  Tianjin 300384, China\\
  \texttt{tianqiang@tjnu.edu.cn}
  \And
  Wenjun Wang\thanks{Corresponding author. ORCID 0009-0009-3547-6061} \\
  School of Artificial Intelligence\\
  Tianjin University\\
  Tianjin 300354, China\\
  \texttt{wjwang@tju.edu.cn}
}

\date{}

\begin{document}
\maketitle

\begin{abstract}
The cluster graphs on $n$ vertices, the disjoint unions of complete graphs, have the integer partitions of $n$ as their isomorphism classes, and the quotient edit distance $q^*(\lambda,\mu)=\min_{\sigma\in S_n}|E(G_\lambda)\triangle\sigma E(G_\mu)|$ makes that set a metric space. Its geometry and its complexity both issue from one identity: $q^*$ is an affine function of the maximum of $\lVert X\rVert_F^2$ over the contingency tables with margins $\lambda,\mu$. Our main result is an explicit optimal embedding. The weighted dyadic sums of the Ferrers staircase, taken at the critical exponent $\frac14$, give a map $F_n$ into $\ell_2^{\,<4n}$ that acts on a single partition and is computable in $O(n)$ time, and its distortion is $\Theta(n^{1/4})$. That order is optimal, since $c_2(\mathcal K_n)=\Theta(n^{1/4})$: the lower half comes from a $\Theta(\sqrt n)$-dimensional Hamming cube of partitions and Enflo's theorem, so no other external input is needed. The analytic core is a scale-free inverse inequality for every integer sequence with $v(1)=v(N+1)=0$ and $v(s)-v(s+1)\in s\mathbb Z$: its critical dyadic energy is at least $\lVert v\rVert_1^2/(63504\sqrt{\mathrm{TV}(v)})$. Combinatorially the same identity yields two explicit $\ell_1$ models, the vertex-mass metric on sorted degree sequences with $\frac12\delta_1\le q^*<\frac32\delta_1$ and the block-energy metric with $q^*\le B\le2q^*-1$, both constants optimal; hence $c_1(\mathcal K_n)\le2$, and an $O(n\log n)$-time algorithm returns an alignment of cost below $2q^*$ with a two-sided certificate. Deciding $q^*(\lambda,\mu)\le Q$ is strongly NP-complete with no FPTAS, while the farthest alignment is polynomial. The best constant is open: an exactly solvable chirp family caps the scale-free constant at $\frac{2\sqrt2}3$ and, since it attains $\mathrm{TV}(v)=2n$, caps the $n$-normalized constant over realizable pairs at $\frac23$.
\end{abstract}


\section{Introduction}\label{sec:1}
\subsection{Cluster profiles and quotient edit distance}\label{sec:1.1}

Comparing two graphs up to isomorphism means relabelling one of them optimally and counting the disagreements. For $n$-vertex graphs this is the \textbf{quotient edit distance}
\begin{equation*}
q^*(G,H)=\min_{\sigma\in S_n}\bigl|E(G)\,\triangle\,\sigma E(H)\bigr| ,
\end{equation*}
the Hamming metric on labelled graphs pushed down along the action of the symmetric group. It is a metric on isomorphism classes, and it is tied to the quadratic assignment problem: minimizing $|E(G)\triangle\sigma E(H)|$ over $\sigma$ is a quadratic assignment instance, and graph isomorphism is its zero-cost decision version. Few sharp quantitative results about it are known, whether exact approximation ratios or exact distortion orders, and this remains so even on restricted graph classes; the state of the art nearest to this setting is surveyed in \S\ref{sec:1.4}.

This paper supplies such results for one class: the \textbf{cluster graphs}, disjoint unions of complete graphs, which are small enough for exact answers to be available and large enough for them not to be trivial. Two features drive everything below.

First, the combinatorics is exact. A cluster graph is determined by the multiset of its block sizes, so
\begin{equation*}
\mathcal K_n/\!\cong\;\longleftrightarrow\;\{\lambda:\lambda\vdash n\}
\end{equation*}
and $q^*$ becomes a metric on the integer partitions of $n$.

Second, the metric has an exact algebraic form. Writing $T(\lambda,\mu)$ for the contingency tables with row margins $\lambda$ and column margins $\mu$, Theorem~\ref{thm:2.3} below states
\begin{equation*}
q^*(\lambda,\mu)=\frac{\lVert\lambda\rVert_2^2+\lVert\mu\rVert_2^2}2-\max_{X\in T(\lambda,\mu)}\lVert X\rVert_F^2 .
\end{equation*}
We call this the \textbf{transportation representation}. It is the single hub of the paper. Every result below is a reading of it, and there are exactly two readings.

Its \emph{combinatorial} reading asks which tables are optimal and what one unit of transported mass costs. Section 3 extracts from it two explicit $\ell_1$ metrics on partitions, each two-sidedly equivalent to $q^*$ with an optimal constant, and Sections 4 and 5 turn those into Euclidean geometry: the exact distortion order of the class, and explicit coordinate maps measured against it.

Its \emph{computational} reading presents $q^*$ as the maximization of a convex objective over a transportation polytope. Section 8 shows that the degree of that objective and the direction of the optimization are jointly decisive, and derives an exact and parameterized classification of the nearest and farthest alignment problems, together with a dichotomy between them.

These are two readings of one identity, not two investigations placed side by side.

\textbf{Where the object comes from.} On cluster graphs the labelled version of this distance is a familiar statistic. A partition of a labelled vertex set into blocks is a clustering, and for two clusterings $P,Q$ of the same $n$ points the quantity $|E(G_P)\triangle E(G_Q)|$ counts the pairs of points on which they disagree, being together in one clustering and separated in the other. That is the pair-counting distance underlying the Rand index, and it is exactly half the Mirkin metric:
\begin{equation*}
M(P,Q)=\sum_i|P_i|^2+\sum_j|Q_j|^2-2\sum_{ij}x_{ij}^2\;=\;2\,\bigl|E(G_P)\,\triangle\,E(G_Q)\bigr| ,
\end{equation*}
where $x_{ij}=|P_i\cap Q_j|$ is the contingency table of the two clusterings. Minimizing over relabellings therefore produces a correspondence-free comparison: a distance between two clustering \emph{profiles}, that is, between the multisets of block sizes, rather than between two labelled clusterings of one data set. That is the appropriate object when the two clusterings live on different data sets, or on one data set with no known correspondence between its points. It also explains why the extremal problem behind Theorem~\ref{thm:2.3}, maximizing $\sum_{ij}x_{ij}^2$ over tables with prescribed margins, is the central quantity of the comparison-of-clusterings literature (\S\ref{sec:1.4}): that maximization is this same optimization, read from the other side.

Throughout we do not distinguish a partition from its isomorphism class, and $p(n)=|\mathcal K_n/\!\cong|$.

\subsection{Main results}\label{sec:1.2}

Let $\mathrm{sd}(\lambda)\in\mathbb Z_{\ge0}^n$ be the sorted degree sequence of $G_\lambda$, put $\delta_1(\lambda,\mu)=\lVert\mathrm{sd}(\lambda)-\mathrm{sd}(\mu)\rVert_1$, and let
\begin{equation*}
e(\lambda)=\Bigl(\tbinom{\lambda_1}2,\dots,\tbinom{\lambda_{k(\lambda)}}2,0,\dots,0\Bigr)\in\mathbb Z_{\ge0}^{\,n},\qquad B(\lambda,\mu)=\lVert e(\lambda)-e(\mu)\rVert_1
\end{equation*}
be the \textbf{block-energy vector} and the metric it induces. Distortion means product distortion: $\rho(f)=L^+(f)/L^-(f)$ with $L^\pm$ the two Lipschitz constants of $f$, and $c_p(\mathcal X)=\inf_f\rho(f)$ over embeddings into $\ell_p$ (\S\ref{sec:2.5}).

The paper's centre is a single map. Let $u_\lambda(s)=\sum_{i:\lambda_i\ge s}\lambda_i$ be the Ferrers staircase of $\lambda$, let $N$ be the least power of two at least $n$, and let $\mathcal D_N$ be the dyadic subintervals of $[N]$, of which there are $2N-1$. Define
\begin{equation}\label{eq:1.1}
F_n(\lambda)=\Bigl(2^{-\ell(I)/4}\textstyle\sum_{s\in I}u_\lambda(s)\Bigr)_{I\in\mathcal D_N}\ \in\ \mathbb R^{2N-1} ,
\end{equation}
the weighted dyadic sums of the staircase at exponent $\frac14$. Section 5 identifies $\frac14$ as the critical endpoint selected by the two obstruction families, and says what is and is not proved about the exponents below it.

\begin{theorem}[an explicit embedding of optimal order]\label{thm:1.1}
 For every $n\ge2$ the map $F_n$ of (1.1) depends only on $n$, acts on a single partition, calls no distance oracle, and is injective on $\mathcal K_n/\!\cong$. Its dimension is $2N-1<4n$ and it is computable in $O(n)$ time and $O(n)$ machine words from either partition representation. Its distortion satisfies
\begin{equation*}
\rho(F_n)\ \le\ 1662\,n^{1/4} ,
\end{equation*}
and no embedding of $(\mathcal K_n/\!\cong,q^*)$ into a Hilbert space does better in order:
\begin{equation*}
c_2(\mathcal K_n)=\Theta\bigl(n^{1/4}\bigr),\qquad \rho(F_n)=\Theta\bigl(n^{1/4}\bigr) .
\end{equation*}

\emph{(Proposition~\ref{prop:5.1}, Theorem~\ref{thm:5.11}, Lemma~\ref{lem:4.2} and Corollary~\ref{cor:5.12}.)}
\end{theorem}

Both halves of $c_2(\mathcal K_n)=\Theta(n^{1/4})$ are proved here. The lower half places a $\Theta(\sqrt n)$-dimensional Hamming cube inside the metric space at distortion at most $8$ and applies Enflo's theorem; the upper half is $F_n$ itself. In particular the determination of $c_2(\mathcal K_n)$ uses no external input beyond Enflo \cite{ref10}, and no general embedding theorem for finite subsets of $\ell_1$ is invoked (Remark~\ref{rem:4.3}). To the best of our knowledge $F_n$ is the first explicit pointwise embedding of this metric attaining the optimal distortion order.

The analytic content of Theorem~\ref{thm:1.1} is a statement about integer sequences, and it is proved in a strictly larger generality than the application uses.

\begin{theorem}[quantized inverse-energy theorem]\label{thm:1.2}
 Let $N$ be a power of two and let $v:\{1,\dots,N+1\}\to\mathbb Z$ satisfy
\begin{equation*}
v(1)=v(N+1)=0,\qquad v(s)-v(s+1)\in s\,\mathbb Z\quad(1\le s\le N) .
\end{equation*}
Write $\mathrm{TV}(v)=\sum_{s\le N}|v(s)-v(s+1)|$ and let
\begin{equation*}
\mathcal E(v)=\sum_{\ell=0}^{\log_2N}2^{-\ell/2}\sum_{I\in\mathcal D_N,\,|I|=2^\ell}\Bigl(\sum_{s\in I}v(s)\Bigr)^2
\end{equation*}
be the critical dyadic energy. Then for every non-zero such $v$,
\begin{equation*}
\mathcal E(v)\ \ge\ \frac1{63504}\cdot\frac{\lVert v\rVert_1^{\,2}}{\sqrt{\mathrm{TV}(v)}} .
\end{equation*}

\emph{(Theorem~\ref{thm:6.2}.)}
\end{theorem}

Three things about this statement. It applies to every closed sequence obeying the divisibility, whether or not it is a difference of two Ferrers staircases, so it is not a theorem about partitions. It is scale-free, normalized by $\mathrm{TV}(v)$ rather than by the budget $\mathrm{TV}(v)\le2n$ that realizability supplies. And some arithmetic hypothesis is indispensable: for an alternating $\pm1$ pattern the inequality fails by a factor $\Theta(\sqrt N)$, and the divisibility is exactly what forbids that pattern (\S\ref{sec:6.1}). The constant is not optimized and is far from the truth; \S\ref{sec:6.7} says by how much and why.

The bridge from Theorem~\ref{thm:1.2} to Theorem~\ref{thm:1.1} is the following pair of $\ell_1$ models, which also carry the paper's algorithmic results. They are not interchangeable: each is optimal for a different purpose.

\begin{center}\begin{small}
\begin{tabularx}{\linewidth}{>{\raggedright\arraybackslash}p{0.13\linewidth}>{\raggedright\arraybackslash}X>{\raggedright\arraybackslash}X>{\raggedright\arraybackslash}X}
\toprule
\textbf{model} & \textbf{definition} & \textbf{relation to $q^*$} & \textbf{carries} \\
\midrule
vertex-mass $\delta_1$ & $\lVert\mathrm{sd}(\lambda)-\mathrm{sd}(\mu)\rVert_1$ & $\tfrac12\delta_1\le q^*<\tfrac32\delta_1$, both constants optimal & the cube lower bound (\S\ref{sec:4.1}), the staircase coordinates and all of \S\ref{sec:5}--\S\ref{sec:6} \\
block-energy $B$ & $\lVert e(\lambda)-e(\mu)\rVert_1$ & $q^*\le B\le2q^*-1$, the constant $2$ optimal and never attained & the $2$-approximation, $c_1(\mathcal K_n)\le2$ \\
\bottomrule
\end{tabularx}
\end{small}\end{center}

\begin{theorem}[two $\ell_1$ models, approximation, and $c_1$]\label{thm:1.3}
\textbf{(i)} \emph{(Vertex-mass model; Theorem~\ref{thm:3.1}.)} For all $n\ge2$ and all $\lambda\ne\mu\vdash n$,
\begin{equation*}
\tfrac12\,\delta_1(\lambda,\mu)\ \le\ q^*(\lambda,\mu)\ <\ \tfrac32\,\delta_1(\lambda,\mu) .
\end{equation*}
The constant $\frac12$ is attained; the constant $\frac32$ is the supremum over all $n$ and all pairs and is attained by none. Hence the sorted-degree map into $\ell_1$ has distortion strictly below $3$ for each $n$, with supremum $3$, and that factor is optimal \emph{for this map}.

\textbf{(ii)} \emph{(Block-energy model; Corollary~\ref{cor:3.10}.)} For all $n$ and all $\lambda\ne\mu\vdash n$,
\begin{equation*}
q^*(\lambda,\mu)\ \le\ B(\lambda,\mu)\ \le\ 2\,q^*(\lambda,\mu)-1 .
\end{equation*}
The constant $2$ is optimal and is approached but never attained, along $\lambda=(2k)$, $\mu=(k,k)$. This is a consequence of a per-table inequality (Theorem~\ref{thm:3.9}) valid for \emph{every} $X\in T(\lambda,\mu)$, whose slack is an explicit measure of how far $X$ is from a bijection between blocks.

\textbf{(iii)} \emph{(Approximation and $c_1$; Corollaries 3.12 and 3.13.)} An $O(n\log n)$-time algorithm outputs a feasible alignment together with its cost and the certificate $B$, of cost at most $B$ and hence strictly less than $2q^*(\lambda,\mu)$ whenever $\lambda\ne\mu$; the same certificate localizes the optimum to $q^*\in\bigl[\lceil(B+1)/2\rceil,\,B\bigr]$. Consequently $c_1(\mathcal K_n)\le2$. Here $n$ is the number of vertices.
\end{theorem}

\begin{theorem}[computational classification]\label{thm:1.4}
 All statements refer to the explicit part-list representation (P1) of \S\ref{sec:2.4}.

\textbf{(i)} \emph{(Hardness; Theorem~\ref{thm:8.1}.)} Deciding whether $q^*(\lambda,\mu)\le Q$ is strongly NP-complete; consequently computing $q^*(\lambda,\mu)$ is strongly NP-hard. This holds already when all blocks of $\lambda$ have equal size.

\textbf{(ii)} \emph{(Parameterized; Corollary~\ref{cor:8.2}.)} Under unary encoding, parameterized by the number $k$ of parts of that uniform side, the problem is $\mathrm W[1]$-hard, and under ETH admits no algorithm of running time $f(k)\cdot L^{o(k/\log k)}$, with $L$ the input length. Under binary encoding it is already NP-hard at $k=2$.

\textbf{(iii)} \emph{(Inapproximability; Corollary~\ref{cor:8.3}.)} Unless $\mathrm P=\mathrm{NP}$, neither $q^*$ nor the underlying maximum $\max_X\lVert X\rVert_F^2$ admits an FPTAS, under the standard binary encoding.

\textbf{(iv)} \emph{(Direction dichotomy; Proposition~\ref{prop:8.4}.)} The farthest alignment $q^{\max}(\lambda,\mu)=\max_\sigma|E(G_\lambda)\triangle\sigma E(G_\mu)|$ is computable exactly in polynomial time, as a minimum-cost flow with separable convex arc costs.
\end{theorem}

Theorem~\ref{thm:1.3}(iii) and Theorem~\ref{thm:1.4}(iii) are consistent: a constant-factor approximation is not an approximation scheme. We make no claim about the existence of a PTAS, and for that reason we describe \S\ref{sec:8} as an exact and parameterized classification together with a nearest/farthest dichotomy, not as a complete one.

\textbf{What remains open, and what does not.} Theorem~\ref{thm:1.2} gives the order and a crude explicit constant. The best constant is open, and it must be quoted in a stated normalization. For the $n$-normalized inequality over realizable pairs, \S\ref{sec:7} exhibits an exactly solvable \emph{chirp} family along which $\mathcal E(v)\sqrt n/\lVert v\rVert_1^2\to\frac23$, so no proof can do better than $\frac23$ there, while what we prove is $1/(63504\sqrt2)$. In the scale-free formulation of Theorem~\ref{thm:1.2} the same family gives $\frac{2\sqrt2}3$, because it attains the budget $\mathrm{TV}(v)=2n$ that realizability supplies as an inequality. Whether either value is the truth we do not know. This is a question about a constant, not about an order: nothing in Theorem~\ref{thm:1.1} is conditional on it.

\subsection{Proof overview}\label{sec:1.3}

The paper has one main line and one branch. The main line runs from the transportation representation to the optimal embedding in five steps, and the complexity results branch off the same representation without touching it.

\textbf{Step 1: the metric becomes an $\ell_1$ norm of a staircase difference.} Write $u_\lambda$ for the Ferrers staircase and $v=v_{\lambda,\mu}=u_\lambda-u_\mu$. Reading a contingency table by vertices bounds the net degree change at a vertex by its degree in the set of flipped edges, which gives $\delta_1\le2q^*$ under \emph{every} alignment (Lemma~\ref{lem:2.5}); routing by rank along elementary unit transfers gives the reverse. Hence $q^*\asymp\lVert v\rVert_1$ with the optimal constants of Theorem~\ref{thm:1.3}(i). Every later step works with $v$ alone.

\textbf{Step 2: the map is a diagonal reweighting of dyadic sums.} By linearity, $\lVert F_n(\lambda)-F_n(\mu)\rVert_2^2=\mathcal E(v)$ exactly. So the upper Lipschitz constant is a summable geometric series, giving $L^+(F_n)\le2\sqrt{2+\sqrt2}$ (Proposition~\ref{prop:5.2}), and the entire problem is the lower bound on $\mathcal E(v)$.

\textbf{Step 3: realizable staircase differences are arithmetically rigid.} Every jump of $v$ satisfies $\Delta v(s)=s\bigl(m_s(\lambda)-m_s(\mu)\bigr)\in s\mathbb Z$, so a non-zero jump at position $s$ has magnitude at least $s$: \emph{oscillation at position $s$ costs amplitude at least $s$}. The same identity gives the budget $\mathrm{TV}(v)\le2n$ (Lemma~\ref{lem:2.6}). These are arithmetic properties of Ferrers diagrams with no continuous analogue, and they are the only properties of realizability that any later step uses.

\textbf{Step 4: the inverse-energy theorem.} Theorem~\ref{thm:1.2} converts Step 3 into $\mathcal E(v)\gtrsim\lVert v\rVert_1^2/\sqrt{\mathrm{TV}(v)}$, and $\mathrm{TV}(v)\le2n$ turns that into $\mathcal E(v)\gtrsim\lVert v\rVert_1^2/\sqrt n$. With Steps 1 and 2 this is $\rho(F_n)=O(n^{1/4})$.

The proof of Theorem~\ref{thm:1.2} has four independent modules. A Whitney tiling of each sign-constant run by the maximal dyadic intervals it contains, at most two per scale, converts $\mathcal E$ into a price bound on every disjoint family of subintervals of runs. A Carleson stopping time spends those prices to peel off the positions where $v$ is locally large, leaving a clean remainder all of whose ancestor intervals are cheap. An island estimate bounds the clean mass lying \emph{above} the half-line $|v(s)|=\frac{s-1}2$ by $14P\sqrt{\mathrm{TV}(v)}$. Finally the folded-remainder capital lemma bounds all the mass \emph{below} that half-line by $8$ times the mass above it: the low region consists of exactly flat plateaus, a plateau can lose height only by paying in a residue that decreases at a controlled rate, and an Abel summation converts the total plateau mass into a sum of such payments. That last module is where the divisibility does more than supply a budget, and it is the part most likely to be reusable elsewhere.

\textbf{Step 5: the matching lower bound.} A $\Theta(\sqrt n)$-dimensional Hamming cube embeds into $(\mathcal K_n/\!\cong,q^*)$ at distortion at most $8$, with block sizes separated by gaps of $2$ so that the encoding is injective and each coordinate gets its own threshold (Lemma~\ref{lem:4.2}). Enflo's theorem then gives $c_2(\mathcal K_n)=\Omega(n^{1/4})$, which applies to every embedding, not to a named map. Combining with Step 4 closes Theorem~\ref{thm:1.1}.

\textbf{Why the exponent is $\frac14$.} Section 5.3 pins it between two failures, both realized by families present in the space at once. A width-one spike of height $\Theta(n)$ is coherent across all $\log_2n$ scales, so the unweighted functional over-prices it, and $\rho(F^{(0)})=\Theta(n^{1/4}\sqrt{\log n})$ exactly. Discounting coarse scales by $2^{-\gamma\ell}$ repairs that for every $\gamma>0$; but over-discounting destroys the coarse levels that carry coherent plateau mass, and on the most classical pair in the space, $K_n$ against the empty graph, every $\gamma>\frac14$ loses $n^{\gamma-1/4}$. The exponent $\frac14$ is where the two failures meet.

\textbf{The branch: degree and direction.} Optimizing a linear objective over the transportation polytope $T(\lambda,\mu)$ is a minimum-cost-flow problem, polynomial in either direction. Replacing it by the quadratic $\sum_{ij}x_{ij}^2$ breaks the symmetry: maximizing $\lVert X\rVert_F^2$ is strongly NP-hard, while minimizing it is a separable convex cost on the arcs of a bipartite network and hence polynomial. Read back through the transportation representation, nearest alignment is hard and farthest alignment is easy.

\subsection{Related work}\label{sec:1.4}

\textbf{Graph edit distance and quadratic assignment.} Exact graph edit distance is intractable in general, and recent work has concentrated on approximation under structural restrictions. Closest to the present setting, Dahan--Grohe--Neuen--Novotný \cite{ref28} study graph edit distance, quadratic assignment and robust graph isomorphism, and obtain $\varepsilon n^2$ \emph{additive} approximation for graphs of bounded VC dimension. Our concern is orthogonal in guarantee type. We ask for exact-decision hardness, multiplicative approximation, and bi-Lipschitz geometry, on a class where all three admit sharp or matching-order results.

\textbf{The extremal problem behind the transportation representation.} As \S\ref{sec:1.1} noted, maximizing $\sum_{ij}x_{ij}^2$ over non-negative integer tables with prescribed margins is not new. It is the central quantity of the comparison-of-clusterings literature, introduced by Hubert--Arabie \cite{ref2} and pursued for its exact maximum by Lerman--Peter \cite{ref1}, Messatfa \cite{ref3}, Brusco--Steinley \cite{ref4}, Steinley et al. \cite{ref5} and Chacón \cite{ref6}, the last of these giving the explicit $2\times2$ solution. Chacón--Rastrojo \cite{ref27} treat the related question in which the numbers of clusters, rather than the full margins, are fixed. The verdict recorded along this line is Hubert and Arabie's, quoted verbatim in \cite{ref1}: \emph{``Constructing an exact bound, conditional on the fixed row and column totals of the given contingency table, is a very difficult problem of combinatorial optimization.''}

Its complexity was settled elsewhere. Kovačević--Stanojević--Šenk \cite{ref7} prove that minimizing the Rényi entropy $H_\alpha$ over a coupling polytope is strongly NP-hard for every $\alpha\ge0$, and that for $\alpha>1$ this is equivalent to maximizing the $\ell_\alpha$ norm. At $\alpha=2$, after writing the coupling as $X/n$, that is the same optimization problem. We therefore do not claim the first complexity classification of this maximization. Section 8 adds three things: the identification with a graph metric, which is what makes the hardness a statement about an edit distance; NP \emph{membership} for an integer-threshold formulation, hence completeness rather than hardness alone; and the parameterized lower bound together with the nearest/farthest dichotomy.

\textbf{Partitions, transfers, and transport.} The partition graph with elementary unit transfers as edges is classical \cite{ref15,ref16,ref17}, and its fine structure is still under study \cite{ref18,ref19}; what we add is a weighting of its edges by the exact metric cost. The identity $\lVert z\rVert_1=W_1(\nu_\lambda,\nu_\mu)$ of Lemma~\ref{lem:2.4} is the classical CDF form of the Wasserstein-1 distance on the line \cite{ref31}. Erickson \cite[Prop.~3.1]{ref14} gives the counting version on Young diagrams; the version we use is mass-weighted.

The block-energy metric of \S\ref{sec:3.2} is a second, unrelated appearance of $W_1$ on the line, this time between the \emph{spectra of block energies}, and the per-table bound of \S\ref{sec:3.4} is proved by a triangle inequality there. The common-block reduction of Lemma~\ref{lem:3.2} is a value-version of \cite[Th.~1]{ref1}, which proves the stronger statement that every optimal table matches a shared margin value to itself. The single-cell $2\times2$ computation inside Lemma~\ref{lem:3.3} is \cite[Th.~3]{ref1} and \cite[Thm~1]{ref6}, convex-parabola endpoint argument included.

\textbf{Metric embeddings.} The determination of $c_2(\mathcal K_n)$ in this paper has exactly one external input: Enflo \cite{ref10}, who gives $c_2(\{0,1\}^d,\ell_1)=\sqrt d$. It supplies the lower half, applied to the cube of partitions of Lemma~\ref{lem:4.2}; the upper half is the explicit map of Theorem~\ref{thm:5.11}. Nothing else is quoted, and in particular no general theorem on the Euclidean distortion of finite subsets of $\ell_1$ is used.

Such theorems do give the same order non-constructively, and Remark~\ref{rem:4.3} records the comparison. Applied through either $\ell_1$ model of \S\ref{sec:3} with $M=p(n)$ and $\log p(n)=\Theta(\sqrt n)$ \cite{ref11}, Bourgain \cite{ref12} with Linial--London--Rabinovich \cite{ref13} gives $O(\sqrt n)$, the negative-type theorem of Arora--Lee--Naor \cite{ref9} gives $O(n^{1/4}\log n)$, and Chang--Naor--Ren \cite{ref22} gives the optimal $O(n^{1/4})$. What Theorem~\ref{thm:5.11} adds over the last of these is not the order but the map: an embedding evaluable on a single partition in $O(n)$ time, in place of an existence statement about an abstractly listed $p(n)$-point metric space.

\textbf{Multiscale methods.} Sections 5 and 6 are written in the language of dyadic analysis, and three of their devices are standard there. The first is the tiling of an interval by the maximal dyadic intervals it contains, with at most two per scale, which is the discrete Whitney decomposition \cite[Ch.~VI]{ref30}. We use it to recover the whole $\ell_1$ mass of a sign-constant run rather than the single interval that the naive argument keeps. The second is the functional $\sum_\ell2^{-\ell/2}E_\ell$, a discrete square function. Appendix A identifies it exactly in Haar coordinates, where it is a homogeneous norm of smoothness index $-\frac12$; a point mass has logarithmically divergent mass in that norm across scales, which is the $\sqrt{\log n}$ of the spike seen analytically. The third is the stopping-time decomposition of \S\ref{sec:6.3}, a Carleson-type selection of the maximal intervals on which a normalized mass exceeds a threshold, together with the layer-cake accounting of \S\ref{sec:6.4}; these are the standard devices of dyadic analysis and we use them in their standard form.

We use this dictionary for orientation only. No statement below depends on it, and the substance of \S\ref{sec:6} is not an analytic estimate but the arithmetic rigidity of realizable staircases, which has no continuous analogue. We do not develop a general Besov framework here.

On the complexity side the 3-Partition template for packing reductions is standard \cite{ref8}, as is the Garey--Johnson mechanism excluding an FPTAS for a strongly NP-hard problem with integral, polynomially bounded objective. Section 8.3 supplies the problem-specific bounds that mechanism requires, separately for each of the two objectives.

\subsection{Organization}\label{sec:1.5}
Section 2 fixes the objects and the input representations, establishes the transportation representation together with the universal inequality $\delta_1\le2q^*$, and records the jump quantization that Sections 5 and 6 run on. Section 3 develops the two $\ell_1$ models, the per-table bound behind the second of them, and their algorithmic and metric consequences. Section 4 proves the lower bound $c_2(\mathcal K_n)=\Omega(n^{1/4})$ from a cube of partitions and locates the sorted-degree coordinate relative to it. Section 5 defines the weighted dyadic family, pins the distortion of its unweighted member exactly, identifies $\gamma=\frac14$ as the critical exponent, and deduces from Section 6 that the critical member has optimal distortion order. Section 6 is independent of the rest of the paper: it proves the quantized inverse-energy theorem for every closed quantized integer sequence. Section 7 calibrates its constant against an exactly solvable family and states the one question that remains. Section 8 reads the same transportation representation computationally. Section 9 collects open problems. Appendix A records an exact Haar identity used only for interpretation, Appendix B the computational cross-checks, which carry no part of any proof, and Appendix C two proofs deferred from the main line.

\section{Transportation representation and metric interfaces}\label{sec:2}
\subsection{Cluster graphs, partitions, and the quotient metric}\label{sec:2.1}

All graphs are finite, simple and uncoloured. A \textbf{cluster graph} is a disjoint union of complete graphs; $\mathcal K_n$ is the class of cluster graphs on $n$ vertices. The connected components are called \textbf{blocks}, and a cluster graph is determined by the multiset of block sizes, so isomorphism classes correspond to integer partitions $\lambda=(\lambda_1\ge\lambda_2\ge\cdots)\vdash n$. We write $G_\lambda$ for a representative, $g(x)=\binom x2$, $E(\lambda)=\sum_ig(\lambda_i)$ for its number of edges, $m_s(\lambda)$ for the number of blocks of size $s$, and $p(n)$ for the number of isomorphism classes. Two counts attached to a partition are used throughout and must be kept apart:
\begin{equation*}
k(\lambda)=\#\{i:\lambda_i>0\}\ \ (\text{the number of parts}),\qquad r(\lambda)=\#\{s:m_s(\lambda)>0\}\ \ (\text{the number of \textit{distinct} part sizes}),
\end{equation*}
so that $r(\lambda)\le k(\lambda)\le n$ and $r(\lambda)=O(\sqrt n)$ by Lemma~\ref{lem:5.3}. The complexity statements of \S\ref{sec:8} are governed by $k$, the geometric statements of \S\ref{sec:5} by $r$. Each vertex has degree one less than the size of its block. The object of study is
\begin{equation*}
q^*(\lambda,\mu)=\min_{\sigma\in S_n}\bigl|E(G_\lambda)\,\triangle\,\sigma E(G_\mu)\bigr| .
\end{equation*}

\begin{proposition}\label{prop:2.1}
$q^*$ is a metric on $\mathcal K_n/\!\cong$.
\end{proposition}

\begin{proof}
Symmetry and non-negativity are immediate, and $q^*(\lambda,\mu)=0$ exactly when some bijection preserves edges, that is when $\lambda=\mu$. For the triangle inequality let $\sigma$ realize $q^*(\lambda,\mu)$ and $\tau$ realize $q^*(\mu,\nu)$. Acting by $\sigma$ on edge sets preserves the size of a symmetric difference, so $|\sigma E(G_\mu)\triangle\sigma\tau E(G_\nu)|=q^*(\mu,\nu)$, and the triangle inequality for symmetric difference gives
\begin{equation*}
q^*(\lambda,\nu)\le\bigl|E(G_\lambda)\triangle\sigma\tau E(G_\nu)\bigr|\le q^*(\lambda,\mu)+q^*(\mu,\nu).
\qedhere
\end{equation*}
\end{proof}

Section 3.2 uses the triangle inequality to bound the distance between the endpoints of a path by the sum of the costs of its steps.

\textbf{Normalization.} We write $d_{\rm edit}=2q^*/n^2$, which places distances and coordinates on a common scale. It cancels from every distortion ratio (Proposition~\ref{prop:4.4}(A)), and all statements below are unaffected by it.

\subsection{Contingency tables and the transportation representation}\label{sec:2.2}

For $\lambda,\mu\vdash n$, let
\begin{equation*}
T(\lambda,\mu)=\Bigl\{X\in\mathbb Z_{\ge0}^{k(\lambda)\times k(\mu)}:\ \textstyle\sum_j x_{ij}=\lambda_i,\ \sum_i x_{ij}=\mu_j\Bigr\}
\end{equation*}
be the integer points of the \textbf{transportation polytope}, $T_{\mathbb R}(\lambda,\mu)$ its real relaxation, and $\lVert X\rVert_F^2=\sum_{ij}x_{ij}^2$.

\begin{lemma}[integrality]\label{lem:2.2}
For integral $\lambda,\mu$ the vertices of $T_{\mathbb R}(\lambda,\mu)$ are integral, and
\begin{equation*}
\max_{X\in T(\lambda,\mu)}\lVert X\rVert_F^2=\max_{X\in T_{\mathbb R}(\lambda,\mu)}\lVert X\rVert_F^2 .
\end{equation*}
In particular the value of any \emph{feasible real} table is a lower bound for the integer maximum.
\end{lemma}

\begin{proof}
The constraint matrix is the incidence matrix of a bipartite graph, hence totally unimodular, so by the Hoffman--Kruskal criterion \cite{ref20} (see \cite{ref21}) integral right-hand sides force integral vertices. A convex function on a compact polytope attains its maximum at a vertex.
\end{proof}

The last sentence is what is used: the proof of Lemma~\ref{lem:3.2} certifies a bound with a real-valued product table and never has to integralize it.

\begin{theorem}[transportation representation]\label{thm:2.3}
For all $n\ge1$ and all $\lambda,\mu\vdash n$, vertex bijections correspond to tables $X\in T(\lambda,\mu)$. Writing
\begin{equation*}
\operatorname{cost}(X)=E(\lambda)+E(\mu)-2\sum_{ij}g(x_{ij})
\end{equation*}
for the number of edges on which the two graphs disagree under the alignment encoded by $X$,
\begin{equation*}
q^*(\lambda,\mu)=\min_{X\in T(\lambda,\mu)}\operatorname{cost}(X)=\frac{\lVert\lambda\rVert_2^2+\lVert\mu\rVert_2^2}2-\max_{X\in T(\lambda,\mu)}\lVert X\rVert_F^2 .
\end{equation*}
\end{theorem}

\begin{proof}
Given a bijection $\pi$, let $x_{ij}$ be the number of vertices lying in the $i$-th block of $\lambda$ and, after $\pi$, in the $j$-th block of $\mu$; row and column sums are $\lambda_i$ and $\mu_j$. Conversely, given $X\in T(\lambda,\mu)$, cut the blocks of $\lambda$ according to the cell sizes, reassemble into the blocks of $\mu$, and take any bijection inside each cell. Under $\pi$, an edge lies in both graphs exactly when its two endpoints share a block on each side, that is, fall in a common cell; so the number of common edges is $\sum_{ij}g(x_{ij})$ and
\begin{equation*}
\bigl|E(G_\lambda)\triangle\pi E(G_\mu)\bigr|=E(\lambda)+E(\mu)-2\sum_{ij}g(x_{ij})=\operatorname{cost}(X) .
\end{equation*}
Minimizing over $\pi$ maximizes the common edges, giving the first form. For the second, $g(a)=\frac{a^2-a}2$ and $\sum_{ij}x_{ij}=n$ give $\sum_{ij}g(x_{ij})=\frac{\lVert X\rVert_F^2-n}2$, $E(\lambda)=\frac{\lVert\lambda\rVert_2^2-n}2$ and $E(\mu)=\frac{\lVert\mu\rVert_2^2-n}2$; the three occurrences of $-n$ cancel.
\end{proof}

\begin{figure}[tb]
\centering
\includegraphics[width=\textwidth]{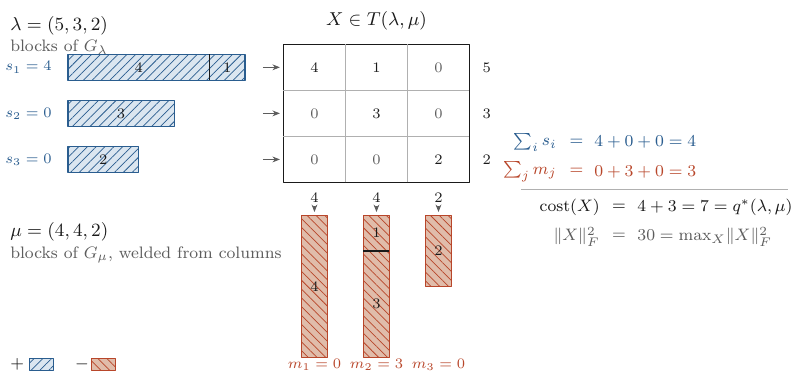}
\caption{\textbf{The transportation table read three ways.}
For $\lambda=(5,3,2)$ and $\mu=(4,4,2)$, the table $X$ records $x_{ij}$, the
number of vertices of the $i$-th block of $G_\lambda$ that the alignment
carries into the $j$-th block of $G_\mu$; its margins are $\lambda$ and $\mu$.
Along the rows $X$ cuts the blocks of $\lambda$, here $\lambda_1=4+1$, at
splitting cost $s_i=\sum_{j<j'}x_{ij}x_{ij'}$; down the columns it welds the
pieces into the blocks of $\mu$, here $\mu_2=1+3$, at merging cost
$m_j=\sum_{i<i'}x_{ij}x_{i'j}$; as an energy it costs
$E(\lambda)+E(\mu)-2\sum_{ij}g(x_{ij})=14+13-20=7$. This $X$ attains
$\max_X\lVert X\rVert_F^2=30$, so it is optimal, and the maximiser is unique
up to exchanging the two blocks of $\mu$ of size~$4$.
Throughout the six figures, blue with a \texttt{///} hatch is the positive,
$\lambda$-side, splitting quantity and vermillion with a
\texttt{\textbackslash\textbackslash\textbackslash} hatch the negative,
$\mu$-side, merging quantity; every colour distinction is doubled by the
hatch, so the figures survive greyscale printing.}
\label{fig:transportation}
\end{figure}

This identity is the hub of the paper, and everything that follows reads it in one of two ways.
Figure~\ref{fig:transportation} carries out both readings on one optimal table of a small pair, priced by rows, by columns, and as an energy.

\emph{Combinatorially}, it converts a minimum over $n!$ alignments into a maximum over tables, and the tables can be manipulated locally. Two such manipulations are carried out in \S\ref{sec:3}, and they are what produce the paper's two $\ell_1$ models. Stripping a shared margin value (Lemma~\ref{lem:3.2}) and solving the resulting $2\times2$ system exactly (Lemma~\ref{lem:3.3}) prices a single unit of transported mass, which the routing argument of \S\ref{sec:3.2} then accumulates. Decomposing $\operatorname{cost}(X)$ into a splitting part and a merging part (Lemma~\ref{lem:3.8}) prices a whole table at once. All of \S\ref{sec:3}--\S\ref{sec:5} rests on these.

\emph{Computationally}, it exhibits $q^*$ as an affine image of the maximum of a convex objective over a transportation polytope. The degree of the objective and the direction of the optimization are then both decisive. The linear counterpart, optimizing $\sum_{ij}c_{ij}x_{ij}$ over $T(\lambda,\mu)$, is a minimum-cost-flow problem and is polynomial in either direction. The quadratic maximum is strongly NP-hard, while the quadratic \emph{minimum} is again polynomial. Section 8 develops all three consequences.

The extremal problem itself is classical \cite{ref1,ref2}; what Theorem~\ref{thm:2.3} contributes is its identification with the quotient edit distance of unlabelled cluster graphs.

\subsection{Staircase coordinates, the universal inequality, and quantization}\label{sec:2.3}

Let $\mathrm{sd}(\lambda)\in\mathbb Z_{\ge0}^n$ be the degree sequence of $G_\lambda$ in non-increasing order, so that it lists the value $\lambda_i-1$ repeated $\lambda_i$ times, for each $i$, sorted; and put
\begin{equation*}
z(\lambda,\mu)=\mathrm{sd}(\lambda)-\mathrm{sd}(\mu)\in\mathbb Z^n,\qquad \delta_1(\lambda,\mu)=\lVert z(\lambda,\mu)\rVert_1 .
\end{equation*}
Reading off multiplicities recovers $\lambda$ from $\mathrm{sd}(\lambda)$, so $\mathrm{sd}$ is injective and $\lambda\ne\mu$ implies $z\ne0$; in particular $\delta_1$ is a metric on $\mathcal K_n/\!\cong$. Both norms of $z$ occur below and are not interchangeable. Here $\lVert z\rVert_1$ is the proxy for $q^*$ in \S\ref{sec:3}, and $\lVert z\rVert_2$ is the distance in the Euclidean coordinate of \S\ref{sec:4.2}. The $n$-dimensional Cauchy--Schwarz factor between them is what turns the constant $\frac32$ into the order $\sqrt n$.

The \textbf{Ferrers staircases} are fixed here once and used unchanged through \S\ref{sec:5}:
\begin{equation*}
u_\lambda(s)=M_\lambda(s)=\sum_{i:\lambda_i\ge s}\lambda_i,\qquad N_\lambda(s)=\#\{i:\lambda_i>s\},\qquad S=\max(\lambda_1,\mu_1),
\end{equation*}
\begin{equation*}
v=v_{\lambda,\mu}=u_\lambda-u_\mu ,
\end{equation*}
so $u_\lambda$ is the staircase weighted by mass and $N_\lambda$ the same staircase weighted by count. We extend $u_\lambda$ by zero above $\lambda_1$ whenever a larger index range is convenient, and adopt the convention $v(s)=0$ for $s$ beyond the support. Note $v(1)=n-n=0$.

\begin{lemma}[bridge identity]\label{lem:2.4}
Let $\nu_\lambda=\sum_{s\ge1}\bigl(s\,m_s(\lambda)\bigr)\delta_s$. Then
\begin{equation*}
\delta_1(\lambda,\mu)=\lVert z(\lambda,\mu)\rVert_1=\sum_{s\ge1}\bigl|u_\lambda(s)-u_\mu(s)\bigr|=\lVert v\rVert_1=W_1(\nu_\lambda,\nu_\mu).
\end{equation*}
\end{lemma}

\begin{proof}
Both $\mathrm{sd}(\lambda)$ and $\mathrm{sd}(\mu)$ are non-increasing staircases, and the $\ell_1$ distance between two such equals that between their generalized inverses, whose level-$s$ counts are $u_\lambda(s)$ and $u_\mu(s)$. The last equality is the classical identity between $W_1$ on the line and the $L^1$ norm of the difference of cumulative distribution functions \cite{ref31}.
\end{proof}

Lemma~\ref{lem:2.4} is a change of coordinates: the same quantity is read either as an $\ell_1$ norm of degree differences, or as a sum over thresholds, or as a transport cost. Section 4.1 needs the threshold form and \S\ref{sec:5} works exclusively with $v$. Erickson \cite[Prop.~3.1]{ref14} gives the counting version on Young diagrams; the version used here is mass-weighted. We use $\delta_1$ and $\lVert v\rVert_1$ interchangeably from now on.

\begin{lemma}[universal inequality]\label{lem:2.5}
For all $n\ge1$ and all $\lambda,\mu\vdash n$,
\begin{equation*}
\delta_1(\lambda,\mu)\ \le\ 2\,q^*(\lambda,\mu) .
\end{equation*}
\end{lemma}

\begin{proof}
\textbf{(i)} Choose $\sigma$ attaining the minimum, put $H=\sigma G_\mu$ and $F=E(G_\lambda)\triangle E(H)$, so $|F|=q^*$. For a vertex $w$ set $a_w=\deg_{G_\lambda}(w)-\deg_H(w)$. Outside $F$ the graphs agree, and those edges contribute equally to both degrees, so $a_w$ receives contributions only from edges of $F$ at $w$: $|a_w|\le t_w$, where $t_w$ is the number of such edges. Each edge has two endpoints, so
\begin{equation*}
\sum_w|a_w|\ \le\ \sum_w t_w\ =\ 2|F|\ =\ 2q^* .
\end{equation*}

\textbf{(ii)} Sorting contracts $\ell_1$ distance: for $x,y\in\mathbb R^n$ with non-increasing rearrangements $x^\downarrow,y^\downarrow$ one has $\lVert x^\downarrow-y^\downarrow\rVert_1\le\lVert x-y\rVert_1$. Indeed, for $\alpha\le\alpha'$ and $\beta\le\beta'$,
\begin{equation*}
|\alpha-\beta|+|\alpha'-\beta'|\ \le\ |\alpha-\beta'|+|\alpha'-\beta| ,
\end{equation*}
which by $|p-q|=2\max(p,q)-p-q$ reduces to $\max(\alpha,\beta)+\max(\alpha',\beta')\le\max(\alpha,\beta')+\max(\alpha',\beta)$, checked according to the positions of $\beta$ and $\beta'$. Hence in $\sum_i|x_i-y_{\pi(i)}|$ any inverted pair may be swapped into same-order position without increasing the cost, and after finitely many swaps the cost is $\lVert x^\downarrow-y^\downarrow\rVert_1$; the identity pairing is one competitor among these.

\textbf{(iii)} Take $x=(\deg_{G_\lambda}(w))_w$, $y=(\deg_H(w))_w$. Rearrangement preserves the multiset of degrees and $H\cong G_\mu$, so $x^\downarrow=\mathrm{sd}(\lambda)$ and $y^\downarrow=\mathrm{sd}(\mu)$. By (ii) and (i),
\begin{equation*}
\lVert z\rVert_1\le\lVert x-y\rVert_1=\sum_w|a_w|\le2q^* .
\qedhere
\end{equation*}
\end{proof}

The inequality holds under every alignment, not merely a canonical one, so it disposes of all cross-block rematchings at once; this is what lets \S\ref{sec:4.1} avoid enumerating contingency tables, and it is the only bridge from $q^*$ to $\lVert v\rVert_1$ used in \S\ref{sec:5}. The $\ell_2$ analogue, Lemma~\ref{lem:4.5}, runs on the same skeleton with the rearrangement inequality in its squared form.

\begin{lemma}[jump quantization and the total-variation budget]\label{lem:2.6}
Let $v=v_{\lambda,\mu}$ and write $\Delta v(s)=v(s)-v(s+1)$ for $1\le s\le N$, with $v(N+1)=0$. Then
\begin{equation*}
\Delta v(s)=s\bigl(m_s(\lambda)-m_s(\mu)\bigr)\in s\,\mathbb Z ,
\end{equation*}
so a non-zero jump at position $s$ has magnitude at least $s$; and
\begin{equation*}
\mathrm{TV}(v):=\sum_{s=1}^N\bigl|\Delta v(s)\bigr|=\sum_s s\,\bigl|m_s(\lambda)-m_s(\mu)\bigr|\ \le\ \sum_s s\bigl(m_s(\lambda)+m_s(\mu)\bigr)=2n .
\end{equation*}
\end{lemma}

\begin{proof}
$u_\lambda(s)-u_\lambda(s+1)=\sum_{i:\lambda_i=s}\lambda_i=s\,m_s(\lambda)$; subtract the same identity for $\mu$. The budget follows from $\sum_ss\,m_s(\lambda)=n$ and the triangle inequality.
\end{proof}

The quantization is the decisive constraint: \emph{oscillation at position $s$ costs amplitude at least $s$.} Cheap wide oscillation is possible only at low positions; at high positions it forces spike-scale amplitudes. Nothing in the ambient dyadic analysis sees this, since it is an arithmetic property of Ferrers staircases, and \S\ref{sec:6} is exactly the argument that exploits it.

\subsection{Input models}\label{sec:2.4}

Three input conventions occur in this paper. They differ by more than a polynomial factor, so each algorithmic and complexity statement below names the one it refers to.

\textbf{(P1) Explicit part-list representation.} Each of the two partitions is written out part by part, $\lambda=(\lambda_1,\dots,\lambda_{k(\lambda)})$, each part in binary. The input length is
\begin{equation*}
L=\Theta\Bigl(\sum_{i}\log(\lambda_i+1)+\sum_j\log(\mu_j+1)\Bigr),
\end{equation*}
so $\Omega(k(\lambda)+k(\mu))\le L\le O\bigl((k(\lambda)+k(\mu))\log n\bigr)$. Under \textbf{unary} encoding of the parts the length is $\Theta(n)$. Thresholds are integers encoded in the same way. \textbf{This is the representation used for all complexity statements in \S\ref{sec:8}.}

\textbf{(P2) Compressed multiplicity representation.} Each partition is given as the set of pairs $\{(s,m_s(\lambda)):m_s(\lambda)>0\}$, both entries in binary. The length is $\Theta\bigl((r(\lambda)+r(\mu))\log n\bigr)$, which may be exponentially shorter than (P1): for $\lambda=(1^{\,n})$ the pair $(1,n)$ has length $O(\log n)$ while (P1) needs $\Theta(n)$. This representation is convenient for evaluating the coordinate maps of \S\ref{sec:4.2} and \S\ref{sec:5}, and we use it only there.

\textbf{(G) Promised graph representation.} The input is the upper-triangular adjacency string $x\in\{0,1\}^{\binom n2}$ of a representative graph under a fixed vertex order, \emph{promised} to encode a cluster graph, with $n$ recovered from the string length. The promise matters: the sorted degree sequence is a complete invariant of $\mathcal K_n$, not of all graphs, so every claim of completeness below is relative to it.

Two consequences of the separation should be recorded, since both are used.

\emph{Vertex count versus input length.} The parameter $n$ is the number of vertices, and it is not interchangeable with the input length. Under (G) one has $L=\Theta(n^2)$; under (P1) with unary parts, $L=\Theta(n)$; but under (P1) with binary parts $L$ may be as small as $O(\log n)$, as for $\lambda=(n)$. A running-time bound stated in $n$ is therefore not a bound in $L$, and running times below are stated in $n$ with the representation named. In particular the $O(n\log n)$ bound of Corollary~\ref{cor:3.12} is a bound in the number of vertices. We do not describe it as near-linear: the output it produces, an alignment of $n$ vertices, already has size $\Omega(n)$ regardless of how short the input is.

\emph{Certificate size.} A contingency table in $T(\lambda,\mu)$ has $k(\lambda)k(\mu)$ cells, which is polynomial in the length of (P1) but need not be polynomial in the length of (P2): for $\lambda=\mu=(1^{\,n})$ it has $n^2$ cells against an input of length $O(\log n)$. This is why \S\ref{sec:8} works in (P1), whose NP membership argument exhibits such a table.

From (G) the sorted degree sequence, hence any of the coordinate maps below, is obtainable in $\Theta(n^2)$ bit operations, the same order as reading the input once. A bit-level accounting of that delivery is in the supplementary material and is used nowhere in what follows.

\subsection{Distortion conventions}\label{sec:2.5}

For an injection $f$ from a metric space $(\mathcal X,d)$ into $\ell_p$ put
\begin{equation*}
L^+(f)=\sup_{x\ne y}\frac{\lVert f(x)-f(y)\rVert_p}{d(x,y)},\qquad L^-(f)=\inf_{x\ne y}\frac{\lVert f(x)-f(y)\rVert_p}{d(x,y)},\qquad \rho(f)=\frac{L^+(f)}{L^-(f)} ,
\end{equation*}
call $\rho$ the \textbf{product distortion}, and set $c_p(\mathcal X)=\inf_f\rho(f)$, the infimum over all injections into $\ell_p$. This paper uses the product distortion throughout.

Two distinctions are used repeatedly. \emph{Class versus map}: $c_p(\mathcal X)$ is an infimum over all embeddings and is an invariant of the space, whereas $\rho(f)$ refers to one named map. Theorem~\ref{thm:4.1} is of the first kind, Proposition~\ref{prop:4.4} and all of \S\ref{sec:5} of the second, and Corollary~\ref{cor:4.8} relates them. \emph{Conventions}: part of the literature uses the symmetrized constant $D(f)=\sqrt{\rho(f)}$, which we do not use again; for the sorted-degree coordinate $\rho=\Theta(\sqrt n)$ and $D=\Theta(n^{1/4})$, and the coincidence of the latter exponent with the exponent in Theorem~\ref{thm:4.1} is an arithmetic accident.

Finally, the Euclidean coordinate studied in \S\ref{sec:4.2} is the scaled sorted degree sequence
\begin{equation*}
\Phi_n([G])=\frac2{n^2}\bigl(d_{(1)},\dots,d_{(n)}\bigr),\qquad d_{(1)}\ge\cdots\ge d_{(n)},
\end{equation*}
an isomorphism invariant, complete on $\mathcal K_n/\!\cong$ by the injectivity noted above, of dimension $n$, and computable without any call to a distance oracle.

\subsection{Notation}\label{sec:2.6}

The symbols below are fixed for the whole paper. Those in the lower block are used only in \S\ref{sec:5}.

\begin{small}
\begin{longtable}{>{\raggedright\arraybackslash}p{0.30\linewidth}>{\raggedright\arraybackslash}p{0.50\linewidth}>{\raggedright\arraybackslash}p{0.12\linewidth}}
\toprule
\textbf{symbol} & \textbf{meaning} & \textbf{fixed in} \\
\midrule\endfirsthead
\toprule
\textbf{symbol} & \textbf{meaning} & \textbf{fixed in} \\
\midrule\endhead
$\lambda,\mu\vdash n$ & partitions of $n$, identified with isomorphism classes of cluster graphs & \S\ref{sec:2.1} \\
$g(x)=\binom x2$, $E(\lambda)=\sum_ig(\lambda_i)$ & pair count; number of edges of $G_\lambda$ & \S\ref{sec:2.1} \\
$m_s(\lambda)$, $k(\lambda)$, $r(\lambda)$, $p(n)$ & multiplicity of the part $s$; number of parts; number of \emph{distinct} part sizes; number of classes & \S\ref{sec:2.1} \\
$q^*$, $q^{\max}$ & nearest and farthest alignment cost & \S\ref{sec:2.1}, \S\ref{sec:8.4} \\
$L$ & input length under the named representation, never the vertex count & \S\ref{sec:2.4} \\
$T(\lambda,\mu)$, $X$, $\operatorname{cost}(X)$ & contingency tables with the two margins; a table; its induced edit cost & \S\ref{sec:2.2} \\
$s_i$, $m_j$, $\operatorname{rowdef}$, $\operatorname{coldef}$ & splitting and merging costs of a table; its row and column deficiencies & \S\ref{sec:3.4} \\
$\mathrm{sd}(\lambda)$, $z$, $\delta_1$ & sorted degree sequence; its difference; the vertex-mass metric $\lVert z\rVert_1$ & \S\ref{sec:2.3} \\
$e(\lambda)$, $B$ & block-energy vector $\bigl(g(\lambda_i)\bigr)_i$; the block-energy metric $\lVert e(\lambda)-e(\mu)\rVert_1$ & \S\ref{sec:3.2} \\
$\widehat e_\lambda^{\,K}$, $\widehat e_X$ & the spectrum of block energies of $\lambda$, of a table $X$, as a measure on $\mathbb R$ with $K$ atoms & \S\ref{sec:3.2}, \S\ref{sec:3.4} \\
$u_\lambda$, $v=u_\lambda-u_\mu$ & Ferrers staircase weighted by mass; the realizable difference & \S\ref{sec:2.3} \\
$N_\lambda$, $S=\max(\lambda_1,\mu_1)$ & staircase weighted by count; largest part occurring & \S\ref{sec:2.3} \\
$L^\pm(f)$, $\rho(f)=L^+/L^-$, $c_p(\mathcal X)$ & Lipschitz constants; product distortion; distortion of the class & \S\ref{sec:2.5} \\
$\Phi_n$ & the scaled sorted-degree coordinate map & \S\ref{sec:2.5} \\
$J=\lceil\log_2n\rceil$, $N=2^J$, $\mathcal D_N$, $\ell(I)$ & top level; padded length; dyadic intervals of $[N]$; level of $I$ & \S\ref{sec:5.1} \\
$E_\ell(v)$, $F^{(\gamma)}$, $k_0(v)$ & level-$\ell$ energy; weighted dyadic map; number of non-zero runs of $v$ & \S\ref{sec:5.1} \\
$\Delta v(s)=v(s)-v(s+1)$, $\mathrm{TV}(v)$, $\tau_i$ & jump at $s$; total variation; total variation localized to the octave $[2^i,2^{i+1})$ & \S\ref{sec:5.3}, \S\ref{sec:6} \\
$\mathcal E(v)=\sum_\ell2^{-\ell/2}E_\ell(v)$, $\kappa$ & critical energy; the normalized ratio $\sqrt{\mathcal E(v)}\,n^{1/4}/\lVert v\rVert_1$ & \S\ref{sec:6}, \S\ref{sec:7} \\
$\mathcal W(R)$, $W_R$, $\mu_R$ & Whitney tiling of a run; its width; its mass & \S\ref{sec:6} \\
$\Gamma$ & shape constant of an octave spectrum, in the hypothesis $\sum_i\sqrt{\tau_i}\le\Gamma\sqrt{\,\cdot\,}$ & \S\ref{sec:6}, \S\ref{sec:7} \\
\bottomrule
\end{longtable}
\end{small}

Three pairs must be kept apart throughout. First, $k(\lambda)$ against $r(\lambda)$: the first governs \S\ref{sec:8}, the second \S\ref{sec:5}. Second, $L$, the input length, against $L^\pm$, the Lipschitz constants. Third, $\mu$, the second partition, against $\mu_R$, the mass of a run; these occur together only in \S\ref{sec:6}.

A few further letters are reused as scratch variables inside a single proof and carry no meaning outside it. They are $D$ and $R$ for the donor and receiver index sets (Lemma~\ref{lem:3.5}), $D$ again for the set of differing coordinates (Lemma~\ref{lem:4.2}), $m$ for $\lfloor n/2\rfloor$ (Lemma~\ref{lem:5.6}) and for the top rung of a ladder (\S\ref{sec:7}), $K$ for the number of atoms of a spectrum (\S\ref{sec:3.4}) and for a decision threshold (\S\ref{sec:8}), and $C$ for a bin capacity (\S\ref{sec:8}). Each is announced where it occurs.

\section{Two explicit $\ell_1$ models}\label{sec:3}

This section develops the combinatorial reading of Theorem~\ref{thm:2.3}. It produces two metrics on $\mathcal K_n/\!\cong$, each given by an explicit coordinate map into $\ell_1$ and each two-sidedly equivalent to $q^*$ with an optimal constant. They are obtained from a table by two different accountings, they are extremal on different families, and neither subsumes the other. The vertex-mass metric is the one that survives into the Euclidean analysis of \S\ref{sec:4}--\S\ref{sec:5}; the block-energy metric gives the better constants for approximation and for $c_1$.

\subsection{Unit transfers and the vertex-mass metric}\label{sec:3.1}

\begin{theorem}\label{thm:3.1}
With the notation of \S\ref{sec:2.3}:

\textbf{(A) (exact cost of a unit transfer)} Move one vertex out of a block of size $p\ge1$ into a \textbf{different} block of size $q\ge0$, where $q=0$ means a new block is created, deleting the source block if it empties; write $\pi\to\pi'$. Then $\pi'=\pi$ if and only if $p=q+1$, and
\begin{equation*}
q^*(\pi,\pi')=\begin{cases}0,&p=q+1,\\[2pt](p-1)+q,&\text{otherwise,}\end{cases}
\qquad
\delta_1(\pi,\pi')=\begin{cases}0,&p=q+1,\\[2pt]2\max(p-1,q),&\text{otherwise.}\end{cases}
\end{equation*}

\textbf{(B) (main bound)} For all $n\ge2$ and all $\lambda\ne\mu\vdash n$,
\begin{equation*}
q^*(\lambda,\mu)\ \le\ \frac{3S-4}{2S}\,\delta_1(\lambda,\mu)\ <\ \frac32\,\delta_1(\lambda,\mu) .
\end{equation*}

\textbf{(C) (supremum)} With Lemma~\ref{lem:2.5} and the family in (D),
\begin{equation*}
\sup_{n\ge2}\ \sup_{\lambda\ne\mu\vdash n}\ \frac{q^*(\lambda,\mu)}{\delta_1(\lambda,\mu)}=\frac32 ,
\end{equation*}
attained by no pair, since $S\le n<\infty$ makes $(3S-4)/(2S)<\frac32$ strict.

\textbf{(D) (witness)} For $b\ge1$ with $b(b+1)\mid n$, the adjacent uniform pair $\lambda=\bigl((b+1)^{\,n/(b+1)}\bigr)$, $\mu=\bigl(b^{\,n/b}\bigr)$ turns both inequalities of (B) into equalities:
\begin{equation*}
q^*=\frac{n(3b-1)}{2(b+1)},\qquad \delta_1=n,\qquad S=b+1,\qquad \frac{q^*}{\delta_1}=\frac{3b-1}{2b+2}\nearrow\frac32 .
\end{equation*}

\textbf{(E) (left end)} For all $n\ge2$ and all $\lambda\ne\mu\vdash n$, $q^*/\delta_1\ge\frac12$, and this is \textbf{attained}: $\lambda=(2,1^{\,n-2})$ against $\mu=(1^{\,n})$ gives $q^*=1$ and $\delta_1=2$.

\textbf{(F) (optimal factor for this map)} The identity map $(\mathcal K_n/\!\cong,q^*)\to(\mathcal K_n/\!\cong,\delta_1)$ has product distortion strictly below $3$ for each fixed $n$, and $\sup_{n\ge2}\rho=3$.
\end{theorem}

The two ends behave differently: the left constant is a minimum, the right one only a supremum. Part (C) is a supremum over both $n$ and the pair; by (B), $\max_{\lambda\ne\mu\vdash n}q^*/\delta_1<\frac32$ for every fixed $n$, and we give no closed form for that maximum. The upper bound in (B) is proved in two stages, through the block-energy metric of \S\ref{sec:3.2}: a routing argument gives $q^*\le B$, and a change of variables compares $B$ with $\delta_1$ term by term.

The upper bound in (A) is immediate. Under the identity bijection the moved vertex loses the $p-1$ edges of its source block and gains the $q$ edges of its target, all other edges unchanged, so the symmetric difference is $(p-1)+q$ and $q^*\le(p-1)+q$. The boundary cases agree: for $q=0$ no edge is gained, cost $p-1$; for $p=1$ the source block disappears and no edge is lost, cost $q$. The degeneracy criterion is equally direct: if $p=q+1$ the two sizes are exchanged and the multiset is unchanged; conversely $\pi'=\pi$ forces $p-1=q$, since only the two entries $p\to p-1$ and $q\to q+1$ change.

The lower bound is the substantive half, and needs two steps: reduce to the two blocks that move, then analyse that subsystem completely.

\begin{lemma}[common-block reduction]\label{lem:3.2}
Let $\alpha,\beta\vdash n$ share a part value $a$, and put $\alpha'=\alpha\setminus\{a\}$, $\beta'=\beta\setminus\{a\}$. Then
\begin{equation*}
\max_{X:(\alpha,\beta)}\lVert X\rVert_F^2\;=\;a^2+\max_{X':(\alpha',\beta')}\lVert X'\rVert_F^2 .
\end{equation*}
\end{lemma}

\begin{proof}
For ``$\ge$'', take an optimal table for $(\alpha',\beta')$, append a row and column, and place $a$ in the new diagonal cell.

For ``$\le$'', let $X$ have margins $(\alpha,\beta)$ with row $i$ and column $j$ both of margin $a$. Write $x=X_{ij}$, $p_{j'}=X_{ij'}$ for $j'\ne j$, $q_{i'}=X_{i'j}$ for $i'\ne i$, and $u=a-x=\sum_{j'}p_{j'}=\sum_{i'}q_{i'}$; let $\Sigma$ be the sum of squares off the cross. Delete row $i$ and column $j$ and move the stranded mass back: choose any plan $t$ with row sums $(q_{i'})$ and column sums $(p_{j'})$, both totalling $u$, and set $Y_{i'j'}=X_{i'j'}+t_{i'j'}$, which has margins $(\alpha',\beta')$. By Lemma~\ref{lem:2.2} the table $Y$ need only be feasible, not integral, so
\begin{equation*}
\max_{X'}\lVert X'\rVert_F^2\ \ge\ \lVert Y\rVert_F^2\ \ge\ \Sigma+\sum t_{i'j'}^2 ,
\end{equation*}
discarding non-negative cross terms. Take the product plan $t_{i'j'}=q_{i'}p_{j'}/u$, assuming $u>0$; if $u=0$ then $x=a$ and the claim is immediate. With $A=\sum q_{i'}^2$, $A'=\sum p_{j'}^2$, $U=u^2$,
\begin{equation*}
\sum t^2=\frac{A\,A'}{U}\ \ge\ A+A'-U ,
\end{equation*}
this being equivalent to $(A-U)(A'-U)\ge0$ with both factors non-positive, since $A\le U$ and $A'\le U$. Now $\lVert X\rVert_F^2=x^2+A+A'+\Sigma$ while the right side is at least $a^2+\Sigma+A+A'-u^2$, so it suffices that $x^2\le a^2-u^2$; and $u=a-x$ gives $a^2-u^2=2ax-x^2$, reducing this to $x\le a$.
\end{proof}

Lerman--Peter \cite[Th.~1]{ref1} prove the stronger statement that every optimal table matches a shared margin value to itself, by induction on $n$. The short proof above gives only the statement about values, which is all we use, and is available because Lemma~\ref{lem:2.2} permits a real-valued certificate.

\begin{lemma}[two-block subsystem]\label{lem:3.3}
Let $p\ne q+1$. Then $q^*(\pi,\pi')=(p-1)+q$.
\end{lemma}

\begin{proof}
Applying Lemma~\ref{lem:3.2} to every part value common to $\pi$ and $\pi'$ leaves $q^*(\pi,\pi')=q^*\bigl(\{p,q\},\{p-1,q+1\}\bigr)$, a block being absent when $q=0$ or $p-1=0$. By Theorem~\ref{thm:2.3}, $q^*=\frac{\lVert\pi\rVert_2^2+\lVert\pi'\rVert_2^2}2-\max_X\lVert X\rVert_F^2$, and the tables with row margins $(p,q)$ and column margins $(p-1,q+1)$ form a one-parameter family:

\begin{center}\begin{small}
\begin{tabular}{lll}
\toprule
 & \textbf{column $p-1$} & \textbf{column $q+1$} \\
\midrule
row $p$ & $c$ & $p-c$ \\
row $q$ & $p-1-c$ & $q+1-p+c$ \\
\bottomrule
\end{tabular}
\end{small}\end{center}

with $\max(0,p-q-1)\le c\le p-1$. The objective
\begin{equation*}
Q(c)=c^2+(p-c)^2+(p-1-c)^2+(q+1-p+c)^2
\end{equation*}
has leading coefficient $4>0$, so it is a convex parabola and its maximum over the interval is at an endpoint: there are exactly two candidates. With $A_0=\frac12(p^2+q^2+(p-1)^2+(q+1)^2)=p^2+q^2-p+q+1$ and $q^*=A_0-\max Q$:

\begin{center}\begin{small}
\begin{tabular}{llll}
\toprule
\textbf{endpoint} & \textbf{exists when} & \textbf{$Q$} & \textbf{$A_0-Q$} \\
\midrule
$c=p-1$ (identity coupling) & always & $(p-1)^2+1+q^2$ & $(p-1)+q$ \\
$c=d:=p-q-1$ & $d\ge0$ & $d^2+(q+1)^2+q^2$ & $d(2q+1)$ \\
$c=0$ & $d<0$, i.e. $p\le q$ & $p^2+(p-1)^2+(q+1-p)^2$ & $-2p^2+2pq+3p-q-1$ \\
\bottomrule
\end{tabular}
\end{small}\end{center}

and $q^*$ is the smaller of the two available values. Case by case: $d=0$ gives $q^*=0$, the excluded degenerate case; $d=1$ gives $d(2q+1)=2q+1=(p-1)+q$, the endpoints agreeing; $d\ge2$ gives $d(2q+1)-[(p-1)+q]=2q(d-1)\ge0$; $p\le q$ gives $[-2p^2+2pq+3p-q-1]-[(p-1)+q]=2(p-q)(1-p)\ge0$; $q=0$ makes the interval the single point $c=p-1$, so $q^*=p-1$; and $p=1$ makes it $\{0\}$, so $q^*=q$.
\end{proof}

The same $2\times2$ extremum is solved in \cite[Th.~3]{ref1} and \cite[Thm~1]{ref6}, the latter by the same convex-parabola argument and the same parametrization; what is used here is its reading as a distance between two partitions, together with the degeneracy criterion.

\textbf{Single-step displacement.} Write $\Delta=\nu_{\pi'}-\nu_\pi$. For $p\ge2$, $\Delta$ equals $-p$ at $p$ and $+(p-1)$ at $p-1$; for $p=1$ it is $-1$ at $1$. For $q\ge1$ it equals $-q$ at $q$ and $+(q+1)$ at $q+1$; for $q=0$ it is $+1$ at $1$. Accumulating $\bigl|\sum_{s\le w}\Delta(s)\bigr|$: for $p>q+1$, $q+(p-q-2)+p=2p-2$; for $p=q+2$, $q+(q+2)=2p-2$; for $p\le q$, $(p-1)+(q-p)+(q+1)=2q$; for $q=0,p\ge2$, $(p-2)+p=2p-2$; and for $p=1,q\ge1$, $(q-1)+(q+1)=2q$. These combine to $2\max(p-1,q)$, completing (A).

\textbf{The obstruction.} By (A), every genuine step satisfies
\begin{equation*}
\frac{q^*_{\rm step}}{\delta_{1,\rm step}}=\frac{(p-1)+q}{2\max(p-1,q)}\in\Bigl[\tfrac12,1\Bigr),
\end{equation*}
the right endpoint excluded because attaining $1$ requires $p-1=q$, the degenerate step. No path is step-wise geodesic for the vertex-mass metric, so the accumulated stretch must be controlled globally, which is why (B) is proved in two stages.

\subsection{Rank routing and the block-energy metric}\label{sec:3.2}

The second model reads a partition not through the degrees of its vertices but through the edge counts of its blocks. Recall $g(x)=\binom x2$ and put
\begin{equation*}
e(\lambda)=\bigl(g(\lambda_1),\dots,g(\lambda_{k(\lambda)}),0,\dots,0\bigr)\in\mathbb Z_{\ge0}^{\,n},\qquad B(\lambda,\mu)=\lVert e(\lambda)-e(\mu)\rVert_1=\sum_{i\ge1}\bigl|g(\lambda_i)-g(\mu_i)\bigr| ,
\end{equation*}
both lists sorted non-increasingly and zero-padded to the common length $n$. We call $e(\lambda)$ the \textbf{block-energy vector}; note $E(\lambda)=\lVert e(\lambda)\rVert_1$, so $e$ resolves the total edge count into its per-block contributions. Where $\delta_1$ matches the two partitions vertex by vertex, $B$ matches them block by block in order of size; the following proposition identifies that matching and records three equivalent forms of $B$.

\begin{proposition}\label{prop:3.4}
\emph{(a)} $e$ is injective on $\mathcal K_n/\!\cong$, and $B$ is a metric on $\mathcal K_n/\!\cong$.

\emph{(b)} For any $K\ge\max\bigl(k(\lambda),k(\mu)\bigr)$,
\begin{equation*}
B(\lambda,\mu)=W_1\bigl(\widehat e_\lambda^{\,K},\widehat e_\mu^{\,K}\bigr),\qquad \widehat e_\lambda^{\,K}=\sum_{i}\delta_{g(\lambda_i)}+\bigl(K-k(\lambda)\bigr)\delta_0 ,
\end{equation*}
the optimal-transport cost on $\mathbb R$ with ground cost $|x-y|$ between the two spectra of block energies.

\emph{(c)} $B(\lambda,\mu)=\sum_{s\ge1}s\,\bigl|N_\lambda(s)-N_\mu(s)\bigr|$.
\end{proposition}

\begin{proof}
(a) $g$ is strictly increasing on $\{1,2,\dots\}$ with $g(1)=0$, so the multiset of parts $\ge2$ is recovered from the positive entries of $e(\lambda)$, and the number of parts equal to $1$ is $n$ minus their sum; hence $e$ separates points, and $B$ inherits symmetry and the triangle inequality from $\ell_1$.

(b) Both measures have total mass $K$. For equal-size multisets on the line the sorted matching is optimal: if $a_1\le a_2$ and $b_1\le b_2$ then $|a_1-b_1|+|a_2-b_2|\le|a_1-b_2|+|a_2-b_1|$, checked on the three interleavings, and any matching is turned into the sorted one by such exchanges without increasing cost \cite{ref31}. Sorting both multisets descending lists $g(\lambda_1)\ge\cdots\ge g(\lambda_{k(\lambda)})\ge0=\cdots=0$ against the analogous list for $\mu$, which is exactly the padded rank pairing defining $B$.

(c) Since $g(x)=\sum_{s\ge0}s\,\mathbb 1[s<x]$ and, for fixed $i$, the difference of the indicators has constant sign,
\begin{equation*}
\sum_i|g(\lambda_i)-g(\mu_i)|=\sum_{s\ge1}s\sum_i\bigl|\mathbb 1[\lambda_i>s]-\mathbb 1[\mu_i>s]\bigr| ,
\end{equation*}
and because both lists are sorted, $\{i:\lambda_i>s\}$ and $\{i:\mu_i>s\}$ are initial segments whose symmetric difference has size $|N_\lambda(s)-N_\mu(s)|$.
\end{proof}

Form (c) is the one that compares with $\delta_1$ in \S\ref{sec:3.3}, form (b) the one that yields the per-table bound in \S\ref{sec:3.4}.

\begin{lemma}[routing by rank]\label{lem:3.5}
For all $\lambda,\mu\vdash n$, $\ q^*(\lambda,\mu)\le B(\lambda,\mu)$.
\end{lemma}

\begin{proof}
Pair the two lists by rank: the $i$-th block of $\lambda$ is responsible for becoming the $i$-th block of $\mu$. Put $D=\{i:\lambda_i>\mu_i\}$ and $R=\{i:\lambda_i<\mu_i\}$, disjoint, with $\sum_i(\lambda_i-\mu_i)=0$ so that total surplus equals total deficit.

\emph{The pairing is frozen.} A partition is a multiset and the sorted order is only a normal form, so the bookkeeping is done on $n$ labelled \textbf{slots}, each carrying its current size. Along the path the normal form of the intermediate partition is reshuffled, but the identity of a slot and of its target is not, and by Theorem~\ref{thm:3.1}(A) the cost of a step depends only on the current sizes of the two slots involved, not on their rank. This is why the total is independent of the interleaving.

\emph{Schedule.} While an unfinished donor remains, pick any unfinished donor and any unfinished receiver and move one vertex. Source and target lie in disjoint index sets, so they are distinct. A donor with $\mu_i=0$ makes its last move out of a block of size $1$, which is then deleted; a receiver with $\lambda_j=0$ makes its first move with $q=0$. Every intermediate state is a partition of $n$, and conservation of mass guarantees that an unfinished receiver implies an unfinished donor, so the schedule does not deadlock.

\emph{Summation.} Each donor's size decreases monotonically and each receiver's increases monotonically, so the ``size after removal'' seen by donor $i$ runs exactly once through $\lambda_i-1,\dots,\mu_i$ and the ``size before insertion'' seen by receiver $j$ runs exactly once through $\lambda_j,\dots,\mu_j-1$. With the single-step bound of \S\ref{sec:3.1},
\begin{equation*}
\sum_t q^*(\pi^{t-1},\pi^t)\ \le\ \sum_{i\in D}\sum_{s=\mu_i}^{\lambda_i-1}s+\sum_{j\in R}\sum_{s=\lambda_j}^{\mu_j-1}s=\sum_{i}\bigl|g(\lambda_i)-g(\mu_i)\bigr|=B(\lambda,\mu) .
\end{equation*}
A degenerate step has true cost $0$, still at most the $(p-1)+q$ charged to it. By Proposition~\ref{prop:2.1}, $q^*(\lambda,\mu)$ is at most the right-hand side.
\end{proof}

The partition graph with elementary unit transfers as edges is classical \cite{ref15,ref16,ref17,ref18,ref19}, where it is studied unweighted; the argument above weights it by the exact $q^*$ cost of a transfer computed in \S\ref{sec:3.1}.

\subsection{Comparing the two models, and sharpness}\label{sec:3.3}

$B$ counts parts above a threshold and $\delta_1$ counts mass above a threshold. The next lemma writes both over the same difference sequence, after which the comparison is term by term.

\begin{lemma}[staircase transform]\label{lem:3.6}
$B(\lambda,\mu)\le\dfrac{3S-4}{2S}\,\delta_1(\lambda,\mu)<\dfrac32\,\delta_1(\lambda,\mu)$.
\end{lemma}

\begin{proof}
For $s\ge1$,
\begin{equation*}
u_\lambda(s)=s\,N_\lambda(s-1)+\sum_{w\ge s}N_\lambda(w),
\end{equation*}
since $N_\lambda(s-1)=\#\{i:\lambda_i\ge s\}$ and $\sum_{w\ge s}N_\lambda(w)=\sum_{i:\lambda_i\ge s}(\lambda_i-s)$. Put $t_s=N_\lambda(s)-N_\mu(s)$ for $s\ge0$, of finite support, and $y_s=\sum_{w\ge s}t_w$; then
\begin{equation*}
v(s)=u_\lambda(s)-u_\mu(s)=s\,t_{s-1}+y_s=s\,y_{s-1}-(s-1)\,y_s .
\end{equation*}
At $s=1$ the left side is $n-n=0$ and the right side is $y_0$, so $y_0=0$; finite support gives $y_s=0$ for $s\ge S$.

For $s\ge1$ put $\eta_s=y_s/s$, so $\eta_s=0$ for $s\ge S$. For $s\ge2$, $s\,y_{s-1}-(s-1)y_s=s(s-1)(\eta_{s-1}-\eta_s)$, and the $s=1$ term vanishes. Let $d_j=\eta_j-\eta_{j+1}$ for $j\ge1$, supported in $[1,S-1]$, so $\eta_s=\sum_{j\ge s}d_j$. By Lemma~\ref{lem:2.4},
\begin{equation*}
\delta_1=\sum_{s\ge1}|v(s)|=\sum_{j\ge1}j(j+1)\,|d_j| .
\end{equation*}
From $t_s=y_s-y_{s+1}=s\,d_s-\sum_{j>s}d_j$ and Proposition~\ref{prop:3.4}(c),
\begin{equation*}
B=\sum_{s\ge1}s\Bigl|s\,d_s-\sum_{j>s}d_j\Bigr|\ \le\ \sum_{s\ge1}s^2|d_s|+\sum_{s\ge1}s\sum_{j>s}|d_j|=\sum_{s\ge1}|d_s|\,\frac{s(3s-1)}2 ,
\end{equation*}
the second sum reorganized by $\sum_{j\ge2}|d_j|\sum_{s=1}^{j-1}s=\sum_j|d_j|\frac{j(j-1)}2$. The ratio of the $s$-th coefficients is
\begin{equation*}
\frac{s(3s-1)/2}{s(s+1)}=\frac{3s-1}{2(s+1)},
\end{equation*}
strictly increasing in $s$ with supremum $\frac32$. Since $d$ is supported in $[1,S-1]$, taking $s=S-1$ gives $B\le\frac{3S-4}{2S}\delta_1<\frac32\delta_1$.
\end{proof}

The loss is entirely in that coefficient ratio, which approaches $\frac32$ only as $s\to\infty$, that is, only for partitions with large blocks. Combining with Lemma~\ref{lem:3.5} gives the chain
\begin{equation*}
q^*\ \le\ B\ \le\ \tfrac{3S-4}{2S}\,\delta_1\ <\ \tfrac32\,\delta_1 ,
\end{equation*}
which is the upper half of Theorem~\ref{thm:3.1}(B). Sharpness needs one exact evaluation.

\begin{lemma}[optimal table for an adjacent uniform pair]\label{lem:3.7}
Let $a=b+1$ with $b\ge1$ and $b(b+1)\mid n$, and put $\lambda=(a^{\,n/a})$, $\mu=(b^{\,n/b})$. Then
\begin{equation*}
\max_{X\in T(\lambda,\mu)}\lVert X\rVert_F^2=\frac n{b+1}\bigl(b^2+1\bigr),\qquad q^*(\lambda,\mu)=\frac{n(3b-1)}{2(b+1)} .
\end{equation*}
\end{lemma}

\begin{proof}
\emph{Upper bound.} All row margins are $a=b+1$ and all column margins are $b$, so every cell is at most $b$ and each row lies in $\{y\ge0,\sum_j y_j=a,\ y_j\le b\}$. The function $\sum_j y_j^2$ is convex, so its maximum is at a vertex, and every vertex of that polytope has all coordinates in $\{0,b\}$ except at most one, that is, is ``$c'$ entries $b$ plus a remainder $t\in[0,b)$''. From $c'b+t=b+1$ we get $c'=1$, $t=1$, so the row maximum is $b^2+1$; with $n/a$ rows, $\max_X\lVert X\rVert_F^2\le\frac n{b+1}(b^2+1)$.

\emph{Attainability.} Let $m=n/(b+1)$ and $m'=n/b$. Let each row fill one entire column with $b$ and deposit its remaining $1$ in some remainder column. This uses $m$ full columns, leaving
\begin{equation*}
m'-m=\frac nb-\frac n{b+1}=\frac n{b(b+1)}=\frac mb
\end{equation*}
columns of total capacity $\frac mb\cdot b=m$, exactly the total mass of the unit cells, which may be placed in any column. The hypothesis $b(b+1)\mid n$ is used here, to make $m/b$ an integer.

\emph{Substitution.} From $\lVert\lambda\rVert_2^2=n(b+1)$ and $\lVert\mu\rVert_2^2=nb$, Theorem~\ref{thm:2.3} gives
\begin{equation*}
q^*=\frac{n(b+1)+nb}2-\frac{n(b^2+1)}{b+1}=\frac{n\bigl[(2b+1)(b+1)-2(b^2+1)\bigr]}{2(b+1)}=\frac{n(3b-1)}{2(b+1)} .
\qedhere
\end{equation*}
\end{proof}

\textbf{The residue hypothesis is not a formality.} Our family has $a=b+1$, and attainability of the row-wise bound genuinely depends on $a\equiv1\pmod b$: for $a\bmod b\notin\{0,1\}$ the bound is in general not attained. Take $(a,b,n)=(5,3,15)$, that is three rows of margin $5$ and five columns of margin $3$, every cell at most $3$. Row-wise convexity allows $3^2+2^2=13$ per row, hence $39$. Attaining it requires every row to consist of one $3$ and one $2$, so three cells of value $3$ and three of value $2$; the $3$-cells occupy three full columns, leaving two columns of total capacity $6$ to be filled by three $2$-cells, and a column of capacity $3$ cannot be composed of $2$'s. The true maximum is $37$. The obstruction is that a unit cell can be split among columns while a cell of size $r\ge2$ cannot.

\begin{proof}[Proof of Theorem~\ref{thm:3.1}]
(A) is \S\ref{sec:3.1} and (B) is Lemmas~\ref{lem:3.5} and~\ref{lem:3.6}.

For (D), equality on the $B$ side is a computation: $N_\lambda(s)=n/(b+1)$ for $s\le b$ and $N_\mu(s)=n/b$ for $s\le b-1$, so $t_s=-\frac n{b(b+1)}$ for $1\le s\le b-1$, $t_b=\frac n{b+1}$, and $t_s=0$ for $s\ge b+1$; hence $\eta_s$ is constant on $s\le b$ and $d$ is non-zero only at $s=b=S-1$, exactly where the term-by-term bound of Lemma~\ref{lem:3.6} is an equality. Directly,
\begin{equation*}
B=\sum_{s=1}^{b-1}\frac{sn}{b(b+1)}+\frac{bn}{b+1}=\frac{n(3b-1)}{2(b+1)},\qquad\delta_1=n .
\end{equation*}
Equality on the $q^*$ side is Lemma~\ref{lem:3.7}. In (C) the upper estimate is (B) and the lower one is this family with $b\to\infty$.

For (E), the bound is Lemma~\ref{lem:2.5}. It is attained: for $n\ge2$, $\lambda=(2,1^{\,n-2})$ and $\mu=(1^{\,n})$ differ by one edge, so $q^*=1$, while $\mathrm{sd}(\lambda)=(1,1,0^{\,n-2})$ and $\mathrm{sd}(\mu)=(0^{\,n})$ give $\delta_1=2$.

For (F), write $R=q^*/\delta_1$; the distortion of the identity map is $\rho_n=(\sup_{\lambda\ne\mu\vdash n}R)/(\inf_{\lambda\ne\mu\vdash n}R)$. By (E) the denominator is $\frac12$ for every $n\ge2$; by (B) the numerator is at most $\frac{3S-4}{2S}<\frac32$ since $S\le n$; so $\rho_n<3$ for each fixed $n$. Letting $b\to\infty$ in (D) drives $\sup_n R\to\frac32$, so $\sup_{n\ge2}\rho_n=3$.
\end{proof}

Lemma~\ref{lem:3.5} alone yields only $q^*\le B$; equality along the family of (D) comes from computing that family's optimal table directly. The asymmetry of the two ends is likewise structural: the left end is realized by the smallest non-trivial difference, a single edge, whereas the right end requires block sizes tending to infinity.

Two remarks fix the status of the constant $3$. First, since $\delta_1$ is by construction the metric induced by $\ell_1^n$ on $p(n)$ points, Theorem~\ref{thm:3.1}(B) and (E) give $c_1(\mathcal K_n)\le3$. Corollary~\ref{cor:3.13} below improves this to $2$ through the other model. So $3$ should be read as the \emph{optimal distortion of the sorted-degree map}, not as the best available bound on $c_1$. Second, the constant $\frac32$ in (B) is an artefact neither of Lemma~\ref{lem:3.5} nor of Lemma~\ref{lem:3.6} alone: both are equalities on the family of (D).

\subsection{The splitting--merging identity and the per-table bound}\label{sec:3.4}

So far $B$ has been an upper bound for $q^*$ obtained from one particular alignment. We now show it is also a lower bound up to the factor $2$. In fact we prove something stronger and more informative: an inequality valid for \emph{every} table, whose slack is an explicit non-negative quantity --- the number of vertices that fail, row-wise and column-wise, to stay in a single cell --- vanishing exactly on the tables that are bijections between the blocks of $\lambda$ and those of $\mu$. Write $k_\lambda=k(\lambda)$ and $k_\mu=k(\mu)$.

\begin{lemma}[splitting--merging identity]\label{lem:3.8}
For non-negative integers $a_1,\dots,a_t$,
\begin{equation*}
g\Bigl(\sum_j a_j\Bigr)=\sum_j g(a_j)+\sum_{j<j'}a_ja_{j'} .
\end{equation*}
Consequently, for every $X\in T(\lambda,\mu)$,
\begin{equation*}
\operatorname{cost}(X)=\sum_{i=1}^{k_\lambda}s_i+\sum_{j=1}^{k_\mu}m_j,\qquad s_i:=\sum_{j<j'}x_{ij}x_{ij'},\quad m_j:=\sum_{i<i'}x_{ij}x_{i'j} ,
\end{equation*}
and for each $i$, $\ g(\lambda_i)=s_i+\sum_jg(x_{ij})$, symmetrically for columns.
\end{lemma}

\begin{proof}
$g(a+b)-g(a)-g(b)=ab$, and induction on $t$ gives the first identity. Applying it to row $i$, whose cells sum to $\lambda_i$, gives $g(\lambda_i)=\sum_jg(x_{ij})+s_i$; summing over rows, $E(\lambda)=\sum_{ij}g(x_{ij})+\sum_is_i$, and symmetrically $E(\mu)=\sum_{ij}g(x_{ij})+\sum_jm_j$. Adding these and subtracting $2\sum_{ij}g(x_{ij})$ gives $\operatorname{cost}(X)=\sum_is_i+\sum_jm_j$ by the definition in Theorem~\ref{thm:2.3}.
\end{proof}

The reading is exact: $s_i$ is the \textbf{splitting cost} of breaking row block $i$ into its cells, $m_j$ the \textbf{merging cost} of assembling column block $j$ from its cells, and every alignment pays splitting plus merging, nothing else. This is the second of the two accountings promised in \S\ref{sec:2.2}; the first, by vertices, produced Lemma~\ref{lem:2.5}.

For $X\in T(\lambda,\mu)$ define the \textbf{row and column deficiencies}
\begin{equation*}
\operatorname{rowdef}(X)=\sum_{i=1}^{k_\lambda}\Bigl(\lambda_i-\max_j x_{ij}\Bigr)\ \ge0,\qquad \operatorname{coldef}(X)=\sum_{j=1}^{k_\mu}\Bigl(\mu_j-\max_i x_{ij}\Bigr)\ \ge0 ,
\end{equation*}
which vanish exactly when every row, respectively every column, keeps all of its mass in a single cell.

\begin{theorem}[per-table energy bound]\label{thm:3.9}
For every $X\in T(\lambda,\mu)$,
\begin{equation*}
B(\lambda,\mu)\ \le\ 2\operatorname{cost}(X)\ -\ \operatorname{rowdef}(X)\ -\ \operatorname{coldef}(X) .
\end{equation*}
\end{theorem}

\begin{proof}
Set $K=k_\lambda k_\mu$, let $\widehat e_\lambda,\widehat e_\mu$ be as in Proposition~\ref{prop:3.4}(b) with this $K$, and let $\widehat e_X=\sum_{i,j}\delta_{g(x_{ij})}$, which also has $K$ atoms, a zero cell contributing an atom at $0$. By Proposition~\ref{prop:3.4}(b) and the triangle inequality for $W_1$,
\begin{equation*}
B=W_1(\widehat e_\lambda,\widehat e_\mu)\ \le\ W_1(\widehat e_\lambda,\widehat e_X)+W_1(\widehat e_X,\widehat e_\mu) .
\end{equation*}
We bound the first term by exhibiting one coupling; the second is symmetric under transposition.

For each row $i$ pick $j_i\in\arg\max_jx_{ij}$. Couple the atom $g(\lambda_i)$ of $\widehat e_\lambda$ with the atom $g(x_{ij_i})$ of $\widehat e_X$, and for each $j\ne j_i$ couple one padding atom $\delta_0$ of $\widehat e_\lambda$ with the atom $g(x_{ij})$. This uses $\sum_i(k_\mu-1)=K-k_\lambda$ padding atoms, exactly the supply. Since $x_{ij_i}\le\lambda_i$ gives $g(x_{ij_i})\le g(\lambda_i)$, the cost of the coupling is
\begin{equation*}
\sum_i\Bigl[\bigl(g(\lambda_i)-g(x_{ij_i})\bigr)+\sum_{j\ne j_i}g(x_{ij})\Bigr] =\sum_i\Bigl[s_i+2\sum_{j\ne j_i}g(x_{ij})\Bigr],
\end{equation*}
using $g(\lambda_i)=s_i+\sum_jg(x_{ij})$ from Lemma~\ref{lem:3.8}. Now
\begin{equation*}
2\sum_{j\ne j_i}g(x_{ij})=\sum_{j\ne j_i}x_{ij}(x_{ij}-1)=\sum_{j\ne j_i}x_{ij}^2-\bigl(\lambda_i-x_{ij_i}\bigr),
\end{equation*}
and since $x_{ij}\le x_{ij_i}$ for every $j$,
\begin{equation*}
\sum_{j\ne j_i}x_{ij}^2\ \le\ \sum_{j\ne j_i}x_{ij}x_{ij_i}\ \le\ \sum_{j<j'}x_{ij}x_{ij'}\ =\ s_i ,
\end{equation*}
the middle step because the unordered pairs $\{j,j_i\}$ with $j\ne j_i$ are distinct. Hence the contribution of row $i$ is at most $2s_i-\bigl(\lambda_i-\max_jx_{ij}\bigr)$, and summing over rows,
\begin{equation*}
W_1(\widehat e_\lambda,\widehat e_X)\le2\sum_is_i-\operatorname{rowdef}(X),\qquad W_1(\widehat e_X,\widehat e_\mu)\le2\sum_jm_j-\operatorname{coldef}(X).
\end{equation*}
Adding and applying Lemma~\ref{lem:3.8} completes the proof.
\end{proof}

\begin{corollary}[energy sandwich, with automatic strictness]\label{cor:3.10}
For all $\lambda\ne\mu\vdash n$,
\begin{equation*}
q^*(\lambda,\mu)\ \le\ B(\lambda,\mu)\ \le\ 2\,q^*(\lambda,\mu)-1 .
\end{equation*}
In particular $\tfrac12B\le q^*\le B$ and the ratio $B/q^*$ is strictly below $2$ on every instance.
\end{corollary}

\begin{proof}
The left inequality is Lemma~\ref{lem:3.5}. For the right, apply Theorem~\ref{thm:3.9} to a table $X$ attaining $\operatorname{cost}(X)=q^*$ (Theorem~\ref{thm:2.3}); it suffices to show $\operatorname{rowdef}(X)+\operatorname{coldef}(X)\ge1$ whenever $\lambda\ne\mu$. Suppose both vanish. Then each row $i$ has all its mass in one cell, $x_{ij_i}=\lambda_i$, and each column $j$ has a cell equal to $\mu_j$. A column's full cell lies in a row all of whose mass is that cell, so the occupied cells define a bijection $\sigma$ between rows and columns with $\lambda_i=\mu_{\sigma(i)}$, that is, $\lambda=\mu$ as partitions, a contradiction. Since $B$ and $\operatorname{cost}$ are integers, $B\le2q^*-1$.
\end{proof}

\begin{proposition}[the constant $2$ is optimal and unattained]\label{prop:3.11}
For $k\ge1$ let $\lambda=(2k)$ and $\mu=(k,k)$ as partitions of $n=2k$. Then $T(\lambda,\mu)$ has the single element $X=(k\ \ k)$, and
\begin{equation*}
q^*=k^2,\qquad B=g(2k)=2k^2-k,\qquad \frac Bq^*=2-\frac1k\ \nearrow\ 2 ,
\end{equation*}
with Theorem~\ref{thm:3.9} an equality: $\operatorname{rowdef}(X)=k$, $\operatorname{coldef}(X)=0$, and $2\operatorname{cost}(X)-\operatorname{rowdef}(X)-\operatorname{coldef}(X)=2k^2-k=B$.
\end{proposition}

\begin{proof}
A $1\times2$ table with row sum $2k$ and column sums $k,k$ is forced, so $q^*=\operatorname{cost}(X)=s_1=k\cdot k$ by Lemma~\ref{lem:3.8}. For $B$: $e(\lambda)=(g(2k),0,\dots)$ and $e(\mu)=(g(k),g(k),0,\dots)$, so $B=\bigl(g(2k)-g(k)\bigr)+g(k)=g(2k)$, which equals $2k^2-k$. The deficiencies are $2k-k=k$ and $(k-k)+(k-k)=0$.
\end{proof}

Thus Theorem~\ref{thm:3.9} is exactly tight along the extremal direction, and its slack $\operatorname{rowdef}+\operatorname{coldef}$ vanishes exactly on block-bijective alignments. The slack is a correction term and not a criterion for equality: by the proof of Corollary~\ref{cor:3.10} the two deficiencies vanish only when $\lambda=\mu$, whereas $B=q^*$ already on non-trivial pairs, the smallest being $\lambda=(2)$ against $\mu=(1,1)$, where $B=q^*=1$ while $\operatorname{rowdef}=1$. We do not characterize the pairs on which the two models agree. Padding both sides of Proposition~\ref{prop:3.11} with a common part $1$ extends the family to odd $n$, and exhaustive enumeration for $2\le n\le18$ (Appendix B) finds
\begin{equation*}
\max_{\lambda\ne\mu\vdash n}\frac{B}{q^*}=2-\frac1{\lfloor n/2\rfloor}
\end{equation*}
throughout that range, always attained by this family. We record the pattern as an observation only; nothing below uses it.

\subsection{Algorithmic and metric consequences}\label{sec:3.5}

Since $q^*$ is a minimization problem, ``approximation'' needs to be said precisely, and the two statements below carry different guarantees. Running times are in the number of vertices $n$, not the input length; see \S\ref{sec:2.4}.

\begin{corollary}[a strict $2$-approximation with a two-sided certificate]\label{cor:3.12}
There is an algorithm which, given $\lambda,\mu\vdash n$ in representation (P1), outputs a \textbf{feasible alignment} (a labelled graph $\pi^T$ on a fixed vertex set $V$ isomorphic to $G_\mu$) together with its cost and the certificate $B(\lambda,\mu)$, such that
\begin{equation*}
q^*(\lambda,\mu)\ \le\ \bigl|E(G_\lambda)\triangle E(\pi^T)\bigr|\ \le\ B(\lambda,\mu)\ \le\ 2\,q^*(\lambda,\mu)-1 ,
\end{equation*}
in $O(n\log n)$ time plus output time: $O(n\log n)$ for a compact alignment, $O(n^2)$ if the aligned edge set or adjacency matrix is written explicitly. The certificate localizes the optimum two-sidedly,
\begin{equation*}
q^*(\lambda,\mu)\ \in\ \Bigl[\bigl\lceil\tfrac{B+1}2\bigr\rceil,\ B\Bigr] ,
\end{equation*}
an interval of multiplicative width below $2$. If $\lambda=\mu$ the algorithm returns the identity alignment and all quantities are $0$.
\end{corollary}

\begin{proof}
Sort both lists, zero-pad to a common length, pair by rank, and set $D,R$ as in Lemma~\ref{lem:3.5}. On a fixed vertex set $V$ of size $n$, realize $G_\lambda$, then run the schedule of \S\ref{sec:3.2}, each step moving one vertex from an unfinished donor slot to an unfinished receiver slot; there are $T=\sum_i(\lambda_i-\mu_i)^+\le n$ steps and the result $\pi^T$ has block-size multiset $\mu$.

Labels never change, so all intermediate graphs live on the same $V$ and
\begin{equation*}
E(G_\lambda)\triangle E(\pi^T)\subseteq\bigcup_t\bigl(E(\pi^{t-1})\triangle E(\pi^t)\bigr) .
\end{equation*}
At fixed labels a step deletes exactly $p-1$ edges and adds exactly $q$, so its symmetric difference is $(p-1)+q$, and the summation of Lemma~\ref{lem:3.5} gives $\sum_t|E(\pi^{t-1})\triangle E(\pi^t)|=B$. Since $\pi^T\cong G_\mu$, the left-hand side is the cost of a feasible alignment, so $q^*\le|E(G_\lambda)\triangle E(\pi^T)|\le B$, and Corollary~\ref{cor:3.10} supplies $B\le2q^*-1$ for $\lambda\ne\mu$. Inverting that inequality gives $q^*\ge(B+1)/2$, and $q^*$ is an integer, which is the stated interval. For $\lambda=\mu$ the sets $D$ and $R$ are empty, so $T=0$ and all quantities vanish.

For the running time: sorting costs $O(n\log n)$; the $T\le n$ steps maintain block sizes in $O(1)$ each; and both the cost of the alignment and the certificate $B=\sum_i|g(\lambda_i)-g(\mu_i)|$ are computable in $O(n)$ from the block-size lists, the cost by accumulating the contingency table of the fixed labelling vertex by vertex rather than expanding any edge set.
\end{proof}

The output is a feasible solution, its cost, and a verifiable certificate that is simultaneously an upper and a lower bound; no distance oracle is used, and the approximation ratio is strictly below $2$ on every instance rather than in the limit. The bound $O(n\log n)$ is in the number of vertices: the alignment itself has $n$ entries, so no algorithm producing one can run faster than $\Omega(n)$, however short the input encoding may be.

\begin{corollary}\label{cor:3.13}
$c_1(\mathcal K_n)\le2$.
\end{corollary}

\begin{proof}
The map $\lambda\mapsto e(\lambda)\in\ell_1^{\,n}$ is injective (Proposition~\ref{prop:3.4}(a)) and $\ell_1^n$ embeds isometrically into $\ell_1$. By Corollary~\ref{cor:3.10} its two Lipschitz constants against $q^*$ satisfy $L^+\le2$ and $L^-\ge1$, so $\rho\le2$.
\end{proof}

This is an upper bound only, and it is not claimed optimal for the class: what Proposition~\ref{prop:3.11} shows is that $2$ is optimal \emph{for the block-energy map}. The corresponding constant for the sorted-degree map is $3$ (Theorem~\ref{thm:3.1}(F)), so the two models are genuinely different embeddings and not reparametrizations of one another. The value of $c_1(\mathcal K_n)$ remains open (\S\ref{sec:9}).

\emph{A numerical estimate, not an alignment.} One may also discard the alignment and keep only a number: $\widehat q=B/\sqrt2$, computable in $O(n\log n)$, satisfies $\max\{\widehat q/q^*,\,q^*/\widehat q\}\le\sqrt2$, since $B/q^*\in[1,2)$ and $1/\sqrt2$ is the geometric midpoint of that interval. Note that $\widehat q$ is the cost of no alignment and in general is not an integer: it is a two-sided numerical estimate of multiplicative error $\sqrt2$, not a $\sqrt2$-approximation algorithm.

\section{The optimal lower bound for Hilbert embeddings}\label{sec:4}
\begin{theorem}\label{thm:4.1}
$c_2(\mathcal K_n)=\Theta\bigl(n^{1/4}\bigr)$.
\end{theorem}

The two halves are proved independently. Section 4.1 embeds a large Hamming cube into the metric space, using no external input beyond Enflo's theorem on cubes. The upper half is Theorem~\ref{thm:5.11}, proved constructively in Section 5; Remark~\ref{rem:4.3} records the non-constructive route through the optimal Euclidean distortion of finite subsets of $\ell_1$, which the first version of this work took and which nothing here now needs. Section 4.2 locates the sorted-degree coordinate relative to the answer, and the gap it exhibits is what \S\ref{sec:5} sets out to close.

\subsection{A cube of partitions}\label{sec:4.1}

\begin{lemma}\label{lem:4.2}
Let $n\ge16$ and $d=\lfloor\sqrt n/2\rfloor\ (\ge2)$. There is an injection $\lambda:\{0,1\}^d\to\mathcal K_n/\!\cong$ such that for all $x,y\in\{0,1\}^d$, with $h=h(x,y)$ the Hamming distance,
\begin{equation*}
\frac d2\,h\ \le\ q^*\bigl(\lambda(x),\lambda(y)\bigr)\ \le\ 4d\,h .
\end{equation*}
Equivalently, the $d$-dimensional $\ell_1$ Hamming cube embeds into $(\mathcal K_n/\!\cong,q^*)$ with distortion at most $8$, and $d=\Theta(\sqrt n)$.
\end{lemma}

\begin{proof}
\textbf{Encoding.} Put
\begin{equation*}
A_i=2d+2i\ (i=1,\dots,d),\qquad \Sigma=\sum_{i=1}^d A_i=3d^2+d,\qquad F=n-\Sigma,
\end{equation*}
\begin{equation*}
\lambda(x)=\bigl(A_1+x_1,\dots,A_d+x_d,\ 1^{\,F-|x|}\bigr),\qquad |x|=\textstyle\sum_i x_i .
\end{equation*}
The sizes total $\Sigma+|x|+(F-|x|)=n$. For the tail to be legal for every $x$ it suffices that $F\ge d$, that is $n\ge3d^2+2d$; and since $d\le\sqrt n/2$,
\begin{equation*}
3d^2+2d\le\frac{3n}4+\sqrt n\le n\iff\sqrt n\le\frac n4\iff n\ge16 .
\end{equation*}
Each $A_i+x_i$ lies in $\{2d+2i,2d+2i+1\}$; these $d$ pairs are disjoint and all exceed $1$, so the multiset of block sizes determines every bit $x_i$, giving injectivity. (For $n<16$ no encoding is needed: Theorem~\ref{thm:4.1} is asymptotic and $n_0=16$ suffices.)

\textbf{Upper bound.} Let $D=\{i:x_i\ne y_i\}$, $D_+=\{i\in D:x_i=1\}$, $D_-=\{i\in D:y_i=1\}$, so $D=D_+\sqcup D_-$ and $|D|=h$. Put
\begin{equation*}
\nu=\bigl(\{A_i+x_i\}_{i\notin D},\ \{A_i\}_{i\in D},\ 1^{\,c}\bigr),\qquad c=F-\sum_{i\notin D}x_i ,
\end{equation*}
which is symmetric in $x$ and $y$ since $x_i=y_i$ off $D$; from $\sum_{i\notin D}x_i\le d\le F$ we get $c\ge0$, and the sizes total $n$.

The partition $\nu$ refines $\lambda(x)$: blocks agree off $D$; for $i\in D_-$ we have $x_i=0$ and the blocks coincide; for $i\in D_+$ the block $A_i+1$ splits into $A_i$ plus a unit block. The number of unit blocks required is then
\begin{equation*}
(F-|x|)+|D_+|=F-\sum_{i\notin D}x_i=c ,
\end{equation*}
exactly the number $\nu$ has, and symmetrically for $\lambda(y)$. So a single $\nu$ refines both. Partitioning $V$ into the blocks of $\nu$ and merging in two ways gives a bijection $\sigma$ with $E(G_\nu)\subseteq E(G_{\lambda(x)})\cap\sigma E(G_{\lambda(y)})$, whence
\begin{equation*}
q^*\le\bigl(E(\lambda(x))-E(\nu)\bigr)+\bigl(E(\lambda(y))-E(\nu)\bigr).
\end{equation*}
Only the blocks indexed by $D_+$ and $D_-$ change and unit blocks carry no edges, so
\begin{equation*}
E(\lambda(x))-E(\nu)=\sum_{i\in D_+}\Bigl[\tbinom{A_i+1}2-\tbinom{A_i}2\Bigr]=\sum_{i\in D_+}A_i ,
\end{equation*}
and symmetrically; adding, $q^*\le\sum_{i\in D}A_i\le h\max_i A_i=4d\,h$.

\textbf{Lower bound.} By Lemma~\ref{lem:2.5}, $q^*\ge\frac12\delta_1$ for every pair, with no reference to any alignment; and by Lemma~\ref{lem:2.4}, $\delta_1=\sum_{s\ge1}|u_{\lambda(x)}(s)-u_{\lambda(y)}(s)|$. Take the $d$ distinct thresholds $s_i=2d+2i+1$. A block $2d+2j+x_j$ satisfies $2d+2j+x_j\ge s_i$ iff $2j+x_j\ge2i+1$: always for $j>i$, iff $x_i=1$ for $j=i$, never for $j<i$; unit blocks never count. Hence
\begin{equation*}
u_{\lambda(x)}(s_i)=\sum_{j>i}(2d+2j+x_j)+x_i(2d+2i+1),
\end{equation*}
and subtracting, the terms $2d+2j$ cancelling,
\begin{equation*}
\Delta_i=\sum_{j>i}(x_j-y_j)+(x_i-y_i)(2d+2i+1).
\end{equation*}
For $i\in D$ we have $|x_i-y_i|=1$ and $\bigl|\sum_{j>i}(x_j-y_j)\bigr|\le d-i$, so $|\Delta_i|\ge(2d+2i+1)-(d-i)=d+3i+1>d$. All terms of the sum being non-negative, keeping only $i\in D$ gives $\delta_1\ge d\,h$ and hence $q^*\ge\frac d2h$.

\textbf{Composition.} The two bounds give ratios in $[\frac12,4]$ against the metric $d\cdot h$, so the distortion is at most $8$, and $d=\lfloor\sqrt n/2\rfloor=\Theta(\sqrt n)$.
\end{proof}

Two features of the encoding are forced. If all $A_i$ were equal then $\lambda(x)$ would depend only on $|x|$ and the cube would collapse to a one-dimensional family; the gap of $2$ buys injectivity and gives each coordinate its own threshold in the lower bound. And the weights satisfy $A_i\in[2d+2,4d]$, so $n\ge\sum_i A_i=\Theta(d^2)$: this construction cannot give more than $d=O(\sqrt n)$, hence not more than $n^{1/4}$ below.

\textbf{Transfer to $c_2$.} Enflo \cite{ref10} gives $c_2(\{0,1\}^d,\ell_1)=\sqrt d$, and distortion is scale-invariant, so rescaling the Hamming metric by $d$ changes nothing. Distortion is also monotone under passage to submetrics. By Lemma~\ref{lem:4.2} the image of $\lambda$ is a submetric of $(\mathcal K_n/\!\cong,q^*)$ at distortion at most $8$ from $(\{0,1\}^d,d\cdot h)$, so for $n\ge16$,
\begin{equation*}
c_2(\mathcal K_n)\ \ge\ \frac{\sqrt d}8\ \ge\ \frac18\sqrt{\frac{\sqrt n}2-1}\ =\ \Omega\bigl(n^{1/4}\bigr).
\end{equation*}
This constrains every embedding, not a named map.

\begin{remark}[general upper bounds for finite subsets of $\ell_1$]\label{rem:4.3}
An upper bound of the same order is also available non-constructively, and it was the route taken in the first version of this work. By Corollary~\ref{cor:3.10} the map $\lambda\mapsto e(\lambda)$ carries $q^*$ to the metric $B$ at distortion at most $2$, and $(\mathcal K_n/\!\cong,B)$ is by definition an $M$-point subset of $\ell_1^n\subseteq\ell_1$ with $M=p(n)$. Any theorem bounding the Euclidean distortion of finite subsets of $\ell_1$ then applies: the negative-type theorem of Arora--Lee--Naor \cite{ref9} gives $O(\sqrt{\log M}\,\log\log M)=O(n^{1/4}\log n)$, and Chang--Naor--Ren \cite{ref22} give the optimal $O(\sqrt{\log M})=O(n^{1/4})$, using $\log p(n)=\Theta(\sqrt n)$ \cite{ref11}. Two things should be said about this route. It is genuinely non-constructive: it supplies no coordinate map evaluable on a single input in $\mathrm{poly}(n)$ time, and the general procedure that computes $c_2$ of an explicitly listed finite metric space is a semidefinite program \cite{ref13,ref29} of size $\mathrm{poly}(p(n))=\exp(O(\sqrt n))$, which moreover needs a distance table whose entries are strongly NP-hard to produce (Theorem~\ref{thm:8.1}). And it is no longer needed: Theorem~\ref{thm:5.11} supplies the upper half of Theorem~\ref{thm:4.1} constructively, so nothing in this paper depends on \cite{ref22}, and the determination of $c_2(\mathcal K_n)$ rests on Enflo's theorem alone.
\end{remark}

\subsection{The sorted-degree coordinate}\label{sec:4.2}

Theorem~\ref{thm:4.1} concerns the class. We now measure the most natural named coordinate against it.

\begin{proposition}\label{prop:4.4}
With $\Phi_n$ as in \S\ref{sec:2.5}, put
\begin{equation*}
L^+(n)=\max_{\lambda\ne\mu\vdash n}\frac{\lVert\Delta\Phi_n\rVert_2}{d_{\rm edit}},\quad L^-(n)=\min_{\lambda\ne\mu\vdash n}\frac{\lVert\Delta\Phi_n\rVert_2}{d_{\rm edit}},\quad \rho(n)=\frac{L^+(n)}{L^-(n)} .
\end{equation*}
Then:

\textbf{(A)} $\dfrac{\lVert\Delta\Phi_n\rVert_2}{d_{\rm edit}}=\dfrac{\lVert z\rVert_2}{q^*}$ for all $n\ge2$ and all $\lambda\ne\mu\vdash n$.

\textbf{(B)} $L^+(n)=\sqrt2$ for every $n\ge2$, attained at $\lambda=(2,1^{n-2})$, $\mu=(1^n)$.

\textbf{(C)} $L^-(n)\ge\dfrac2{3\sqrt n}$.

\textbf{(D)} $\rho(n)\le\dfrac{3\sqrt2}2\sqrt n$ for all $n\ge2$, strictly for every finite $n$.

\textbf{(E)} With $b=b(n)=\max\{j\ge1:j(j+1)\le n\}$, for every $n\ge6$,
\begin{equation*}
\rho(n)\ \ge\ \frac{\sqrt2\,b(3b-1)}{2\sqrt{b(b+1)}} .
\end{equation*}

\textbf{(F)} For every $n\ge6$,
\begin{equation*}
\frac{3\sqrt2}2\Bigl(\sqrt n-\frac73\Bigr)\ <\ \rho(n)\ <\ \frac{3\sqrt2}2\sqrt n ,
\end{equation*}
hence $\rho(n)=\frac{3\sqrt2}2\sqrt n\,(1+o(1))=\Theta(\sqrt n)$ and $L^-(n)=\frac2{3\sqrt n}(1+o(1))$.
\end{proposition}

\begin{proof}[Proof of (A)]
$\Phi_n$ is a linear rescaling, so $\lVert\Delta\Phi_n\rVert_2=\frac2{n^2}\lVert z\rVert_2$ while $d_{\rm edit}=\frac2{n^2}q^*$, and the factor cancels.
\end{proof}

In particular every statement below applies verbatim to the unnormalized sorted-degree vector.

\begin{lemma}\label{lem:4.5}
For all $\lambda\ne\mu\vdash n$, $\lVert z\rVert_2\le\sqrt2\,q^*$, with equality for the single-edge pair. Hence (B).
\end{lemma}

\begin{proof}
Fix an optimal alignment, let $F$ be the flipped edges, $|F|=q^*$, let $t_w$ be the number of edges of $F$ at $w$, and let $a_w$ be the net degree difference at $w$. Additions and deletions at a vertex cancel in part, so $|a_w|\le t_w$; also $\sum_w t_w=2q^*$ and $\max_w t_w\le|F|=q^*$. Therefore
\begin{equation*}
\lVert a\rVert_2^2\ \le\ \Bigl(\max_w|a_w|\Bigr)\sum_w|a_w|\ \le\ q^*\cdot2q^*=2(q^*)^2 .
\end{equation*}
Here $a$ is indexed by aligned labels while $z$ is the difference of separately sorted sequences; by the Hardy--Littlewood--Pólya rearrangement inequality the same-order pairing minimizes squared distance, so $\lVert z\rVert_2\le\lVert a\rVert_2\le\sqrt2q^*$. For equality take $\lambda=(2,1^{n-2})$, $\mu=(1^n)$: the graphs differ by one edge so $q^*=1$, and $\lVert z\rVert_2=\sqrt2$.
\end{proof}

\begin{lemma}\label{lem:4.6}
$L^-(n)\ge\frac2{3\sqrt n}$, giving (C), and $\rho(n)\le\frac{3\sqrt2}2\sqrt n$, giving (D).
\end{lemma}

\begin{proof}
Theorem~\ref{thm:3.1}(B) gives $q^*<\frac32\lVert z\rVert_1$. Pairing $|z|$ with the all-ones vector by Cauchy--Schwarz over $n$ coordinates,
\begin{equation*}
\lVert z\rVert_1\le\sqrt n\,\lVert z\rVert_2 .
\end{equation*}
This is the usable direction: the reverse inequality $\lVert z\rVert_2\le\lVert z\rVert_1$ places $\lVert z\rVert_2$ on the small side and cannot be composed with the first bound. Composing, $q^*<\frac32\sqrt n\lVert z\rVert_2$, so $\lVert z\rVert_2/q^*>\frac2{3\sqrt n}$ for each pair, and the minimum over the finitely many pairs gives (C), strictly. Dividing Lemma~\ref{lem:4.5} by this yields (D).
\end{proof}

The matching family is the adjacent uniform pair of Theorem~\ref{thm:3.1}(D), padded with isolated vertices so as to exist for every $n$. Padding is harmless, but not for a trivial reason: by Theorem~\ref{thm:8.1} the quantity $\max_X\lVert X\rVert_F^2$ is itself strongly NP-hard to compute, so one cannot argue that isolated vertices are irrelevant because they carry no edges.

\begin{lemma}[padding invariance]\label{lem:4.7}
Let $\lambda_0,\mu_0\vdash n_0$, let $r\ge0$, and put $\lambda=(\lambda_0,1^r)$, $\mu=(\mu_0,1^r)$. Then $q^*(\lambda,\mu)=q^*(\lambda_0,\mu_0)$ and $\lVert z(\lambda,\mu)\rVert_2=\lVert z(\lambda_0,\mu_0)\rVert_2$.
\end{lemma}

\begin{proof}
Isolated vertices have degree $0$, so each side gains $r$ zero coordinates, which after sorting occupy the same tail; $z$ merely gains $r$ zeros. For the metric, apply Lemma~\ref{lem:3.2} with $a=1$: removing one common unit block decreases $\max_X\lVert X\rVert_F^2$ by $1$, while $\lVert\lambda\rVert_2^2$ and $\lVert\mu\rVert_2^2$ each decrease by $1$, so half their sum decreases by $1$ as well. By Theorem~\ref{thm:2.3} the two decrements cancel; iterate over the $r$ common unit blocks.
\end{proof}

\begin{proof}[Proof of (E)]
Let $b=b(n)\ge2$, put $n_0=b(b+1)$ and $r=n-n_0\ge0$, and take $\lambda=((b+1)^{\,b},1^{\,r})$, $\mu=(b^{\,b+1},1^{\,r})$. The core pair is the family of Theorem~\ref{thm:3.1}(D) at parameter $b$ on $n_0$ vertices, both $b\ge2$ and $b(b+1)\mid n_0$ holding, so Lemma~\ref{lem:3.7} gives $q^*(\lambda_0,\mu_0)=\frac{b(3b-1)}2$. Every vertex of $G_{\lambda_0}$ has degree $b$ and every vertex of $G_{\mu_0}$ degree $b-1$, so $z=(1^{n_0})$ and $\lVert z\rVert_2=\sqrt{b(b+1)}$. By Lemma~\ref{lem:4.7} neither quantity changes under padding, so
\begin{equation*}
L^-(n)\le\frac{\lVert z\rVert_2}{q^*}=\frac{2\sqrt{b(b+1)}}{b(3b-1)},\qquad \rho(n)=\frac{\sqrt2}{L^-(n)}\ \ge\ \frac{\sqrt2\,b(3b-1)}{2\sqrt{b(b+1)}} ,
\end{equation*}
the numerator being the exact value $L^+(n)=\sqrt2$ of Lemma~\ref{lem:4.5} rather than an estimate.
\end{proof}

\begin{proof}[Proof of (F)]
By definition $b(b+1)\le n<(b+1)(b+2)$, so $b<\sqrt n$ from the left and $n<b^2+3b+2<(b+\frac32)^2$ from the right, giving $\sqrt n-\frac32<b<\sqrt n$. From $\sqrt{1+1/b}\le1+\frac1{2b}$,
\begin{equation*}
\frac{b(3b-1)}{\sqrt{b(b+1)}}=\frac{3b-1}{\sqrt{1+1/b}}\ \ge\ \frac{(3b-1)2b}{2b+1}=3b-\frac{5b}{2b+1}\ \ge\ 3b-\frac52 ,
\end{equation*}
so (E) gives $\rho(n)\ge\frac{3\sqrt2}2(b-\frac56)>\frac{3\sqrt2}2(\sqrt n-\frac73)$. With (D) this is (F), of relative error $O(n^{-1/2})$; and $L^-=\sqrt2/\rho$ gives the last claim.
\end{proof}

\textbf{Why the leading constant is pinned.} Three inequalities saturate on the same family: the routing bound of Lemma~\ref{lem:3.5}, the term-by-term comparison of Lemma~\ref{lem:3.6}, and the Cauchy--Schwarz step of Lemma~\ref{lem:4.6}. On the core pair $z=(1^{n_0})$ is an all-ones vector in $n_0$ coordinates, so Cauchy--Schwarz is exact in dimension $n_0$; it was applied in dimension $n$, and after padding $\lVert z\rVert_1/(\sqrt n\lVert z\rVert_2)=\sqrt{1-r/n}<1$ for $r>0$. So all three are simultaneously exact only on the subsequence $n=b(b+1)$, and for general $n$ the third saturates at rate $1-O(n^{-1/2})$, since $r=O(\sqrt n)$; this is why (F) is a two-sided sandwich rather than an identity. The two ends are witnessed by different families: $L^+=\sqrt2$ exactly, by a single edge against the empty graph, and $L^-$ asymptotically, by the padded adjacent uniform family.

\begin{corollary}\label{cor:4.8}
$\dfrac{\rho(\Phi_n)}{c_2(\mathcal K_n)}=\dfrac{\Theta(\sqrt n)}{\Theta(n^{1/4})}=\Theta\bigl(n^{1/4}\bigr)$.
\end{corollary}

\begin{proof}
Both are product distortions, so they are directly comparable; the numerator is Proposition~\ref{prop:4.4} and the denominator Theorem~\ref{thm:4.1}.
\end{proof}

This measures the gap for one coordinate; it does not improve $\Phi_n$, and it produces no map attaining $c_2$. The source of the loss is visible in the proof of Lemma~\ref{lem:4.6}: the passage from $\lVert z\rVert_1$ to $\lVert z\rVert_2$ over $n$ coordinates, saturated by a difference vector that is constant on its support. Section 5 replaces the single sorted-degree coordinate by a family of multiscale coordinates for which that passage is no longer the bottleneck.

\section{The critical dyadic embedding}\label{sec:5}
The lower bound of Theorem~\ref{thm:4.1} constrains every embedding, and the most natural named coordinate is a factor $\Theta(n^{1/4})$ away from it (Corollary~\ref{cor:4.8}). This section constructs a coordinate that is not. The construction is a one-parameter family, introduced first and in full so that no member of it appears later as an improvised repair; the map $F_n$ of Theorem~\ref{thm:1.1} is the member $\gamma=\frac14$, and the exponent is forced by two failures that \S\ref{sec:5.3} exhibits.
\subsection{The weighted dyadic family and the critical member}\label{sec:5.1}
Keep $u_\lambda$ and $v=v_{\lambda,\mu}=u_\lambda-u_\mu$ from \S\ref{sec:2.3}. Let $J=\lceil\log_2n\rceil$ and $N=2^J$, so $n\le N<2n$ and $J\le\log_2n+1$. (The letter $L$ is reserved for input length, as in \S\ref{sec:2.4}.) Regard
\begin{equation*}
u_\lambda(s)\qquad(1\le s\le N)
\end{equation*}
as a function on $[N]=\{1,\dots,N\}$, extended by zero for $s>\lambda_1$, which is well defined since $\lambda_1\le n\le N$. Let
\begin{equation*}
\mathcal D_N=\bigl\{\,\{(j-1)2^\ell+1,\dots,j2^\ell\}\ :\ 0\le\ell\le J,\ 1\le j\le N/2^\ell\,\bigr\}
\end{equation*}
be the \textbf{dyadic intervals} of $[N]$, write $\ell(I)$ for the level of $I$, and note $|\mathcal D_N|=\sum_{\ell=0}^JN/2^\ell=2N-1$. The only property of this family that any proof below uses is that the intervals of each level $\ell$ partition $[N]$. For $\gamma\ge0$ define
\begin{equation*}
F^{(\gamma)}(\lambda)=\Bigl(2^{-\gamma\ell(I)}\sum_{s\in I}u_\lambda(s)\Bigr)_{I\in\mathcal D_N}\ \in\ \mathbb R^{2N-1},
\end{equation*}
a diagonal rescaling of the plain dyadic-sum map $F^{(0)}$. Put
\begin{equation*}
E_\ell(w)=\sum_{I\in\mathcal D_N,\ \ell(I)=\ell}\Bigl(\sum_{s\in I}w(s)\Bigr)^2\qquad(w\in\mathbb R^N),
\end{equation*}
the \textbf{level-$\ell$ energy}, so that by linearity of $u\mapsto F^{(\gamma)}$,
\begin{equation}\label{eq:5.1}
\bigl\lVert F^{(\gamma)}(\lambda)-F^{(\gamma)}(\mu)\bigr\rVert_2^2=\sum_{\ell=0}^J2^{-2\gamma\ell}E_\ell(v) .
\end{equation}
Finally, partition $[N]$ into the maximal constant runs of $v$ and let $k_0(v)$ be the number of runs on which $v$ is non-zero. Figure~\ref{fig:staircase} carries one pair through the whole of this construction, from the two staircases to the level energies and the weights \eqref{eq:5.1} applies to them.

\begin{figure}[tb]
\centering
\includegraphics[width=\textwidth]{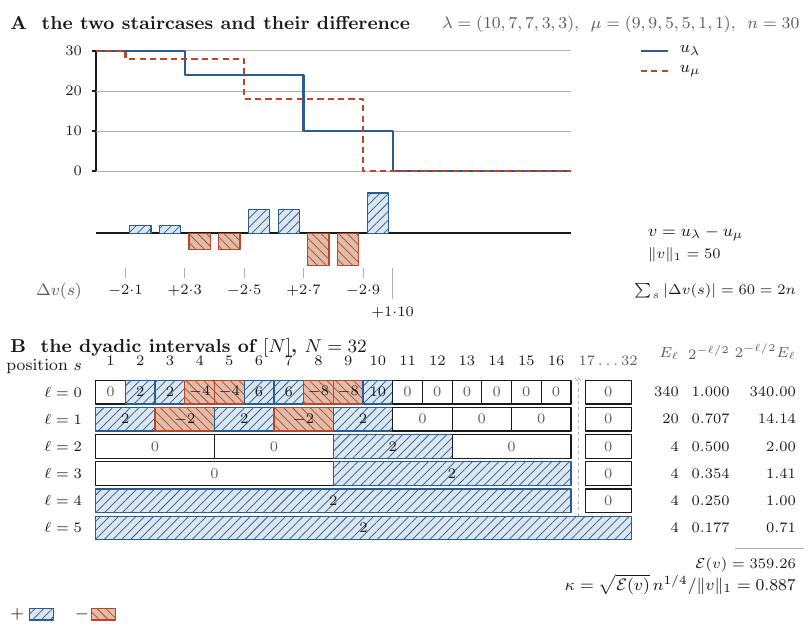}
\caption{\textbf{From the Ferrers staircase to the weighted dyadic
coordinates.} The running example is $\lambda=(10,7,7,3,3)$ against
$\mu=(9,9,5,5,1,1)$ at $n=30$, the extremal chirp pair of
\S\ref{sec:7}. \textbf{A.} The staircases
$u_\lambda(s)=\sum_{i:\lambda_i\ge s}\lambda_i$ and $u_\mu$, drawn on one
vertical scale with their difference $v=u_\lambda-u_\mu$ beneath, so that the
gap between the two steps is the height of the bar. Each jump is printed as a
signed product: every interior jump is $2s$, twice the minimum
Lemma~\ref{lem:2.6} permits at $s$, while the terminal jump at $s=m=10$ is
$1\cdot s$, the minimum itself. Their total, $60=2n$, saturates the lemma's
budget exactly. \textbf{B.} The dyadic intervals of $[N]$, each carrying its
coordinate $\sum_{s\in I}v(s)$; positions $17,\dots,32$ are folded away, $v$
being zero there. The ledger weights the level energies $E_\ell$ by
$2^{-\ell/2}$, which is the weight \eqref{eq:5.1} applies at $\gamma=\frac14$,
and totals them to the critical energy
$\mathcal E(v)=\sum_\ell2^{-\ell/2}E_\ell(v)=359.26$ and to the normalized
ratio $\kappa=\sqrt{\mathcal E(v)}\,n^{1/4}/\lVert v\rVert_1=0.887$ of
\S\ref{sec:7}, the record low of Appendix~\ref{sec:B}.
Legend as in Fig.~\ref{fig:transportation}.}
\label{fig:staircase}
\end{figure}

\begin{proposition}[the family is explicit and cheap]\label{prop:5.1}
For every $n\ge2$ and every $\gamma\ge0$:

\textbf{(A)} $F^{(\gamma)}$ depends only on $n$ and $\gamma$, acts on a single partition, calls no distance oracle, and is injective on $\mathcal K_n/\!\cong$.

\textbf{(B)} The dimension is $2N-1<4n$. At $\gamma=0$ each coordinate is an integer in $[0,2n^2)$, of bit width $O(\log n)$, so the image occupies $O(n\log n)$ bits and $\Theta(n\log n)$ in the worst case under a dense fixed-width encoding.

\textbf{(C)} Given the input in either partition representation (P1) or (P2), $F^{(\gamma)}(\lambda)$ is computable in $O(n)$ arithmetic operations and $O(n)$ machine words; given it in representation (G), the cost is $\Theta(n^2)$, since reading the input alone is $\Theta(n^2)$ bits.
\end{proposition}

\begin{proof}
Uniformity and pointwise action are immediate. For injectivity, level $0$ consists of the singletons and carries weight $2^0=1$, so $F^{(\gamma)}(\lambda)$ returns all the values $u_\lambda(s)$; since $u_\lambda(s)-u_\lambda(s+1)=s\,m_s(\lambda)$ with $u_\lambda(N+1)=0$, every multiplicity $m_s(\lambda)$ is recovered. For (B), $|\mathcal D_N|=2N-1<4n$ and $0\le\sum_{s\in I}u_\lambda(s)\le N\max_su_\lambda(s)=Nn<2n^2$. For (C), computing $u_\lambda$ from the multiplicities is one suffix sum, each level is formed bottom-up by adding pairs from the level below, and the weights are applied coordinatewise, for total work $\sum_\ell N/2^\ell=O(n)$; in representation (G) the block sizes are recoverable within the single $\Theta(n^2)$-bit scan.
\end{proof}

\emph{The arithmetic model in (C).} The count is of arithmetic operations on integers of $O(\log n)$ bits, followed by one multiplication per coordinate by a constant of the field $\mathbb Q(2^{1/4})$; it is not a count of real-number operations. All of the work is done at $\gamma=0$, where every intermediate quantity is an integer below $2n^2$; the weights enter only in the final coordinatewise pass, and at $\gamma=\frac14$ the weight $2^{-\ell/4}$ is $2^{-\lfloor\ell/4\rfloor}\cdot2^{-(\ell\bmod4)/4}$, so the four values $2^{-j/4}$, $0\le j\le3$, are the only irrational constants involved and each coordinate is exactly an integer multiple of $2^{-j/4}$ for a known $j$. The output is therefore exact in $O(n)$ words; rounding it to floating point is a separate step that we do not need, since $\lVert F^{(\gamma)}(\lambda)-F^{(\gamma)}(\mu)\rVert_2^2$ is the rational combination \eqref{eq:5.1} of integer level energies.

The first thing the weight buys is that the upper Lipschitz constant stops growing.

\begin{proposition}[constant upper Lipschitz for every $\gamma>0$]\label{prop:5.2}
For every $\gamma>0$ and all $\lambda\ne\mu\vdash n$,
\begin{equation*}
\bigl\lVert F^{(\gamma)}(\lambda)-F^{(\gamma)}(\mu)\bigr\rVert_2\ \le\ \bigl(1-2^{-2\gamma}\bigr)^{-1/2}\lVert v\rVert_1\ \le\ 2\bigl(1-2^{-2\gamma}\bigr)^{-1/2}q^*(\lambda,\mu) .
\end{equation*}
At $\gamma=\frac14$ this reads $\lVert F^{(1/4)}(\lambda)-F^{(1/4)}(\mu)\rVert_2\le\sqrt{2+\sqrt2}\,\lVert v\rVert_1<1.848\,\lVert v\rVert_1$, so $L^+\bigl(F^{(1/4)}\bigr)\le2\sqrt{2+\sqrt2}<3.696$.
\end{proposition}

\begin{proof}
At each level the intervals partition $[N]$, so $\sum_{\ell(I)=\ell}\bigl|\sum_{s\in I}v(s)\bigr|\le\lVert v\rVert_1$ and hence $E_\ell(v)\le\lVert v\rVert_1^2$. Summing the geometric series in \eqref{eq:5.1} gives the first inequality; Lemma~\ref{lem:2.5} gives the second.
\end{proof}

By contrast the unweighted member has $L^+$ of order $\sqrt{J+1}$, and \S\ref{sec:5.2} shows that this is not an artefact of the estimate.

\subsection{The unweighted member is pinned}\label{sec:5.2}

Two facts about partitions carry the upper bound for $F^{(0)}$: the difference of two Ferrers staircases has few pieces, and a dyadic interval can capture a constant fraction of any one of them.

\begin{lemma}[breakpoints and runs]\label{lem:5.3}
For every $\lambda\vdash n$ the number of distinct part sizes satisfies $r(\lambda)\le\frac{\sqrt{8n+1}-1}2=O(\sqrt n)$, and the set of interior breakpoints of $u_\lambda$ on $[N]$, the $s<N$ with $u_\lambda(s)\ne u_\lambda(s+1)$, is contained in the set of distinct part sizes of $\lambda$ and has cardinality at most $r(\lambda)$. Consequently, for all $\lambda,\mu\vdash n$,
\begin{equation*}
k_0(v)\ \le\ r(\lambda)+r(\mu)+1\ \le\ \sqrt{8n+1} .
\end{equation*}
\end{lemma}

\begin{proof}
Let the distinct part sizes of $\lambda$ be $s_1<\dots<s_r$, $r=r(\lambda)$. All are positive integers, so $s_j\ge j$ and $n=\sum_i\lambda_i\ge\sum_js_j\ge\frac{r(r+1)}2$, giving the bound on $r$. Since $u_\lambda(s)-u_\lambda(s+1)=\sum_{i:\lambda_i=s}\lambda_i=s\,m_s(\lambda)$ and $s\ge1$, this difference vanishes iff $m_s(\lambda)=0$; hence the interior breakpoint set is contained in the set of distinct part sizes. Containment rather than equality, since if $\lambda_1=n=N$ then $s=N$ is a change of the zero extension but is not interior to $[N]$; in particular the zero-extended range contains no interior breakpoint. If $v(s)\ne v(s+1)$ then $u_\lambda$ or $u_\mu$ has a breakpoint at $s$, so the breakpoint set of $v$ has size at most $r(\lambda)+r(\mu)$ and the number of maximal constant runs is at most one more. Three degeneracies leave this intact: the range above $\max(\lambda_1,\mu_1)$ is a single constant run of value zero, absorbed by the $+1$ and excluded from $k_0$; merging of adjacent equal runs is implied by maximality and only decreases the count; and $v(1)=n-n=0$, so there is a zero run at the left end too.
\end{proof}

\begin{lemma}[dyadic capture]\label{lem:5.4}
Every integer interval $R\subseteq[N]$ of length $W\ge1$ contains some $I\in\mathcal D_N$ with $|I|>W/4$.
\end{lemma}

\begin{proof}
Take $\ell=\max(0,\lfloor\log_2W\rfloor-1)$. If $W=1$ then $R$ is a point, which lies at level $0$, and $1>\frac14$. If $W\ge2$ then
\begin{equation*}
2^\ell\le\frac W2,\qquad 2^{\ell+1}=2^{\lfloor\log_2W\rfloor}>\frac W2\ \Longrightarrow\ 2^\ell>\frac W4 ,
\end{equation*}
and $\ell\le\log_2W-1\le J-1$, so level $\ell$ exists. Writing $R=\{a,\dots,a+W-1\}$ and letting $a'\ge a$ be the least level-$\ell$ left endpoint, we have $a'-a\le2^\ell-1$ and $a'+2^\ell-1\le a+2^{\ell+1}-2<a+W-1$, so that interval lies inside $R$.
\end{proof}

\begin{theorem}[upper bound for the unweighted map]\label{thm:5.5}
For all $\lambda\ne\mu\vdash n$:

\textbf{(A)} $\displaystyle\frac{\lVert v\rVert_1}{4\sqrt{k_0(v)}}\ \le\ \bigl\lVert F^{(0)}(\lambda)-F^{(0)}(\mu)\bigr\rVert_2\ \le\ \sqrt{J+1}\,\lVert v\rVert_1 .$

\textbf{(B)} With $k_{\max}(n)=\max_{\lambda\ne\mu\vdash n}k_0(v)$,
\begin{equation*}
\rho\bigl(F^{(0)}\bigr)\ \le\ 12\,\sqrt{J+1}\,\sqrt{k_{\max}(n)}\ \le\ 12\,(8n+1)^{1/4}\sqrt{\log_2n+2}\ =\ O\bigl(n^{1/4}\sqrt{\log n}\bigr).
\end{equation*}

\textbf{(C)} If the constant weights of $F^{(0)}$ are replaced by arbitrary positive weights $(\alpha_\ell)_{\ell=0}^J$, running the two steps of (A) verbatim yields a guarantee containing the factor $\bigl(\sum_\ell\alpha_\ell^2\bigr)^{1/2}/\min_\ell\alpha_\ell\ge\sqrt{J+1}$, with equality iff $\alpha$ is constant.
\end{theorem}

\begin{proof}
\textbf{Upper estimate in (A).} Fix a level $\ell$. Its intervals partition $[N]$, so $\sum_j|\langle v,\mathbf 1_{I_{\ell,j}}\rangle|\le\lVert v\rVert_1$, and since $\sum_jx_j^2\le(\sum_jx_j)^2$ for non-negative $x_j$, that level contributes at most $\lVert v\rVert_1^2$ to the squared norm. Summing over $\ell=0,\dots,J$ gives the bound; each coordinate $s$ lies in one interval per level and is counted $J+1$ times, and signed cancellation only decreases $|\langle v,\mathbf 1_I\rangle|$.

\textbf{Lower estimate in (A).} Let $R_1,\dots,R_{k_0}$ be the runs on which $v$ is non-zero, $v\equiv a_j\ne0$ on $R_j$ with $W_j=|R_j|$. They are disjoint and zero runs contribute nothing, so $\lVert v\rVert_1=\sum_j|a_j|W_j$. For each $j$ take $I_j\subseteq R_j$ from Lemma~\ref{lem:5.4} with $|I_j|>W_j/4$; then $|\langle v,\mathbf 1_{I_j}\rangle|=|a_j||I_j|>\frac{|a_j|W_j}4$. The $I_j$ are $k_0$ distinct coordinates, their distinctness following from the disjointness of the runs alone. Keeping only these and using $(\sum_jx_j)^2\le k_0\sum_jx_j^2$,
\begin{equation*}
\bigl\lVert F^{(0)}(\lambda)-F^{(0)}(\mu)\bigr\rVert_2^2\ \ge\ \sum_j\Bigl(\frac{|a_j|W_j}4\Bigr)^2\ \ge\ \frac1{k_0}\Bigl(\sum_j\frac{|a_j|W_j}4\Bigr)^2=\frac{\lVert v\rVert_1^2}{16\,k_0} .
\end{equation*}

\textbf{(B).} Lemma~\ref{lem:2.5} gives $\lVert v\rVert_1\le2q^*$ and Theorem~\ref{thm:3.1}(B) gives $\lVert v\rVert_1>\frac23q^*$; substituting into (A),
\begin{equation*}
L^+\bigl(F^{(0)}\bigr)\le2\sqrt{J+1},\qquad \frac1{L^-(F^{(0)})}\le6\sqrt{k_{\max}(n)} ,
\end{equation*}
and multiplying gives the first inequality; Lemma~\ref{lem:5.3} and $J+1\le\log_2n+2$ give the rest.

\textbf{(C).} Rerun the two estimates: the upper end becomes $(\sum_\ell\alpha_\ell^2)^{1/2}\lVert v\rVert_1$ and the lower end $\frac{\min_\ell\alpha_\ell}{4\sqrt{k_0}}\lVert v\rVert_1$, and
\begin{equation*}
\frac{(\sum_{\ell=0}^J\alpha_\ell^2)^{1/2}}{\min_\ell\alpha_\ell}\ \ge\ \frac{((J+1)\min_\ell\alpha_\ell^2)^{1/2}}{\min_\ell\alpha_\ell}=\sqrt{J+1},
\end{equation*}
with equality iff all $\alpha_\ell$ coincide.
\end{proof}

Part (C) says that no reweighting can improve the \emph{guarantee that this argument yields}; \S\ref{sec:6} obtains a better bound for $\gamma=\frac14$ by an argument of a different kind, namely one that uses the arithmetic of realizable staircases rather than the dyadic system alone, and \S\ref{sec:6.1} explains why there is no contradiction. First we show that for $\gamma=0$ the bound in (B) is the truth, not merely a guarantee: the two Lipschitz constants are saturated by two families that are present in $\mathcal K_n/\!\cong$ simultaneously.

\begin{lemma}[spike pair: the expansion witness]\label{lem:5.6}
Let $m=\lfloor n/2\rfloor$ and take $\lambda^{\mathrm{sp}}=(m,1^{\,n-m})$, $\mu^{\mathrm{sp}}=(m-1,1^{\,n-m+1})$. Then $w:=v_{\lambda^{\mathrm{sp}},\mu^{\mathrm{sp}}}$ satisfies $w(s)=1$ for $2\le s\le m-1$, $w(m)=m$, $\lVert w\rVert_1=2m-2$, and
\begin{equation*}
\frac{\bigl\lVert F^{(0)}(\lambda^{\mathrm{sp}})-F^{(0)}(\mu^{\mathrm{sp}})\bigr\rVert_2}{q^*(\lambda^{\mathrm{sp}},\mu^{\mathrm{sp}})}\ \ge\ \frac{\sqrt{J+1}\;m}{\tfrac32(2m-2)}\ >\ \frac{\sqrt{J+1}}3 .
\end{equation*}
\end{lemma}

\begin{proof}
The formula for $w$ is a direct computation of the two staircases: for $2\le s\le m-1$ the large block contributes $m$ on one side and $m-1$ on the other, and at $s=m$ only $\lambda^{\mathrm{sp}}$ contributes. At each level exactly one dyadic interval contains the position $m$, and its coefficient is $m$ plus a non-negative contribution from the plateau; hence $\lVert F^{(0)}w\rVert_2^2\ge(J+1)m^2$. Conclude with $q^*<\frac32\lVert w\rVert_1$ (Theorem~\ref{thm:3.1}(B)).
\end{proof}

The spike is \emph{coherent across all $J+1$ scales}: a single tall, narrow feature is seen, at full height, by one interval at every level, so the unweighted functional charges it $\sqrt{J+1}$ times what the $\ell_1$ mass warrants. The opposite extreme is a feature that cancels above the finest scale.

\begin{lemma}[breathing pairs: the contraction witness]\label{lem:5.7}
Let $n\ge2^{10}$. Let $h_0$ be the least even integer $\ge2\sqrt n$, let $k=\lfloor\sqrt n/16\rfloor\ge2$, and put $b_i=h_0+4i$ for $1\le i\le k$, all even. Define
\begin{equation*}
\lambda^{\mathrm{br}}=\bigl\{b_i,\ b_i+2:1\le i\le k\bigr\}\cup\{1^{\,f}\},\qquad \mu^{\mathrm{br}}=\bigl\{b_i+1,\ b_i+1:1\le i\le k\bigr\}\cup\{1^{\,f}\},
\end{equation*}
with $f=n-\sum_i(2b_i+2)$ common unit blocks. Then both are partitions of $n$, and with $v=v_{\lambda^{\mathrm{br}},\mu^{\mathrm{br}}}$,
\begin{equation*}
v=\sum_{i=1}^k\Bigl(-b_i\,\mathbf 1_{\{b_i+1\}}+(b_i+2)\,\mathbf 1_{\{b_i+2\}}\Bigr),\qquad \lVert v\rVert_1=\sum_i(2b_i+2)\in\Bigl[\frac n8,\frac n2\Bigr],
\end{equation*}
\begin{equation*}
E_0(v)\in\Bigl[\frac{n^{3/2}}4,\ n^{3/2}\Bigr],\qquad \sum_{\ell=1}^JE_\ell(v)\le12k^2J\le\frac3{64}\,nJ ,
\end{equation*}
and consequently
\begin{equation*}
\frac{q^*(\lambda^{\mathrm{br}},\mu^{\mathrm{br}})}{\bigl\lVert F^{(0)}(\lambda^{\mathrm{br}})-F^{(0)}(\mu^{\mathrm{br}})\bigr\rVert_2}\ \ge\ \frac{n^{1/4}}{18} .
\end{equation*}
\end{lemma}

\begin{proof}
\emph{Feasibility.} Since $k\le\sqrt n/16$ and $h_0\le2\sqrt n+2$,
\begin{equation*}
\sum_i(2b_i+2)=2kh_0+4k(k+1)+2k\ \le\ \frac n4+\frac{\sqrt n}4+\frac n{64}+\frac{\sqrt n}4+\frac{\sqrt n}8\ <\ \frac n2
\end{equation*}
for $n\ge2^{10}$, so $f>0$; and the largest part is $b_k+2\le2\sqrt n+\frac{\sqrt n}4+4<N$.

\emph{Dipole form.} $u_\lambda$ is additive over the multiset of parts, and the $f$ unit blocks contribute identically on both sides, so $v$ is the sum over $i$ of the difference produced by the single exchange $\{b,b+2\}\to\{b+1,b+1\}$ with $b=b_i$. For that exchange both sides contribute $2b+2$ for $s\le b$; at $s=b+1$ the left contributes $b+2$ and the right $2b+2$; at $s=b+2$ only the left contributes $b+2$; nothing above. This is the displayed dipole, and the $b_i$ being spaced by $4$ the supports are disjoint, whence $\lVert v\rVert_1=\sum_i(b_i+(b_i+2))$. The upper bound was just shown; for the lower, $\lVert v\rVert_1\ge2kh_0\ge2\bigl(\frac{\sqrt n}{16}-1\bigr)2\sqrt n=\frac n4-4\sqrt n\ge\frac n8$ for $n\ge2^{10}$.

\emph{Level $0$.} $E_0(v)=\sum_i\bigl(b_i^2+(b_i+2)^2\bigr)$. Below, $E_0\ge2kh_0^2\ge2\bigl(\frac{\sqrt n}{16}-1\bigr)4n\ge\frac{n^{3/2}}2-8n\ge\frac{n^{3/2}}4$. Above, $E_0\le2k(b_k+2)^2\le\frac{\sqrt n}8\bigl(2.25\sqrt n+4\bigr)^2\le n^{3/2}$ for $n\ge2^{10}$.

\emph{Levels $\ell\ge1$.} Each $b_i$ is even, so the dipole support $\{b_i+1,b_i+2\}$ is itself a level-$1$ dyadic interval and therefore lies inside exactly one interval at every level $\ell\ge1$: no dipole straddles a boundary. Its contribution to the coefficient of the containing interval is $-b_i+(b_i+2)=2$. Hence at level $\ell$ the coefficient of $I$ is $2t_I$ with $t_I=\#\{i:b_i+1\in I\}$, and $E_\ell=4\sum_It_I^2$. The points $b_i+1$ are spaced exactly $4$ apart, so $t_I\le\min(k,2^{\ell-2}+1)$, while $\sum_It_I=k$; hence $E_\ell\le4k\min(k,2^{\ell-2}+1)$. Splitting the sum at the least $\ell^*$ with $2^{\ell^*-2}\ge k$, and using $2^{\ell^*-2}<2k$,
\begin{equation*}
\sum_{\ell=1}^JE_\ell\ \le\ 4k\Bigl[\sum_{\ell<\ell^*}\bigl(2^{\ell-2}+1\bigr)+\sum_{\ell\ge\ell^*}k\Bigr]\ \le\ 4k\bigl[2k+J+kJ\bigr]\ \le\ 12k^2J
\end{equation*}
for $k\ge2$ and $J\ge10$; finally $k^2\le n/256$.

\emph{Conclusion.} By Lemma~\ref{lem:2.5}, $q^*\ge\frac12\lVert v\rVert_1\ge\frac n{16}$, while by \eqref{eq:5.1} at $\gamma=0$,
\begin{equation*}
\bigl\lVert F^{(0)}v\bigr\rVert_2\le\sqrt{n^{3/2}+\tfrac3{64}nJ}\ \le\ n^{3/4}\sqrt{1+\tfrac{3J}{64\sqrt n}}\ \le\ 1.008\,n^{3/4}
\end{equation*}
for $n\ge2^{10}$, the correction being decreasing in $n$. The quotient is at least $n^{1/4}/(16\cdot1.008)>n^{1/4}/18$.
\end{proof}

A breathing train is the exact opposite of a spike: each dipole has height $\Theta(\sqrt n)$ and width $2$ but \textbf{net sum $2$}, so it is nearly mean-free and is therefore invisible to every level above the first. All of its energy sits at level $0$, where the unweighted functional charges it far less than its $\ell_1$ mass. The two witnesses combine.

\begin{theorem}[the distortion of the unweighted map is pinned]\label{thm:5.8}
For $n\ge2^{10}$,
\begin{equation*}
\rho\bigl(F^{(0)}\bigr)\ \ge\ \frac{\sqrt{J+1}}3\cdot\frac{n^{1/4}}{18}\ \ge\ \frac{n^{1/4}\sqrt{\log_2n}}{54} ,
\end{equation*}
and hence, with Theorem~\ref{thm:5.5}(B),
\begin{equation*}
\rho\bigl(F^{(0)}\bigr)=\Theta\bigl(n^{1/4}\sqrt{\log n}\bigr).
\end{equation*}
\end{theorem}

\begin{proof}
$\rho=L^+/L^-$, and $L^+$ is at least the ratio realized by the spike pair (Lemma~\ref{lem:5.6}) while $1/L^-$ is at least the reciprocal ratio realized by the breathing pair (Lemma~\ref{lem:5.7}). Multiply.
\end{proof}

\begin{figure}[tb]
\centering
\includegraphics[width=\textwidth]{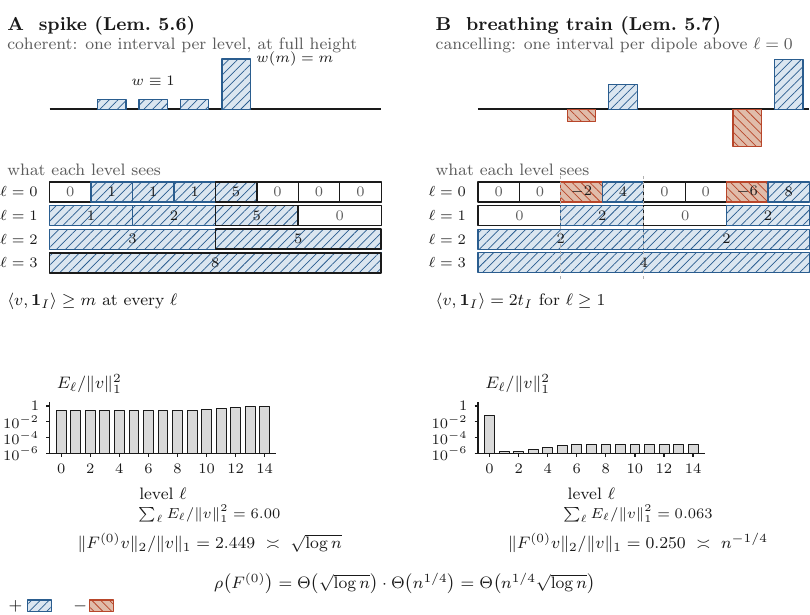}
\caption{\textbf{Two witnesses, one conflict of levels.}
\textbf{A.} The spike pair of Lemma~\ref{lem:5.6} at $n=10$, $m=5$. Exactly
one dyadic interval per level contains the spike, and every one of them sees
it at full height (heavy outline), so no level is negligible.
\textbf{B.} A dipole train of the kind used in Lemma~\ref{lem:5.7}, drawn with
two dipoles. At level $0$ the two halves of a dipole are counted apart; at
every coarser level a dipole lies inside a single interval (dashed guides) and
contributes only its net~$2$. Both mini-instances are exact difference vectors
of partition pairs, not sketches.
Below, the level energies of the two witness families at $n=2^{14}$, each
divided by $\lVert v\rVert_1^2$ so that the two panels share one ruler: the
spike holds every level within a factor $4$ of its largest, the breathing
train puts $99.8\%$ at level~$0$. The first ratio grows like $\sqrt{\log n}$
and the second decays like $n^{-1/4}$; their quotient is the distortion of
$F^{(0)}$ (Theorem~\ref{thm:5.8}). Legend as in Fig.~\ref{fig:transportation}.}
\label{fig:levelconflict}
\end{figure}

So no sharper analysis of $F^{(0)}$ can reach the class-optimal order $n^{1/4}$: the loss is not in the estimates but in the map. The obstruction is conceptual, and it is a conflict of scales. A spike is coherent across all $J+1$ levels and an unweighted multiscale $\ell_2$ functional counts it $J+1$ times; a breathing train is incoherent above level $1$ and the same functional counts it once. Figure~\ref{fig:levelconflict} shows the two witnesses and their level spectra on a common ruler, the quotient of the two being exactly the distortion just computed. No uniform weight can price both correctly, and Theorem~\ref{thm:5.5}(C) says that within the two-step argument no choice of weights can either. Both statements point to the same repair: discount the coarse levels, and prove the lower bound by a method that uses more about $v$ than its $\ell_1$ norm.

\subsection{Why the critical exponent is $\gamma=\frac14$}\label{sec:5.3}
Discounting coarse levels cannot be carried too far, and the two flanks pin the exponent. On one side, Proposition~\ref{prop:5.2} already shows that any $\gamma>0$ removes the $\sqrt{\log n}$ from $L^+$, so the spike is repaired immediately. On the other side, over-discounting destroys the coarse levels that carry \emph{coherent} mass, and the witness is the most classical pair in the space.

\begin{proposition}[super-critical weights fail on the classical pair]\label{prop:5.9}
Let $\frac14<\gamma<\frac12$ and take $\lambda=(n)$, $\mu=(1^{\,n})$, so that $v=n\cdot\mathbf 1_{[2,n]}$. Then
\begin{equation*}
\bigl\lVert F^{(\gamma)}(\lambda)-F^{(\gamma)}(\mu)\bigr\rVert_2=\Theta_\gamma\bigl(n\cdot n^{1-\gamma}\bigr),\qquad \frac{q^*}{\lVert F^{(\gamma)}v\rVert_2}=\Theta_\gamma\bigl(n^{\gamma}\bigr) ,
\end{equation*}
and since the spike pair of Lemma~\ref{lem:5.6} gives $\lVert F^{(\gamma)}w\rVert_2/q^*\ge\frac13$ for every $\gamma\ge0$,
\begin{equation*}
\rho\bigl(F^{(\gamma)}\bigr)=\Omega_\gamma\bigl(n^{\gamma}\bigr)\ \gg\ n^{1/4}\qquad\Bigl(\gamma>\tfrac14\Bigr).
\end{equation*}
\end{proposition}

\begin{proof}
Write $T=n-1$ and $h=n$. For $2^\ell\le T$ the level-$\ell$ intervals lying entirely inside $[2,n]$ number $\Theta(T2^{-\ell})$ and each has coefficient $h2^\ell$, so $E_\ell=\Theta(h^2T2^\ell)$; the at most two boundary intervals per level add $O(h^24^\ell)$, which is absorbed for $\gamma<\frac12$. Hence
\begin{equation*}
\sum_\ell2^{-2\gamma\ell}E_\ell=\Theta\Bigl(h^2T\sum_{2^\ell\le T}2^{(1-2\gamma)\ell}\Bigr)=\Theta\bigl(h^2T^{2-2\gamma}\bigr)
\end{equation*}
for $\gamma<\frac12$. Since $\lVert v\rVert_1=hT$ and $q^*\ge\frac12\lVert v\rVert_1$ by Lemma~\ref{lem:2.5}, the displayed ratio is $\Theta(T^\gamma)=\Theta(n^\gamma)$. The expansion bound for the spike pair follows from the level-$0$ term of Lemma~\ref{lem:5.6} alone, which carries no weight.
\end{proof}

The witness is $K_n$ against the empty graph, and it is \emph{sign-coherent}: $v\ge0$ pointwise, all of its mass a single wide plateau. Such mass is invisible at level $0$ relative to its $\ell_1$ norm, the level-$0$-only functional having $E_0\sqrt n/\lVert v\rVert_1^2=\sqrt n/T\to0$ on it, and is seen only by the coarse levels, which a super-critical weight suppresses. Collecting what is proved:

\begin{itemize}
\item $\gamma=0$ over-prices scale-coherent spikes, and loses exactly $\sqrt{\log n}$ (Theorem~\ref{thm:5.8});
\item $\gamma>\frac14$ under-prices coherent plateaus, and loses $n^{\gamma-1/4}$ (Proposition~\ref{prop:5.9});
\item the level-$0$-only functional, the formal limit $\gamma\to\infty$, fails outright on plateaus.
\end{itemize}

Thus $\gamma=\frac14$ is the largest exponent not excluded by the plateau family. It is \emph{not} the smallest exponent at which $L^+$ is bounded: by Proposition~\ref{prop:5.2} every $\gamma>0$ has an absolute upper Lipschitz constant, so the lower flank excludes only $\gamma=0$ itself, and it is the plateau family, acting from above, that selects $\frac14$. We do not claim more than this. In particular we have no proof that $0<\gamma<\frac14$ is worse than $\gamma=\frac14$; the behaviour of the family on that range is not classified here, and we deliberately state no phase diagram.

\begin{remark}[a scaling heuristic for the exponent, not a proof]\label{rem:5.10}
The following order-of-magnitude computation explains why $\frac14$ is the natural balancing exponent for the two failure modes, and we record it as motivation only. Replace the width-$2$ dipoles of Lemma~\ref{lem:5.7} by dipoles of width $W$, a power of two: take blocks $\{b_i,b_i+2W\}$ against $\{b_i+W,b_i+W\}$ with $b_i\equiv0\bmod2W$ and spaced $4W$, so that $v$ takes the value $-b_i$ on an interval of width $W$ and $+(b_i+2W)$ on the next. The mass constraint $\sum_i2(b_i+W)\le n$ together with the packing requirement $b_i\gtrsim Wk$ caps the number of dipoles at $k\lesssim\sqrt{n/W}$, so that $k\sqrt W\lesssim\sqrt n$ \emph{independently of $W$}. Tracking the two surviving contributions to the ratio $\kappa^2=\lVert F^{(1/4)}v\rVert_2^2\sqrt n/\lVert v\rVert_1^2$ of \S\ref{sec:7} gives $\kappa^2\asymp(k\sqrt W)^{-1}+W^{3/2}k^{-1/2}b^{-2}$, bounded below uniformly in $W$. The exponent $\frac12$ in the energy weight $2^{-\ell/2}$ of \eqref{eq:5.1} at $\gamma=\frac14$ is the unique geometric weight under which thin and fat trains are priced alike; that invariance is the structural reason for criticality. Numerically, $\kappa$ over dipole trains of width $W\in\{2,4,\dots,256\}$ at $n=2^{16}$ lies in $[1.98,2.83]$ and is flat in $n$ at fixed $W$ (Appendix B). None of \S\ref{sec:6} depends on this remark, whose statements we do not claim as theorems.
\end{remark}

\subsection{The optimal upper bound}\label{sec:5.4}

Everything is now in place. Proposition~\ref{prop:5.2} bounds the upper Lipschitz constant of $F_n=F^{(1/4)}$ by an absolute constant; \S\ref{sec:5.2} shows that no unweighted or over-discounted member of the family can do better in order; and the lower Lipschitz constant is exactly the inverse-energy question, which Section 6 answers for the whole lattice of closed quantized sequences. We record the consequence here and prove the analytic input separately, so that Section 6 can be read on its own.

\begin{theorem}[the critical map has optimal order]\label{thm:5.11}
For all $n\ge2$,
\begin{equation*}
\rho\bigl(F_n\bigr)\ \le\ 1662\,n^{1/4}\ =\ O\bigl(n^{1/4}\bigr) .
\end{equation*}
\end{theorem}

\begin{proof}
Let $\lambda\ne\mu\vdash n$ and $v=v_{\lambda,\mu}$. By Lemma~\ref{lem:2.6}, $v$ is a closed quantized sequence of length $N$ with $\mathrm{TV}(v)\le2n$, so Theorem~\ref{thm:6.2} and \eqref{eq:5.1} give
\begin{equation*}
\bigl\lVert F_n(\lambda)-F_n(\mu)\bigr\rVert_2^2=\mathcal E(v)\ \ge\ \frac{\lVert v\rVert_1^{\,2}}{63504\sqrt{2n}},\qquad\text{that is}\qquad \bigl\lVert F_n(\lambda)-F_n(\mu)\bigr\rVert_2\ \ge\ \frac{\lVert v\rVert_1}{252\,(2n)^{1/4}} .
\end{equation*}
By Theorem~\ref{thm:3.1}, $q^*<\frac32\lVert v\rVert_1$, so $L^-(F_n)\ge\frac2{3\cdot252\,(2n)^{1/4}}$. By Proposition~\ref{prop:5.2}, $L^+(F_n)\le2\sqrt{2+\sqrt2}$. Multiplying,
\begin{equation*}
\rho(F_n)\ \le\ 3\sqrt{2+\sqrt2}\cdot252\cdot2^{1/4}\,n^{1/4}\ <\ 1662\,n^{1/4} . 
\qedhere
\end{equation*}
\end{proof}

\begin{corollary}[the class and the map are both pinned]\label{cor:5.12}
For all $n\ge2$,
\begin{equation*}
c_2(\mathcal K_n)=\Theta\bigl(n^{1/4}\bigr)\qquad\text{and}\qquad \rho\bigl(F_n\bigr)=\Theta\bigl(n^{1/4}\bigr) .
\end{equation*}
Thus $F_n$ is an explicit, deterministic, uniform embedding of $(\mathcal K_n/\!\cong,q^*)$ into $\ell_2^{\,2N-1}$ with $2N-1<4n$, computable in $O(n)$ time from a single partition, whose distortion attains the optimal order for the space.
\end{corollary}

\begin{proof}
The upper bound $c_2(\mathcal K_n)\le\rho(F_n)=O(n^{1/4})$ is Theorem~\ref{thm:5.11}, since $c_2$ is an infimum over embeddings. The lower bound $c_2(\mathcal K_n)=\Omega(n^{1/4})$ is Lemma~\ref{lem:4.2} with Enflo's theorem (\S\ref{sec:4.1}), and it applies to every embedding, in particular to $F_n$.
\end{proof}

This closes Theorem~\ref{thm:1.1} and, with \S\ref{sec:4.1}, Theorem~\ref{thm:4.1}. Three features of the resulting statement are worth separating, because they are usually not available together. The distortion order is optimal for the space, not merely the best proved for some family of maps. The map is \emph{pointwise}: it is evaluated on one partition without reference to the other, to the rest of the space, or to any distance oracle, which is what an embedding must be if it is to be used as a feature map. And it is cheap: linear dimension, linear time, and integer arithmetic of $O(\log n)$ bits at $\gamma=0$, the weights being applied once at the end.

The constant $1662$ is honest but crude, and Theorem~\ref{thm:6.2} is where all of it comes from; \S\ref{sec:6.7} accounts for it and \S\ref{sec:7} shows how far from the truth it is.

\section{A scale-free inverse theorem for quantized sequences}\label{sec:6}

Section 5 reduced the distortion of $F_n=F^{(1/4)}$ to a single inverse inequality: the critical energy $\mathcal E(v)$ of a realizable staircase difference must be bounded below by $\lVert v\rVert_1^2/\sqrt n$. This section proves it. The theorem obtained is strictly stronger than the application needs, in two ways. It holds for every closed integer sequence obeying the divisibility constraint $\Delta v(s)\in s\mathbb Z$, whether or not that sequence arises from a pair of partitions; and it is scale-free, normalized by the total variation of $v$ rather than by the worst-case budget $2n$ that Lemma~\ref{lem:2.6} supplies for realizable differences. Neither strengthening costs anything in the proof, and the first is what makes the statement reusable.

\subsection{Statement and basic geometry}\label{sec:6.1}

\begin{definition}\label{defn:6.1}
Let $N$ be a power of two. A \textbf{closed quantized sequence} of length $N$ is a map $v:\{1,\dots,N+1\}\to\mathbb Z$ with
\begin{equation*}
v(1)=v(N+1)=0,\qquad \Delta v(s):=v(s)-v(s+1)\in s\,\mathbb Z\quad(1\le s\le N) .
\end{equation*}
Write $\mathcal Q_N$ for the set of these. It is an additive subgroup of $\mathbb Z^{N+1}$, closed under integer scaling but not under scaling by arbitrary positive reals, so it is a lattice and not a cone. For $v\in\mathcal Q_N$ put
\begin{equation*}
L=L(v)=\lVert v\rVert_1=\sum_{s=1}^N|v(s)|,\qquad T=T(v)=\mathrm{TV}(v)=\sum_{s=1}^N\bigl|\Delta v(s)\bigr| .
\end{equation*}
The critical energy $\mathcal E(v)=\sum_{\ell=0}^J2^{-\ell/2}E_\ell(v)$, $J=\log_2N$, is as in \S\ref{sec:5.1}.
\end{definition}

By Lemma~\ref{lem:2.6} the difference vector $v_{\lambda,\mu}$ of any two partitions of $n$ lies in $\mathcal Q_N$ for $N$ the least power of two at least $n$, and satisfies $T\le2n$. The lattice $\mathcal Q_N$ is strictly larger: it imposes no realizability beyond the divisibility, and in particular does not require $v$ to be a difference of two monotone staircases.

\begin{theorem}[quantized inverse-energy theorem]\label{thm:6.2}
 For every power of two $N$ and every non-zero $v\in\mathcal Q_N$,
\begin{equation}\label{eq:6.1}
\mathcal E(v)\ \ge\ \frac1{63504}\cdot\frac{\lVert v\rVert_1^{\,2}}{\sqrt{\mathrm{TV}(v)}} .
\end{equation}
The constant is not optimized.
\end{theorem}

\textbf{Realizability is not a technical hypothesis.} No inequality of the form (6.1) can hold for arbitrary integer sequences. Take $N\ge2$ and $w=(0,-1,+1,-1,\dots)$ alternating on $[2,N]$. Every dyadic interval of length at least two that meets $[2,N-1]$ has sum $0$ or $\pm1$, so $E_0(w)=N-1$ and $E_\ell(w)=O(1)$ for $\ell\ge1$, giving $\mathcal E(w)=\Theta(N)$; while $\lVert w\rVert_1=N-1$ and $\mathrm{TV}(w)=\Theta(N)$, so the right side of (6.1) is $\Theta(N^{3/2})$. The divisibility constraint forbids exactly this: an alternating jump of magnitude $2$ at position $s$ would require $2\in s\mathbb Z$. Every step below uses the constraint, and the argument has no analogue for unconstrained vectors. This is also why the two-step estimate of Theorem~\ref{thm:5.5}, which quantifies over arbitrary $w\in\mathbb R^N$, cannot be repaired by any reweighting of the levels: its limitation is real, and Theorem~\ref{thm:6.2} escapes it by leaving that class of arguments entirely.

\textbf{Runs and prices.} A \textbf{run} of $v$ is a maximal integer interval $R=[a,b]$ on which $v$ is non-zero and of constant sign. Runs are pairwise disjoint; two of them may be adjacent, with no zero between, when $v$ changes sign in one step. Put
\begin{equation*}
W_R=b-a+1,\qquad \mu_R=\sum_{s\in R}|v(s)| ,
\end{equation*}
and for an integer interval $I$,
\begin{equation*}
\mu(I)=\sum_{s\in I}|v(s)|,\qquad p(I)=\frac{\mu(I)}{\sqrt{|I|}} ,
\end{equation*}
the \textbf{price} of $I$. Since $v(1)=0$, every run satisfies $a\ge2$.

Two consequences of the divisibility are used repeatedly and we record them once.

\begin{lemma}[position costs amplitude]\label{lem:6.3}
Let $v\in\mathcal Q_N$ be non-zero. Then
\begin{equation}\label{eq:6.2}
\Delta v(s)\ne0\ \Longrightarrow\ |\Delta v(s)|\ge s ,
\end{equation}
and $T\ge2\max_s|v(s)|\ge2$.
\end{lemma}

\begin{proof}
The implication is the definition of $\mathcal Q_N$ together with $\Delta v(s)\in s\mathbb Z$. For the second claim fix $s$ and follow $v$ from $v(1)=0$ up to $v(s)$ and back down to $v(N+1)=0$: the two passages use disjoint sets of jump positions and each accumulates total variation at least $|v(s)|$.
\end{proof}

The implication (6.2) is the whole of the arithmetic input. Read contrapositively it says that $v$ may oscillate cheaply only near the origin: an oscillation at position $s$ costs amplitude at least $s$, so wide low-amplitude wobble is confined to small $s$, and at large $s$ every change of direction is spike-scale. Nothing in the ambient dyadic analysis sees this.

\textbf{Architecture of the proof.} Four modules, independent of one another:

\begin{enumerate}
\item a Whitney energy estimate (\S\ref{sec:6.2}), which converts $\mathcal E$ into a supply of price bounds on every family of disjoint subintervals of runs;
\item a Carleson stopping time (\S\ref{sec:6.3}), which spends that supply to isolate the positions where $v$ is \emph{locally} large, leaving a clean remainder on which every ancestor interval is cheap;
\item an island estimate (\S\ref{sec:6.4}), which bounds the clean mass sitting above the half-line $|v(s)|=\frac{s-1}2$;
\item the folded-remainder capital lemma (\S\ref{sec:6.5}), which bounds \emph{all} the mass below that half-line by $8$ times the mass above it.
\end{enumerate}

Modules 1--3 are analytic and use the divisibility only through the total-variation budget of the island count. Module 4 is arithmetic and is where the divisibility does its real work. Section 6.6 assembles them and disposes of the complementary regime $L<252\,T$ by a separate and much shorter argument.

\subsection{Dyadic--Whitney energy}\label{sec:6.2}

\begin{lemma}[Whitney tiling]\label{lem:6.4}
Let $I\subseteq[N]$ be an integer interval and let $\mathcal W(I)$ be the family of maximal dyadic intervals contained in $I$. Then $\mathcal W(I)$ is a partition of $I$, every member has length at most $|I|$, at most two members have any given length, and
\begin{equation}\label{eq:6.3}
\sum_{K\in\mathcal W(I)}\sqrt{|K|}\ <\ 7\sqrt{|I|} .
\end{equation}
\end{lemma}

\begin{proof}
Two dyadic intervals are nested or disjoint, so maximal ones are pairwise disjoint; every point of $I$ lies in a dyadic interval contained in $I$, namely its own singleton, hence in a maximal one. So $\mathcal W(I)$ partitions $I$, and $|K|\le|I|$ for every $K\in\mathcal W(I)$.

Fix a length $2^\ell$ and let $K\in\mathcal W(I)$ have that length. The dyadic parent of $K$ is not contained in $I$, by maximality, and it meets $I$; hence it contains a point outside $I$ on one side of $K$, so it crosses an endpoint of $I$. Distinct level-$\ell$ intervals have disjoint parents, and each of the two endpoints of $I$ lies in at most one level-$\ell$ parent, so at most two members of $\mathcal W(I)$ have length $2^\ell$. Summing over the levels present,
\begin{equation*}
\sum_{K\in\mathcal W(I)}\sqrt{|K|}\ \le\ 2\sum_{2^\ell\le|I|}2^{\ell/2}\ \le\ \frac{2\sqrt2}{\sqrt2-1}\sqrt{|I|}\ <\ 6.829\,\sqrt{|I|} . 
\qedhere
\end{equation*}
\end{proof}

Figure~\ref{fig:whitney} shows the tiling of one run, with the dyadic parent of each tile drawn escaping the run.

\begin{figure}[tb]
\centering
\includegraphics[width=0.7\textwidth]{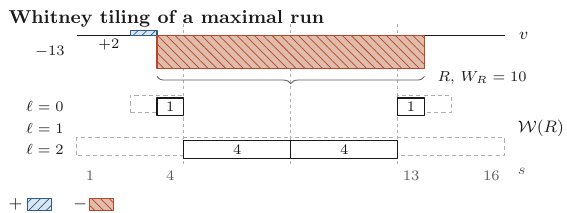}
\caption{\textbf{Whitney tiling of a sign-constant run.}
For $\lambda=(3,3,3,3,3)$ against $\mu=(13,2)$ at $n=15$, $N=16$, the
difference $v=u_\lambda-u_\mu$ is $+2$ at $s=3$ and $-13$ on $R=[4,13]$, so
$R$ is maximal, closed on the left by a change of sign and on the right by
$v=0$. The maximal dyadic intervals contained in $R$ partition it, and each is
drawn at the height of its own level $\ell$, so the lengths $1,4,4,1$ form a
pyramid with at most two boxes on any horizontal line: that is the second
claim of Lemma~\ref{lem:6.4}. Above each tile, dashed, is its dyadic parent;
every parent leaves $R$, which is exactly why the tile is maximal.
Lemma~\ref{lem:6.5} and Proposition~\ref{prop:6.6} use both facts, the partition to rule out cancellation
between the tiles and the two-per-level cap to bound
$\sum_{I\in\mathcal W(R)}2^{\ell(I)/2}$ by $7\sqrt{W_R}$.
Legend as in Fig.~\ref{fig:transportation}.}
\label{fig:whitney}
\end{figure}

\begin{lemma}[single-interval estimate]\label{lem:6.5}
Let $I$ be an integer interval contained in a run of $v$. Then
\begin{equation}\label{eq:6.4}
\frac{\mu(I)^2}{\sqrt{|I|}}\ \le\ 7\sum_{K\in\mathcal W(I)}2^{-\ell(K)/2}\bigl\langle v,\mathbf 1_K\bigr\rangle^2 .
\end{equation}
\end{lemma}

\begin{proof}
On $I$ the sign of $v$ is constant and $\mathcal W(I)$ partitions $I$, so the interval sums $\langle v,\mathbf 1_K\rangle$, $K\in\mathcal W(I)$, all have that sign and $\sum_K|\langle v,\mathbf 1_K\rangle|=\mu(I)$. Cauchy--Schwarz with the splitting $|\langle v,\mathbf 1_K\rangle|=|K|^{1/4}\cdot|K|^{-1/4}|\langle v,\mathbf 1_K\rangle|$ gives
\begin{equation*}
\mu(I)^2\ \le\ \Bigl(\sum_K\sqrt{|K|}\Bigr)\Bigl(\sum_K|K|^{-1/2}\langle v,\mathbf 1_K\rangle^2\Bigr) ,
\end{equation*}
and (6.3) bounds the first factor by $7\sqrt{|I|}$. Divide by $\sqrt{|I|}$ and note $|K|^{-1/2}=2^{-\ell(K)/2}$.
\end{proof}

\begin{proposition}[disjoint families]\label{prop:6.6}
Let $\{I_j\}$ be pairwise disjoint integer intervals, each contained in a run of $v$. Then
\begin{equation}\label{eq:6.5}
\sum_j\frac{\mu(I_j)^2}{\sqrt{|I_j|}}\ \le\ 7\,\mathcal E(v) .
\end{equation}
\end{proposition}

\begin{proof}
The $I_j$ are disjoint, so the dyadic intervals appearing in the families $\mathcal W(I_j)$ are pairwise distinct. Summing (6.4) over $j$ therefore produces a subsum of $\sum_{K\in\mathcal D_N}2^{-\ell(K)/2}\langle v,\mathbf 1_K\rangle^2=\mathcal E(v)$, all of whose terms are non-negative.
\end{proof}

Taking $\{I_j\}$ to be the runs themselves gives the per-run bound $\mu_R^2/\sqrt{W_R}\le7\mathcal E(v)$, which is all that the argument behind Theorem~\ref{thm:5.5} could use. The strength of (6.5) over that special case is that the family is arbitrary: the stopping time of \S\ref{sec:6.3} produces disjoint subintervals of runs that are not runs, and (6.5) prices them all at once.

\subsection{Stopping-time decomposition}\label{sec:6.3}

Fix a threshold $P>0$. On each run $R=[a,b]$ build the \textbf{balanced interval tree}: the root is $R$; a node $[c,d]$ with $c<d$ has the two children $[c,\lfloor\frac{c+d}2\rfloor]$ and $[\lfloor\frac{c+d}2\rfloor+1,d]$, whose lengths differ by at most one; a node $[c,c]$ is a leaf. Every node is an integer subinterval of $R$, and the parent of a node $Q$ has length at most $2|Q|+1$.

Sort the runs and their nodes as follows.

* If $p(R)>P$, call $R$ \textbf{root-expensive} and place its whole mass $\mu_R$ in the ledger $A$; do not descend into it.
* Otherwise traverse the tree from the root. On reaching a node $I$ with $p(I)>P$, \textbf{stop} at $I$: place $\mu(I)$ in the ledger $B$ and do not visit the descendants of $I$.
* A position $s$ lying in no root-expensive run and in no stopped node is \textbf{clean}.

\begin{lemma}[the two ledgers]\label{lem:6.7}
With $A$ and $B$ the totals just defined,
\begin{equation}\label{eq:6.6}
A+B\ \le\ \frac{14\,\mathcal E(v)}P ,
\end{equation}
and every clean position satisfies $|v(s)|\le P$.
\end{lemma}

\begin{proof}
The root-expensive runs are pairwise disjoint and each is a run, so (6.5) applies to them; and $p(R)>P$ gives $P\mu_R<\mu_R^2/\sqrt{W_R}$, whence $PA<\sum_{p(R)>P}\mu_R^2/\sqrt{W_R}\le7\mathcal E(v)$. The stopped nodes are pairwise disjoint --- within one tree no stopped node is a descendant of another, by construction, and nodes of different trees lie in different runs --- and each is contained in a run, so (6.5) applies again and the same computation gives $PB<7\mathcal E(v)$. Adding yields (6.6).

If $s$ is clean then its run is not root-expensive, so the traversal begins, and no ancestor of the leaf $\{s\}$ is stopped, so the traversal reaches that leaf and does not stop there. Hence $p(\{s\})=|v(s)|\le P$.
\end{proof}

\subsection{The clean-cap island estimate}\label{sec:6.4}

Split the support of $v$ by the \textbf{half-line} $|v(s)|=\frac{s-1}2$:
\begin{equation*}
\mathcal L=\Bigl\{s:0<|v(s)|<\tfrac{s-1}2\Bigr\},\qquad \mathcal H=\Bigl\{s:v(s)\ne0,\ |v(s)|\ge\tfrac{s-1}2\Bigr\} ,
\end{equation*}
and write $M_{\rm low}$ and $M_{\rm cap}$ for the corresponding masses $\sum|v(s)|$. Note that $s\in\mathcal H$ is the same as $s\le2|v(s)|+1$: a position lies above the half-line exactly when its own index is at most twice its height. Put
\begin{equation*}
M_{\rm cap}^{\rm clean}=\sum_{\substack{s\ \text{clean}\\ s\in\mathcal H}}|v(s)| .
\end{equation*}

\begin{figure}[tb]
\centering
\includegraphics[width=\textwidth]{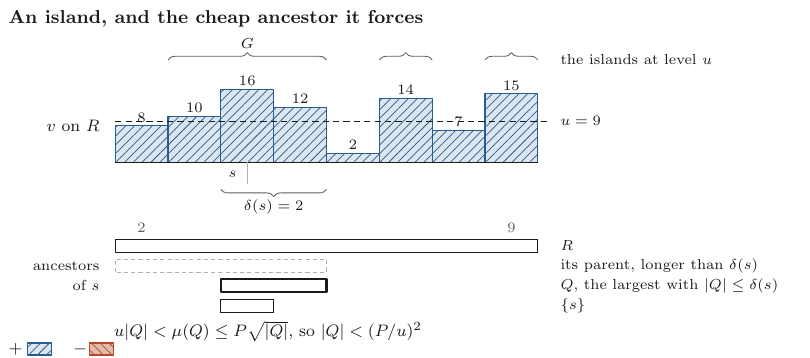}
\caption{\textbf{An island, and the cheap ancestor it forces.}
The run $R=[2,9]$ of a closed quantized sequence, on which $v$ is positive with
the values shown, at threshold $P=30$. Every node of the interval tree of $R$
is cheap, so every position of $R$ is clean. At level $u=9$ the set
$\{|v|>u\}$ has three islands, bracketed above; the first is $G=[3,5]$. The
position $s=4$ lies in $G$ at distance $\delta(s)=2$ from the nearer end. Its
ancestor chain is $\{s\}\subset Q\subset[2,5]\subset R$, and $Q=[4,5]$ is the
largest member of that chain of length at most $\delta(s)$: the next one up is
longer than $\delta(s)$, which is what forces $\delta(s)\le2|Q|$, and $Q$ lies
inside $G$ because it is an interval through $s$ no longer than the distance
from $s$ to either end of $G$. Cheapness of $Q$ together with $Q\subseteq G$
then gives $18=u|Q|<\mu(Q)=28\le P\sqrt{|Q|}=42.4$, hence $|Q|<(P/u)^2$, which
is \eqref{eq:6.10}. That is the step converting one level of the island count
into a bound on how many positions can be both clean and high.
Legend as in Fig.~\ref{fig:transportation}.}
\label{fig:island}
\end{figure}

\begin{lemma}[clean-cap island estimate]\label{lem:6.8}
If $P\ge\sqrt T$ then
\begin{equation}\label{eq:6.7}
M_{\rm cap}^{\rm clean}\ \le\ 14\,P\sqrt T .
\end{equation}
\end{lemma}

\begin{proof}
Write $f(s)=|v(s)|$. Since $\int_0^\infty\mathbf 1\{u\in[f(s)/2,f(s))\}\,du=f(s)/2$,
\begin{equation}\label{eq:6.8}
M_{\rm cap}^{\rm clean}=2\int_0^\infty\#S(u)\,du,\qquad S(u)=\Bigl\{s:\ s\ \text{clean},\ s\in\mathcal H,\ \tfrac{f(s)}2\le u<f(s)\Bigr\} .
\end{equation}
We bound $\#S(u)$ in two ways and interpolate. Figure~\ref{fig:island} draws the second of them on one run.

\textbf{Counting by position.} If $s\in S(u)$ then $f(s)\le2u$ and $s\le2f(s)+1\le4u+1$, so
\begin{equation}\label{eq:6.9}
\#S(u)\le4u+1 .
\end{equation}

\textbf{Counting by islands.} Fix $u>0$. Inside each run, call a maximal integer interval on which $f>u$ a \textbf{level-$u$ island}, and let $I(u)$ be the total number of islands. Every $s\in S(u)$ has $f(s)>u$, so it lies in an island $G=[l,r]$; put
\begin{equation*}
\delta(s)=\min(s-l+1,\ r-s+1) ,
\end{equation*}
its distance to the nearer end of its island. Let $Q$ be the largest node on the ancestor chain of $\{s\}$ with $|Q|\le\delta(s)$; it exists because the leaf $\{s\}$ has length $1\le\delta(s)$.

\emph{The node $Q$ lies inside $G$.} Indeed $Q$ is an interval containing $s$ of length at most $\delta(s)$, so $Q\subseteq[s-|Q|+1,\ s+|Q|-1]$, and
\begin{equation*}
s-|Q|+1\ \ge\ s-\delta(s)+1\ \ge\ l,\qquad s+|Q|-1\ \le\ s+\delta(s)-1\ \le\ r
\end{equation*}
by the two branches of the minimum defining $\delta(s)$.

\emph{The node $Q$ is not much shorter than $\delta(s)$.} If $Q$ is not the root, its parent has length exceeding $\delta(s)$, by maximality of $Q$, and at most $2|Q|+1$; so $\delta(s)<2|Q|+1$ and hence $\delta(s)\le2|Q|$ by integrality. If $Q$ is the root of its tree, then the whole run $R$ satisfies $|R|\le\delta(s)$; since $G\subseteq R$ and $\delta(s)\le\lceil|G|/2\rceil\le|G|\le|R|\le\delta(s)$, all these are equalities, which forces $|G|=1$ and $\delta(s)=1\le2|Q|$. So $\delta(s)\le2|Q|$ in every case.

Since $s$ is clean, $Q$ is a visited node that was not stopped, so $\mu(Q)\le P\sqrt{|Q|}$; and $Q\subseteq G$ gives $\mu(Q)>u|Q|$. Therefore $u\sqrt{|Q|}<P$ and
\begin{equation}\label{eq:6.10}
\delta(s)\ \le\ 2|Q|\ <\ 2\Bigl(\frac Pu\Bigr)^2 .
\end{equation}
An island contains at most $2D$ positions with $\delta\le D$, so
\begin{equation}\label{eq:6.11}
\#S(u)\ \le\ 4\Bigl(\frac Pu\Bigr)^2I(u) .
\end{equation}

\textbf{The island count integrates to $T$.} Let $R=[a,b]$ be a run and $I_R(u)$ its number of level-$u$ islands. The islands of $R$ are the maximal intervals of $\{s\in R:f(s)>u\}$, so $I_R(u)=\mathbf 1\{f(a)>u\}+\#\{s\in(a,b]:f(s-1)\le u<f(s)\}$, and integrating in $u$,
\begin{equation}\label{eq:6.12}
\int_0^\infty I_R(u)\,du=f(a)+\sum_{s=a+1}^b\bigl(f(s)-f(s-1)\bigr)_+ .
\end{equation}
Each term of (6.12) is charged to a jump position, as follows. Since $a\ge2$ and $R$ is maximal, $v(a-1)$ is zero or of the sign opposite to $v(a)$, so $|\Delta v(a-1)|\ge|v(a)|=f(a)$; charge the first term to position $a-1$. For $s\in(a,b]$ the values $v(s-1)$ and $v(s)$ share a sign, so $(f(s)-f(s-1))_+\le|\Delta v(s-1)|$; charge that term to position $s-1$. The positions used by $R$ are therefore $a-1,a,\dots,b-1$, exactly $W_R$ of them.

\emph{These charges do not collide.} Let $R'=[a',b']$ be the next run, so $a'\ge b+1$; it uses the positions $a'-1,\dots,b'-1$, which begin at $a'-1\ge b$, strictly after the last position $b-1$ used by $R$. This is tight when $a'=b+1$, that is when two runs are adjacent because $v$ changes sign in a single step, and it is the reason the charge for the left end of a run is taken at $a-1$ rather than at $a$. Summing (6.12) over the runs,
\begin{equation}\label{eq:6.13}
\int_0^\infty I(u)\,du\ \le\ \sum_{s=1}^N|\Delta v(s)|=T .
\end{equation}

\textbf{Interpolation.} For any $U>0$, by (6.8), (6.9), (6.11) and (6.13),
\begin{equation*}
M_{\rm cap}^{\rm clean}\ \le\ 2\int_0^U(4u+1)\,du+8P^2\int_U^\infty u^{-2}I(u)\,du\ \le\ 4U^2+2U+\frac{8P^2T}{U^2} .
\end{equation*}
Take $U=(2P^2T)^{1/4}$. The outer two terms become $4\sqrt2\,P\sqrt T$ each, totalling $8\sqrt2\,P\sqrt T$. For the middle term, $T\ge1$ and $P\ge\sqrt T$ give $\sqrt P\,T^{1/4}\le P\sqrt T$, so $2U=2\cdot2^{1/4}\sqrt P\,T^{1/4}\le2\cdot2^{1/4}P\sqrt T$. Since $8\sqrt2+2\cdot2^{1/4}<13.7$, (6.7) follows.
\end{proof}

\subsection{The folded-remainder capital lemma}\label{sec:6.5}

Figure~\ref{fig:halfline} draws the half-line on one closed quantized sequence, with its two low plateaus and the bridge between them.

The estimates so far see only the mass above the half-line. Nothing yet prevents $v$ from carrying almost all of its $\ell_1$ mass in a wide, low plateau, at heights far below $\frac{s-1}2$, where the clean-cap bound says nothing. The following lemma removes that possibility outright, and it is the one place where the divisibility constraint is used for more than a total-variation budget.

\begin{figure}[tb]
\centering
\includegraphics[width=\textwidth]{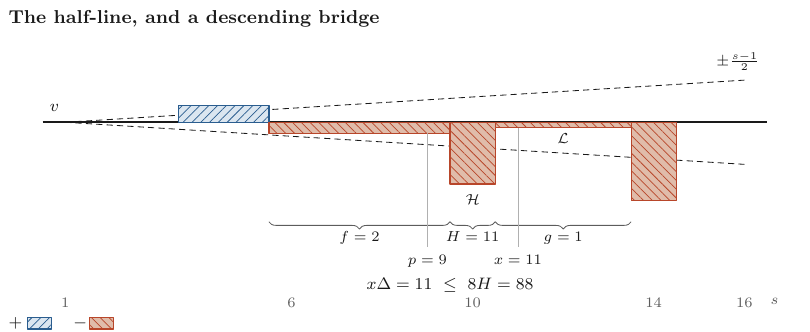}
\caption{\textbf{The half-line, and a descending bridge.}
A closed quantized sequence on $[16]$: $v(1)=v(17)=0$ and every non-zero jump
is $\pm s$ at its position $s$, the smallest the divisibility allows, at
$s=3,5,9,10,13,14$. The dashed envelope is $\pm\frac{s-1}2$. Positions strictly
inside it form $\mathcal L$ and, by Lemma~\ref{lem:6.10}, are exactly flat
there: two plateaus, of heights $f=2$ on $[6,9]$ and $g=1$ on $[11,13]$.
Positions outside it form $\mathcal H$, here $s=4,5,10,14$. Between the two
plateaus sits the bridge $(9,11)$, lying wholly in $\mathcal H$ and carrying
mass $H=11$; Lemma~\ref{lem:6.12} asks only for $x\Delta=11\le8H=88$. Summed
over both plateaus this is Proposition~\ref{prop:6.9}, here
$M_{\mathrm{low}}=11$ against $8M_{\mathrm{cap}}=248$. Note that $v$ changes
sign between $s=5$ and $s=6$ with no zero in between, so the two runs $[4,5]$
and $[6,14]$ are adjacent; the proof of \eqref{eq:6.12} charges the left end of
a run to the jump before it precisely so that adjacent runs cannot collide.
Legend as in Fig.~\ref{fig:transportation}.}
\label{fig:halfline}
\end{figure}

\begin{proposition}[folded-remainder capital lemma]\label{prop:6.9}
 For every $v\in\mathcal Q_N$,
\begin{equation}\label{eq:6.14}
M_{\rm low}\ \le\ 8\,M_{\rm cap} .
\end{equation}
\end{proposition}

The proof occupies the rest of this subsection. Its shape is: the low region is a union of exactly flat plateaus (Lemma~\ref{lem:6.10}); a plateau can only lose height by paying, and the currency is a residue that decreases at a controlled rate (Lemma~\ref{lem:6.11}); each drop is paid for by the mass above the half-line lying between the two plateaus (Lemma~\ref{lem:6.12}); and an Abel summation converts the total plateau mass into a sum of such drops.

\begin{lemma}[the low region is flat]\label{lem:6.10}
If $s$ and $s+1$ both lie in $\mathcal L$ then $v(s)=v(s+1)$. Consequently every maximal contiguous component of $\mathcal L$ is an interval $F=[a,b]$ on which $v$ is constant, of some height $f=|v(a)|$.
\end{lemma}

\begin{proof}
Suppose $s,s+1\in\mathcal L$, so $|v(s)|<\frac{s-1}2$ and $|v(s+1)|<\frac s2$. If the two values have the same sign then $|\Delta v(s)|\le\max(|v(s)|,|v(s+1)|)<\frac s2$; if they have opposite signs then $|\Delta v(s)|=|v(s)|+|v(s+1)|<\frac{s-1}2+\frac s2<s$. Either way $|\Delta v(s)|<s$, so $\Delta v(s)=0$ by (6.2), giving $v(s)=v(s+1)$; and the opposite-sign case is thereby excluded.
\end{proof}

Since a plateau of height $f\ge1$ at position $a$ requires $f<\frac{a-1}2$, every plateau satisfies $a\ge2f+2\ge4$. In particular no plateau meets $\{1,2,3\}$.

\begin{lemma}[folded remainder, one step]\label{lem:6.11}
For $z\in\mathbb Z$ and $m\ge1$ put $\rho_m(z)=\operatorname{dist}(z,m\mathbb Z)$, and for $2\le t\le N+1$ put $R_t=\rho_{t-1}(v(t))$, the value at $t=N+1$ being available because $v$ is closed. Then $R_{t+1}=\rho_t(v(t))$ for $2\le t\le N$, and
\begin{equation}\label{eq:6.15}
R_t-R_{t+1}\ \le\ \frac{|v(t)|}t+\frac12 .
\end{equation}
\end{lemma}

\begin{proof}
Divisibility gives $v(t+1)\equiv v(t)\pmod t$, whence $R_{t+1}=\rho_t(v(t+1))=\rho_t(v(t))$. Choose $k\in\mathbb Z$ nearest to $v(t)/t$, so that $|v(t)-kt|=\rho_t(v(t))=R_{t+1}$ and $|k|\le\frac{|v(t)|}t+\frac12$. Then $k(t-1)\in(t-1)\mathbb Z$ is a competitor for $R_t$:
\begin{equation*}
R_t\ \le\ |v(t)-k(t-1)|\ \le\ |v(t)-kt|+|k|\ \le\ R_{t+1}+\frac{|v(t)|}t+\frac12 . 
\qedhere
\end{equation*}
\end{proof}

The quantity $R_t$ is the \textbf{folded remainder}: it measures how far $v(t)$ is from being a legal value at the modulus available one step to the left. On a low plateau of height $f$ ending at $p$ we have $f<\frac{p-1}2$, so $0$ is the nearest multiple of $p-1$ and $R_p=f$ exactly; the folded remainder reads off the plateau height. Inequality (6.15) says that this reading can decrease by at most $\frac{|v(t)|}t+\frac12$ per step. A plateau therefore cannot lose height without either passing through positions of large $|v|$, which is mass above the half-line, or spending many steps, which by (6.2) again forces large $|v|$. Lemma~\ref{lem:6.12} makes the dichotomy quantitative.

\begin{lemma}[descending bridge]\label{lem:6.12}
Let $F=[a,p]$ be a low plateau of height $f$, and let $x>p$ be the least position after $p$ that is either the first position of the next low plateau, of height $g$, or the first zero of $v$ after $p$, in which case $g=0$. Suppose $\Delta:=f-g>0$ and put
\begin{equation*}
H=\sum_{t=p+1}^{x-1}|v(t)| .
\end{equation*}
Then $H\ge\frac18x\Delta$, and every position counted in $H$ lies in $\mathcal H$.
\end{lemma}

\begin{proof}
By construction no position of $(p,x)$ lies in $\mathcal L$ and none is a zero of $v$, so $(p,x)\subseteq\mathcal H$; in particular
\begin{equation}\label{eq:6.16}
H\ \ge\ \frac12\sum_{t=p+1}^{x-1}(t-1) .
\end{equation}
Both endpoints lie below the half-line, so $R_p=f$ and $R_x=g$, the latter also in the terminal case since $\rho_{x-1}(0)=0$. Summing (6.15) over $t=p,\dots,x-1$ and telescoping,
\begin{equation}\label{eq:6.17}
\Delta=R_p-R_x\ \le\ \sum_{t=p}^{x-1}\frac{|v(t)|}t+\frac{x-p}2\ \le\ \frac fp+\frac Hp+\frac{x-p}2\ \le\ \frac Hp+\frac{x-p+1}2 ,
\end{equation}
using $t\ge p$ in the denominators and $\frac fp<\frac12$. Note $p\ge4$, so (6.15) is available at every $t$ in the range. Write $d=x-p$.

\emph{Case $x\ge2p$.} Then $p\le\frac x2$, and (6.16) gives $H\ge\frac12\sum_{u=p}^{x-2}u\ge\frac12\sum_{u=\lceil x/2\rceil}^{x-2}u\ge\frac{x^2}{16}$ for $x\ge8$, which holds since $x\ge2p\ge8$. On the other hand $\Delta<f<\frac p2\le\frac x4$, so $\frac18x\Delta<\frac{x^2}{32}<H$.

\emph{Case $x<2p$, and $H\ge\frac12p\Delta$.} Then $H\ge\frac12\cdot\frac x2\cdot\Delta=\frac14x\Delta$.

\emph{Case $x<2p$, and $H<\frac12p\Delta$.} Then (6.17) gives $\Delta<\frac\Delta2+\frac{d+1}2$, so $\Delta<d+1$ and hence $d\ge\Delta$.

Suppose first $\Delta\ge2$. The range $(p,x)$ has $d-1$ positions, each with $t-1\ge p$, so (6.16) gives $H\ge\frac12p(d-1)\ge\frac12p\cdot\frac\Delta2=\frac14p\Delta>\frac18x\Delta$, where $d-1\ge\Delta-1\ge\frac\Delta2$ was used.

Suppose next $\Delta=1$. We claim $d\ge2$. If $F$ is followed by another low plateau then $x\ge p+2$ already, since the components of $\mathcal L$ are maximal and so $p+1\notin\mathcal L$. In the terminal case $x$ is the first zero after $p$; if $d=1$ then $\Delta v(p)=v(p)-0=v(p)$ has magnitude $f<\frac{p-1}2<p$ and is non-zero, contradicting (6.2). So $d\ge2$ in both cases, and $H\ge\frac12p(d-1)\ge\frac p2>\frac x4=\frac14x\Delta$.

In every case $H\ge\frac18x\Delta$.
\end{proof}

The case $\Delta=1$, $d=1$ is the only point at which the terminal and non-terminal bridges behave differently, and it is worth isolating why: a plateau of height $1$ can fall to $0$ in a single step only by a jump of magnitude $1$, and (6.2) forbids that at any position beyond the first. The divisibility constraint is doing exactly the work that the naive estimate cannot.

\begin{proof}[Proof of Proposition~\ref{prop:6.9}]
It suffices to prove (6.14) with both sides restricted to one maximal interval of positions on which $v$ is non-zero; summing over these components then gives the proposition, since each side is additive over them and the plateaus and bridges of one component lie inside it.

Fix such a component and let $F_1=[a_1,b_1],\dots,F_m=[a_m,b_m]$ be the low plateaus it contains, in order, of heights $f_1,\dots,f_m$; if $m=0$ there is nothing to prove. Let $z$ be the first zero of $v$ after the component, put
\begin{equation*}
x_0=a_1,\qquad x_j=a_{j+1}\ (1\le j<m),\qquad x_m=z,\qquad f_{m+1}=0 .
\end{equation*}
Since $b_j<a_{j+1}$ for $j<m$ and $b_m<z$, we have $|F_j|\le x_j-x_{j-1}$ for every $j$, so
\begin{equation}\label{eq:6.18}
\sum_{j=1}^mf_j|F_j|\ \le\ \sum_{j=1}^mf_j(x_j-x_{j-1})\ =\ -f_1x_0+\sum_{j=1}^mx_j\bigl(f_j-f_{j+1}\bigr)\ \le\ \sum_{j=1}^mx_j\bigl(f_j-f_{j+1}\bigr)_+ ,
\end{equation}
the middle step being Abel summation and the last using $f_1x_0\ge0$. This is the \textbf{capital identity}: the mass of the low plateaus, which is the left side, is at most the total capital $x_j(f_j-f_{j+1})_+$ released by their descents, each drop weighted by the position at which it occurs.

Each positive term of the right side of (6.18) is a descending bridge in the sense of Lemma~\ref{lem:6.12}, with $p=b_j$, $x=x_j$ and $\Delta=f_j-f_{j+1}$; Lemma~\ref{lem:6.12} applies to $j=m$ as well, where $x_m=z$ is the terminal zero and $f_{m+1}=0$. So $x_j(f_j-f_{j+1})_+\le8H_j$ with $H_j$ the mass of $(b_j,x_j)$. The open intervals $(b_j,x_j)$ are pairwise disjoint and contained in $\mathcal H$, so $\sum_jH_j$ is at most the mass of $\mathcal H$ in this component. Combining with (6.18) gives $M_{\rm low}\le8M_{\rm cap}$ on the component.
\end{proof}

\subsection{The thick and thin regimes}\label{sec:6.6}

\begin{lemma}[thin regime]\label{lem:6.13}
For every non-zero $v\in\mathcal Q_N$,
\begin{equation}\label{eq:6.19}
\mathcal E(v)\ \ge\ \frac{L^{3/2}}{35} .
\end{equation}
\end{lemma}

\begin{proof}
\textbf{A position tax on each run.} Let $R=[a_R,b_R]$ be a run and $A_R=\max_{s\in R}|v(s)|$. Assign to $R$ a position carrying height more than $\frac{b_R}2$, as follows. If $A_R\ge\frac{b_R}2$, assign a position of $R$ at which $|v|=A_R$. If $A_R<\frac{b_R}2$, then $\Delta v(b_R)\ne0$ by maximality of $R$, so $|\Delta v(b_R)|\ge b_R$ by (6.2); since $|v(b_R)|\le A_R<\frac{b_R}2$, this forces $v(b_R+1)\ne0$ and
\begin{equation*}
|v(b_R+1)|\ \ge\ |\Delta v(b_R)|-|v(b_R)|\ >\ b_R-\frac{b_R}2=\frac{b_R}2 ,
\end{equation*}
so assign $b_R+1$ to $R$. This second branch cannot occur when $b_R=N$: there $v(N+1)=0$, so $|\Delta v(N)|=|v(N)|\le A_R<\frac N2<N$ would be a non-zero jump of magnitude below $N$, contradicting (6.2). Hence the assigned position always lies in $[N]$.

A position is assigned at most twice, once as a peak of the run containing it and once as the successor of the run ending just before it. Since $\mathcal E(v)\ge E_0(v)=\sum_sv(s)^2$ and each assigned position contributes at least $\frac{b_R^2}4$,
\begin{equation}\label{eq:6.20}
\mathcal E(v)\ \ge\ \frac18\sum_Rb_R^2\ =:\ \frac F8 .
\end{equation}

\textbf{Width is controlled by the tax.} Put $S=\sum_R\sqrt{W_R}$ and group the runs by the octave of their right endpoint: let $k_i$ be the number of runs with $b_R\in[2^i,2^{i+1})$ and $S_i$ their contribution to $S$. Those runs are pairwise disjoint and contained in $[1,2^{i+1})$, so their widths total at most $2^{i+1}$, and Cauchy--Schwarz gives $S_i\le\sqrt{k_i2^{i+1}}$; the same containment gives $k_i\le2^{i+1}$, since each run has width at least one. Put $x_i=4^ik_i$, so that
\begin{equation*}
S_i\le\sqrt{2x_i}\,2^{-i/2},\qquad \sum_ix_i\le\sum_Rb_R^2=F,\qquad x_i\le2^{3i+1} .
\end{equation*}
Let $i_0$ be the integer with $2^{i_0}\le F^{1/3}<2^{i_0+1}$. Using the capacity bound $x_i\le2^{3i+1}$ below $i_0$ and Cauchy--Schwarz above it,
\begin{equation*}
\sum_{i\le i_0}S_i\le\sqrt2\sum_{i\le i_0}2^{(3i+1)/2}2^{-i/2}=2\sum_{i\le i_0}2^i<2^{i_0+2}\le4F^{1/3} ,
\end{equation*}
\begin{equation*}
\sum_{i>i_0}S_i\le\sqrt2\Bigl(\sum_{i>i_0}x_i\Bigr)^{1/2}\Bigl(\sum_{i>i_0}2^{-i}\Bigr)^{1/2}\le\sqrt{2F}\,2^{-i_0/2}<2F^{1/3} ,
\end{equation*}
the last step by $2^{i_0}>\frac12F^{1/3}$. Hence $S\le6F^{1/3}\le8F^{1/3}$.

\textbf{Combination.} Apply (6.5) to the family of all runs and then Cauchy--Schwarz:
\begin{equation*}
L^2=\Bigl(\sum_R\mu_R\Bigr)^2\le\Bigl(\sum_R\frac{\mu_R^2}{\sqrt{W_R}}\Bigr)\Bigl(\sum_R\sqrt{W_R}\Bigr)\le7\,\mathcal E(v)\,S\le56\,\mathcal E(v)F^{1/3} .
\end{equation*}
By (6.20), $F\le8\mathcal E(v)$, so $L^2\le56\mathcal E(v)\bigl(8\mathcal E(v)\bigr)^{1/3}=112\,\mathcal E(v)^{4/3}$, that is $\mathcal E(v)\ge L^{3/2}/112^{3/4}$, and $112^{3/4}<34.5<35$.
\end{proof}

\begin{proof}[Proof of Theorem~\ref{thm:6.2}]
Let $v\in\mathcal Q_N$ be non-zero, so $T\ge1$.

\textbf{Thick regime: $L\ge252\,T$.} Take $P=\dfrac L{252\sqrt T}$, which satisfies $P\ge\sqrt T$ precisely because $L\ge252T$. Every non-zero position lies in $\mathcal L$ or in $\mathcal H$, and every position of $\mathcal H$ lies in a root-expensive run, in a stopped node, or is clean; so by Proposition~\ref{prop:6.9}, Lemma~\ref{lem:6.7} and Lemma~\ref{lem:6.8},
\begin{equation*}
L=M_{\rm low}+M_{\rm cap}\ \le\ 9M_{\rm cap}\ \le\ 9\bigl(A+B+M_{\rm cap}^{\rm clean}\bigr)\ \le\ \frac{126\,\mathcal E(v)}P+126\,P\sqrt T .
\end{equation*}
The choice of $P$ makes the last term equal to $\frac L2$, so $\frac L2\le\frac{126\,\mathcal E(v)}P=\frac{126\cdot252\sqrt T}L\mathcal E(v)$, that is
\begin{equation*}
\mathcal E(v)\ \ge\ \frac{L^2}{2\cdot126\cdot252\,\sqrt T}=\frac1{63504}\cdot\frac{L^2}{\sqrt T} .
\end{equation*}

\textbf{Thin regime: $L<252\,T$.} By Lemma~\ref{lem:6.13},
\begin{equation*}
\mathcal E(v)\ \ge\ \frac{L^{3/2}}{35}=\frac1{35}\cdot\frac{L^2}{\sqrt L}\ >\ \frac1{35\sqrt{252}}\cdot\frac{L^2}{\sqrt T}\ >\ \frac1{556}\cdot\frac{L^2}{\sqrt T} ,
\end{equation*}
which is stronger than required.
\end{proof}

\subsection{The constant, and a sharper bound off the flat-spectrum regime}\label{sec:6.7}

Theorem~\ref{thm:6.2} settles the order and nothing more. Its constant is far from what the extremal families suggest: \S\ref{sec:7} exhibits pairs with $\mathcal E(v)\sqrt n/\lVert v\rVert_1^2\to\frac23$, and the exhaustive minimum of that ratio over all pairs with $n\le30$ is $0.887^2=0.787$ (Appendix B), against the $\frac1{63504\sqrt2}<1.2\cdot10^{-5}$ that (6.1) guarantees after substituting $T\le2n$. The discrepancy is a factor of about $6\cdot10^4$, and it is entirely an artefact of the bookkeeping: the constant $9$ of Proposition~\ref{prop:6.9}, the two constants $14$ of Lemmas~\ref{lem:6.7} and 6.8, and the threshold $252$ separating the regimes all compound. We have made no attempt to reduce them, since the order is what the application needs.

That the loss is bookkeeping rather than substance can be seen from a second, much cheaper estimate, which is sharp in constant but blind in exactly one configuration. It is the bound the previous version of this work obtained, and we record it because it is the better inequality on every pair whose total variation is concentrated in few octaves.

\begin{proposition}[octave bound]\label{prop:6.14}
For $0\le i\le J$ put $\tau_i=\sum_{s\in[2^i,2^{i+1})}|\Delta v(s)|$, so $\sum_i\tau_i=T$. Then for every non-zero $v\in\mathcal Q_N$,
\begin{equation}\label{eq:6.21}
\mathcal E(v)\ \ge\ \frac{L^2}{14\sum_{i=0}^J\sqrt{\tau_i}} .
\end{equation}
In particular $\mathcal E(v)\ge L^2/\bigl(14\Gamma\sqrt T\bigr)$ whenever $\sum_i\sqrt{\tau_i}\le\Gamma\sqrt T$, and $\mathcal E(v)\ge L^2/\bigl(14\sqrt{T(J+1)}\bigr)$ always.
\end{proposition}

\begin{proof}
In Appendix C. The argument assigns each run to the octave of its right endpoint, uses (6.2) to bound the number of runs an octave can hold by $\tau_i2^{-i}$, and then applies Cauchy--Schwarz once within each octave and once across octaves.
\end{proof}

Comparing (6.21) with (6.1): the two agree in order exactly when the octave spectrum $(\tau_i)$ is concentrated, and (6.21) degrades by $\sqrt{J+1}=\Theta(\sqrt{\log n})$ when the spectrum is flat, which is the configuration Theorem~\ref{thm:6.2} was proved to handle. Off that configuration (6.21) is better in the constant by a factor of about $63504/14\approx4500$. Neither statement supersedes the other, and both are unconditional.

\section{Sharpness and the chirp family}\label{sec:7}

Theorem~\ref{thm:6.2} is an order statement with a crude constant. This section calibrates it from the other side, with a family that is exactly solvable and, as far as we can determine, extremal. The family also explains the mechanism: a closed quantized sequence can be a near-perfect multiscale canceller, putting essentially all of its energy at the finest level, and yet the divisibility constraint stops it from cancelling too much, because a longer oscillating ladder would exceed the total-variation budget.

Exhaustive minimization of
\begin{equation*}
\kappa(\lambda,\mu)=\frac{\sqrt{\mathcal E(v)}\;n^{1/4}}{\lVert v\rVert_1}
\end{equation*}
over all pairs with $n\le30$ (Appendix B) finds the minimizers to be \textbf{chirps}: difference vectors that reverse sign at every second position, each jump being a small multiple of the minimum that Lemma~\ref{lem:2.6} permits there. It is the pair drawn in Fig.~\ref{fig:staircase}. The record low over that range is at $n=30$, $\lambda=(10,7,7,3,3)$ against $\mu=(9,9,5,5,1,1)$, with
\begin{equation*}
v=(0,\,2,\,2,\,-4,\,-4,\,6,\,6,\,-8,\,-8,\,10).
\end{equation*}

\textbf{Construction.} For $m\equiv2\pmod4$, write $m=4M+2$ and define the \textbf{chirp pair}
\begin{equation*}
\lambda^{\mathrm{ch}}_m=(m)\cup\bigl\{t,t\ :\ t\equiv3\!\!\pmod4,\ 3\le t\le m-3\bigr\},\qquad \mu^{\mathrm{ch}}_m=\bigl\{t,t\ :\ t\equiv1\!\!\pmod4,\ 1\le t\le m-1\bigr\},
\end{equation*}
an interleaved doubled odd ladder $m\mid(m-1)^2\mid(m-3)^2\mid\cdots\mid1^2$ with alternating ownership. Both have size $\frac{m(m+2)}4$, so common blocks of size $1$ embed the pair into $\mathcal K_n/\!\cong$ whenever $m(m+2)\le4n$, that is, up to $m\sim2\sqrt n$. The same ladder is available for every even $m$; the residue condition $m\equiv2\pmod4$ only fixes the ownership of the last rung $t=1$, which is what makes the computation below uniform in $m$.

\begin{proposition}[the chirp pair, exactly]\label{prop:7.1}
Let $v=v_m$ be the difference vector of the chirp pair, $m=4M+2$. Then:

\textbf{(a)} $v(1)=0$, $v(m)=m$, and $v\equiv(-1)^i(m-2i)$ on $\{m-2i,\ m-2i+1\}$ for $1\le i\le\frac{m-2}2$. Every non-zero jump below $s=m$ sits at an odd position $s$ and has magnitude exactly $2s$, twice the minimum permitted by Lemma~\ref{lem:2.6}, while the jump at $s=m$ has magnitude $m$, the minimum itself; consequently
\begin{equation*}
\mathrm{TV}(v)=2\!\!\sum_{s<m,\ s\ \text{odd}}\!\!s\ +\ m=\frac{m(m+2)}2 ,
\end{equation*}
which equals $2n$ exactly when the ladder fills $n$, so that the chirp meets the budget of Lemma~\ref{lem:2.6} with nothing to spare. Finally $\sum_sv(s)=\lVert\lambda^{\mathrm{ch}}_m\rVert_2^2-\lVert\mu^{\mathrm{ch}}_m\rVert_2^2=2$.

\textbf{(b)} With $U=\frac{m-2}2$,
\begin{gather*}
\lVert v\rVert_1=\frac{m^2}2,\qquad E_0(v)=m^2+\frac43U(U+1)(2U+1)=\frac{m^3}3+O(m^2),\\
E_1(v)=2m,\qquad E_\ell(v)=4\ \ (2\le\ell\le J).
\end{gather*}

\textbf{(c)} Consequently $\kappa^2=\dfrac{4\sqrt n}{3m}\bigl(1+O(m^{-1})\bigr)$, and choosing $m$ as large as feasibility permits, $m=2\lfloor\sqrt n\rfloor-O(1)$,
\begin{equation*}
\inf_{\lambda\ne\mu\vdash n}\kappa^2\ \le\ \frac23+O\bigl(n^{-1/2}\bigr).
\end{equation*}
In particular no lower bound of the form $\mathcal E(v)\ge c\lVert v\rVert_1^2/\sqrt n$ can hold with $c>\frac23$.
\end{proposition}

\begin{proof}
(a) Order the rungs $r_0=m>r_1=m-1>r_2=m-3>\cdots$, with $r_i=m-(2i-1)$ for $i\ge1$, carrying weight $1$ at $r_0$ and weight $2$ below, owned by $\lambda^{\mathrm{ch}}_m$ for even $i$ and by $\mu^{\mathrm{ch}}_m$ for odd $i$. Since $u$ is additive over parts, $v$ is constant between consecutive rungs, and descending across a rung changes $v$ by $\pm(\text{weight})\cdot(\text{rung})$ according to ownership. Above $m$, $v=0$; crossing $r_0$ gives $v=m$ on $\{m\}$. Inductively, if $v=(-1)^i(m-2i)$ on run $i$ then crossing $r_{i+1}=m-2i-1$ gives $(-1)^{i+1}(m-2i-2)$: for even $i$, $(m-2i)-2(m-2i-1)=-(m-2i-2)$, and for odd $i$, $-(m-2i)+2(m-2i-1)=m-2i-2$. Run $i\ge1$ occupies $(r_{i+1},r_i]=\{m-2i,m-2i+1\}$. The last rung is $t=1$, owned by $\mu^{\mathrm{ch}}_m$, giving $v(1)=2-2=0$. At an odd position $s=m-2i-1$ the two adjacent run values are $(-1)^i(m-2i)$ and $(-1)^{i+1}(m-2i-2)$, of opposite sign, so the jump has magnitude $(m-2i)+(m-2i-2)=2s$; at an even position below $m$ the two adjacent values lie in one run and coincide. One jump remains, at $s=m$: there $v(m)=m$ and $v(m+1)=0$, so its magnitude is $m=1\cdot m$, which is Lemma~\ref{lem:2.6} at $s=m$, where $\lambda^{\mathrm{ch}}_m$ has one part of size $m$ and $\mu^{\mathrm{ch}}_m$ none. The odd positions carrying a jump are $s=1,3,\dots,m-1$, of sum $(m/2)^2$, so $\mathrm{TV}(v)=2(m/2)^2+m=\frac{m(m+2)}2$; since $\lambda^{\mathrm{ch}}_m$ has size $\frac{m(m+2)}4$, this is $2n$ when the ladder needs no padding and less otherwise. For the mean, $\sum_su_\lambda(s)=\sum_i\sum_{s\le\lambda_i}\lambda_i=\lVert\lambda\rVert_2^2$, so $\sum_sv(s)=\lVert\lambda\rVert_2^2-\lVert\mu\rVert_2^2$; substituting the two ladders,
\begin{equation*}
\lVert\lambda^{\mathrm{ch}}_m\rVert_2^2-\lVert\mu^{\mathrm{ch}}_m\rVert_2^2=(4M+2)^2+2\sum_{t=1}^M(4t-1)^2-2\sum_{t=0}^M(4t+1)^2=(4M+2)^2-32\sum_{t=1}^Mt-2=2 .
\end{equation*}

(b) $\lVert v\rVert_1=m+2\sum_{i=1}^{U}(m-2i)=m+\frac{m(m-2)}2=\frac{m^2}2$. Writing $m-2i=2w$ with $1\le w\le U$, $E_0=\lVert v\rVert_2^2=m^2+2\sum_{i}(m-2i)^2=m^2+8\sum_{w\le U}w^2=m^2+\frac43U(U+1)(2U+1)$. For $E_1$, level-$1$ intervals are the pairs $\{2t-1,2t\}$, and by (a) each such pair meeting the support has sum $\pm2$; there are $m/2$ of them, so $E_1=4\cdot\frac m2=2m$. For $\ell\ge2$, group positions into blocks $\{4j+1,\dots,4j+4\}$: by (a) each complete block below $m-2$ has signed sum $4j-(4j+2)-(4j+2)+(4j+4)=0$, so every prefix sum at a multiple of $4$ vanishes below $m-2$, while the positions above contribute the total $\sum_sv(s)=2$. Every level-$\ell$ endpoint with $\ell\ge2$ is a multiple of $4$, so each interval sum is $0$, or $2$ for the single interval containing $\{m-1,m\}$; hence $E_\ell=4$.

(c) By (b), $\mathcal E(v)=E_0+2^{-1/2}\cdot2m+4\sum_{\ell\ge2}2^{-\ell/2}=\frac{m^3}3+O(m^2)$, and $\kappa^2=\mathcal E\sqrt n/\lVert v\rVert_1^2=\frac{(m^3/3)\sqrt n}{m^4/4}(1+O(m^{-1}))$. Feasibility caps $m$ at $2\sqrt n-O(1)$, at which $\kappa^2\to\frac23$.
\end{proof}

The chirp is a perfect multiscale canceller: essentially all of its weighted energy sits at level $0$, and yet it cannot push $\kappa$ below $\sqrt{2/3}$, because the mass constraint stops the ladder at height $m\approx2\sqrt n$; equivalently, by (a) its total variation already exhausts the budget $2n$ of Lemma~\ref{lem:2.6}, leaving no room for a longer ladder. Numerically, $\kappa$ along the chirp descends to $0.81701$ against $\sqrt{2/3}=0.81650$ at $m=802$. The exhaustive minima for $n\le30$ lie in $[0.887,1.159]$, and while they do not decrease monotonically, their record lows fall at exactly the three values of $n$ at which a full ladder fits, $n=12,20,30$, each time at the corresponding ladder pair (Appendix B). This is a calibration, not evidence of a proof: it fixes the constant a proof would have to produce, and it says nothing about whether one exists.

A lower bound $\kappa\ge c$ is exactly a lower bound $\mathcal E(v)\ge c^2\lVert v\rVert_1^2/\sqrt n$, so Theorem~\ref{thm:6.2} with $\mathrm{TV}(v)\le2n$ says $\kappa^2\ge1/(63504\sqrt2)>1.1\cdot10^{-5}$, while Proposition~\ref{prop:7.1}(c) says $\inf\kappa^2\le\frac23+O(n^{-1/2})$. The gap is a factor of about $6\cdot10^4$ and it is entirely in our proof. What the truth is, we do not know.

\begin{conjecture}[sharp constant]\label{conj:7.2}
 Does
\begin{equation*}
\inf_{\lambda\ne\mu\vdash n}\ \frac{\mathcal E\bigl(v_{\lambda,\mu}\bigr)\sqrt n}{\lVert v_{\lambda,\mu}\rVert_1^{\,2}}\ \longrightarrow\ \frac23\ ?
\end{equation*}
Equivalently, is the chirp family asymptotically extremal?
\end{conjecture}

This is a question about a constant, and it should not be confused with the inverse-energy inequality itself, which is Theorem~\ref{thm:6.2} and is proved. Nothing in Theorem~\ref{thm:1.1}, Theorem~\ref{thm:5.11} or Corollary~\ref{cor:5.12} is conditional on Conjecture~\ref{conj:7.2}; a positive answer would give $\rho(F_n)\le\bigl(6.79+o(1)\bigr)n^{1/4}$, since $3\sqrt{2+\sqrt2}\,/\sqrt{2/3}<6.79$, and a negative answer would change none of them. Being a statement about a limit it would not by itself replace the constant $1662$, which holds at every $n$; a uniform replacement needs the bound at every finite $n$ as well.

\section{Computational complexity}\label{sec:8}
Sections 3--5 used the \emph{geometry} of the feasible tables: which of them are optimal, how their cost decomposes, and what that decomposition says about distances. This section uses the same tables and the same identity of Theorem~\ref{thm:2.3}, but reads instead the \emph{direction and degree} of the optimization defined on them, and that is what decides tractability. Nothing here depends on \S\ref{sec:3}--\S\ref{sec:5}, and nothing there depends on this section; both are consequences of Theorem~\ref{thm:2.3}.

Two decision problems are considered, both in the explicit part-list representation (P1), with the encoding named in each statement.

\begin{itemize}
\item \textbf{CED}: given $\lambda,\mu\vdash n$ and $Q\in\mathbb Z$, is $q^*(\lambda,\mu)\le Q$?
\item \textbf{CRE}: given $\lambda,\mu\vdash n$ and $K\in\mathbb Z$, is $\max_{X\in T(\lambda,\mu)}\lVert X\rVert_F^2\ge K$?
\end{itemize}

By Theorem~\ref{thm:2.3} the two are interconvertible under the affine substitution $Q\leftrightarrow\frac{\lVert\lambda\rVert_2^2+\lVert\mu\rVert_2^2}2-K$, which suffices for the equivalence of the decision problems. It does \emph{not} transport approximation guarantees, and \S\ref{sec:8.3} accordingly treats the two optimization problems separately.

Throughout we keep the two categories apart: \textbf{CED} and \textbf{CRE} are decision problems, and are shown NP-complete; the associated evaluation problems (computing $q^*(\lambda,\mu)$, respectively computing $\max_X\lVert X\rVert_F^2$) are optimization problems, and are correspondingly NP-hard.

\textbf{A note on thresholds.} With the objective written as $\sum_{ij}g(x_{ij})$ instead, $\sum_{ij}x_{ij}=n$ gives $\sum_{ij}g(x_{ij})=\frac{\lVert X\rVert_F^2-n}2$, so a threshold $K$ in the Frobenius form corresponds to $K'=\frac{K-n}2$. The two forms define the same decision problem but their thresholds differ by a shift of $\frac n2$. We use the Frobenius form throughout. In this section the letter $C$ denotes a bin or block capacity in a packing instance; the block-energy metric $B(\lambda,\mu)$ of \S\ref{sec:3.2} does not appear.

\subsection{Strong NP-completeness of the decision problem}\label{sec:8.1}

\begin{theorem}\label{thm:8.1}
In representation (P1), CED and CRE are NP-complete under binary encoding, and remain NP-complete when the input numbers are bounded by a polynomial in the input length; that is, they are strongly NP-complete. This holds already when all blocks of $\lambda$ have equal size. Consequently computing $q^*(\lambda,\mu)$, and computing $\max_{X\in T(\lambda,\mu)}\lVert X\rVert_F^2$, are strongly NP-hard.
\end{theorem}

\begin{proof}
\textbf{Membership.} A certificate for CRE is a table $X\in T(\lambda,\mu)$. It has $k(\lambda)k(\mu)$ entries, each at most $n$ and so of $O(\log n)$ bits; since (P1) writes out every part, $k(\lambda)+k(\mu)=O(L)$, and the certificate has length $O(L^2\log n)$, polynomial in the input. Verification checks the margins and $\lVert X\rVert_F^2\ge K$ in time polynomial in that length. A certificate need only meet the threshold and never has to establish optimality. CED lies in NP by the affine substitution above.

This is the step that requires (P1). Under the compressed representation (P2) the same certificate can be exponentially longer than the input, and we claim no membership in NP there (\S\ref{sec:2.4}).

\textbf{Reduction.} We reduce from \textbf{3-Partition}, strongly NP-complete (\cite{ref8}, SP15): given $3m$ positive integers $a_1,\dots,a_{3m}$ with $\sum_ia_i=mC$ and $C/4<a_i<C/2$, decide whether they split into $m$ groups each summing to $C$; the bounds on $a_i$ force each group to have three elements. Put
\begin{equation*}
\lambda=(C^{\,m}),\qquad \mu=(a_1,\dots,a_{3m}),\qquad K=\sum_ia_i^2 ,
\end{equation*}
so both sides total $mC=:n$, $\lambda$ is uniform, and the construction is polynomial.

\textbf{The threshold is exactly attainable.} For any $X\in T(\lambda,\mu)$, column by column,
\begin{equation*}
\lVert X\rVert_F^2=\sum_j\sum_ix_{ij}^2\ \le\ \sum_j\Bigl(\sum_ix_{ij}\Bigr)^2=\sum_j\mu_j^2=K ,
\end{equation*}
with equality iff every column has at most one non-zero entry, that is, iff each $\mu$-block lies entirely inside one $\lambda$-block. Splitting some $a_i$ into $a'+a''$ with $a',a''>0$ strictly lowers that column's contribution from $a_i^2$ to $a'^2+a''^2$.

\textbf{Correctness.} If 3-Partition has a solution, place the three $\mu$-blocks of each group inside one $\lambda$-block; the group sums to $C$ and fills that row exactly, so every column has one non-zero entry and $\lVert X\rVert_F^2=K$. Conversely, attaining $K$ forces each $\mu$-block into a single $\lambda$-block, and the row sums then say that the $a_i$ in a common row sum to $C$, exhibiting $m$ groups.

\textbf{Strongness.} 3-Partition remains NP-hard when the integers $a_i$ are bounded by a polynomial in the number of items \cite{ref8}; on such instances $C\le\mathrm{poly}(m)$, so both $n=mC$ and $K\le n^2$ are bounded by a polynomial in the length of the constructed (P1) instance. Hence CRE, and with it CED, is NP-hard already on instances whose numbers are polynomially bounded in the input length. The statements about the evaluation problems follow, since an oracle for either value decides the corresponding threshold question.
\end{proof}

As discussed in \S\ref{sec:1.4}, the optimization problem $\max_X\lVert X\rVert_F^2$ was shown strongly NP-hard in \cite{ref7}; Theorem~\ref{thm:8.1} adds the identification with a graph metric and NP membership for the integer-threshold formulation, hence completeness for the decision version. The restriction to a uniform side is already implicit in the 3-Partition construction underlying \cite{ref7} and is recorded here for use in Corollary~\ref{cor:8.2}.

\subsection{Parameterized hardness}\label{sec:8.2}

We parameterize by $k=k(\lambda)$, the number of parts on the uniform side. The two encodings of \S\ref{sec:2.4} behave differently under this parameter, and both statements are needed to describe it.

Under \textbf{unary} encoding, fixing $k$ yields a dynamic program over the vector of residual row capacities, running in time $L^{O(k)}$: the problem is in XP for that parameter. Corollary~\ref{cor:8.2} shows that it is nevertheless not in FPT.

Under \textbf{binary} encoding that dynamic program is only pseudo-polynomial, since the capacity may be exponential in $L$, and no XP claim is available. Indeed the parameterization is already para-NP-hard at $k=2$: taking $\lambda=(C,C)$ and $\mu=(a_1,\dots,a_t)$ with $\sum_ja_j=2C$ and $K=\sum_ja_j^2$, the equality condition in the proof of Theorem~\ref{thm:8.1} makes $\max_X\lVert X\rVert_F^2=K$ equivalent to each $a_j$ lying in a single row, that is, to splitting the $a_j$ into two groups of equal sum. That is PARTITION, which is NP-complete under binary encoding \cite{ref8}. The remainder of this subsection therefore works under unary encoding.

\begin{corollary}\label{cor:8.2}
Parameterize CRE, and hence CED, by $k(\lambda)$, the number of parts of the uniform side $\lambda=(C^{\,k(\lambda)})$. In representation (P1) under unary encoding the problem is $\mathrm W[1]$-hard, so unless $\mathrm{FPT}=\mathrm W[1]$ there is no exact algorithm running in time $f(k(\lambda))\cdot L^{O(1)}$, where $L$ is the input length. Under ETH there is no exact algorithm running in time $f(k(\lambda))\cdot L^{o(k(\lambda)/\log k(\lambda))}$.
\end{corollary}

\begin{proof}
Jansen--Kratsch--Marx--Schlotter \cite{ref23} prove that Unary Bin Packing parameterized by the number of bins is $\mathrm W[1]$-hard (Theorem 2) and admits no $f(\varkappa)L^{o(\varkappa/\log\varkappa)}$ algorithm under ETH for $\varkappa$ bins (Theorem 3), and state that both persist for the \textbf{perfect-fill} case, where total item size equals total bin capacity. Given such an instance, with $\varkappa$ bins of capacity $C$ and items $a_1,\dots,a_t>0$ satisfying $\sum_ja_j=\varkappa C$, put
\begin{equation*}
\lambda=(C^{\,\varkappa}),\qquad\mu=(a_1,\dots,a_t),\qquad K=\sum_ja_j^2,\qquad n=\varkappa C .
\end{equation*}
Under unary encoding this is a polynomial-time construction of polynomially bounded length, and the parameter is preserved verbatim: $k(\lambda)=\varkappa$. By the equality condition in the proof of Theorem~\ref{thm:8.1}, $\max_X\lVert X\rVert_F^2=K$ iff each $a_j$ lies entirely in one row, which with the row sums means each bin is filled exactly; the converse is the same computation read backwards.
\end{proof}

Here $L$ is the input length under unary encoding, not the number of vertices; under that encoding $L=\Theta(n)$, whereas under other representations of \S\ref{sec:2.4} the quantity in the exponent would change, so the statement names $L$.

\subsection{Approximation consequences}\label{sec:8.3}

Affine equivalence of exact values does not preserve multiplicative approximation, so the two optimization problems are treated separately. Write
\begin{equation*}
\mathrm{MIN}\text{-}\mathrm{CED}:\ \text{compute }q^*(\lambda,\mu),\qquad \mathrm{MAX}\text{-}\mathrm{CRE}:\ \text{compute }\max_{X\in T(\lambda,\mu)}\lVert X\rVert_F^2 .
\end{equation*}

Both problems ask for an optimum over vertex alignments, and a solution must be returned in a form that can be written in time polynomial in $L$. An explicit permutation $\sigma\in S_n$ is not such a form: under binary encoding $n$ may be exponential in $L$, so $\Omega(n)$ output entries are unaffordable for reasons that have nothing to do with approximation. We therefore take a \textbf{transportation table} $X\in T(\lambda,\mu)$ as the returned object throughout this subsection. By Theorem~\ref{thm:2.3} a table represents a vertex alignment and its induced edit cost is $\operatorname{cost}(X)$, and under (P1) it has $k(\lambda)k(\mu)=O(L^2)$ cells of $O(\log n)$ bits each, hence is of size polynomial in $L$. (No such issue arises in Corollary~\ref{cor:3.12}, whose running time is measured in $n$ and which may therefore emit an explicit alignment.)

\begin{corollary}\label{cor:8.3}
In representation (P1) under binary encoding, unless $\mathrm P=\mathrm{NP}$ neither MIN-CED nor MAX-CRE admits an FPTAS. For each problem this rules out both standard formulations: the \textbf{feasible-solution} formulation, which in $\mathrm{poly}(L,1/\varepsilon)$ time must output a table $X\in T(\lambda,\mu)$ whose value ($\operatorname{cost}(X)$ for MIN-CED, $\lVert X\rVert_F^2$ for MAX-CRE) is within a factor $1+\varepsilon$ of the optimum; and the \textbf{value} formulation, which outputs only a number, one-sidedly or two-sidedly accurate to a factor $1+\varepsilon$.
\end{corollary}

\begin{proof}
Both arguments use the same three ingredients: an integral objective, an optimum bounded by a polynomial in $n$, and NP-hardness on instances where $n$ itself is polynomially bounded in the input length. The bounds and the rounding directions differ.

\textbf{MIN-CED.} First, $q^*(\lambda,\mu)=0$ exactly when $\lambda=\mu$, which is decidable in time polynomial in $L$ by sorting and comparing the two part lists; relative error degenerates at $0$, so assume $\lambda\ne\mu$ and hence $q^*\ge1$. For any $\sigma$ we have $E(G_\lambda)\triangle\sigma E(G_\mu)\subseteq E(K_n)$, so
\begin{equation*}
1\ \le\ q^*(\lambda,\mu)\ \le\ \binom n2 ,
\end{equation*}
the right bound being tight at $\lambda=(n)$, $\mu=(1^{\,n})$, where the symmetric difference is all of $E(K_n)$ independently of $\sigma$. Take $\varepsilon=\frac1{2n^2}$; then $\varepsilon q^*\le\frac1{2n^2}\binom n2<\frac14$. In the feasible-solution formulation, $A=\operatorname{cost}(X)$ is an integer with $q^*\le A<q^*+1$, so $A=q^*$; in the one-sided value formulation $\lfloor\widetilde A\rfloor=q^*$; in the two-sided formulation, rounding to the nearest integer gives $q^*$. The rounding direction is fixed by the formulation: flooring is correct one-sidedly, whereas two-sidedly it would return $q^*-1$ whenever $\widetilde A<q^*$.

\textbf{MAX-CRE.} Here no screening is needed. Every $X\in T(\lambda,\mu)$ has non-negative integer entries summing to $n$, so $x_{ij}^2\ge x_{ij}$ gives $\lVert X\rVert_F^2\ge n$, while $\lVert X\rVert_F^2\le(\sum_{ij}x_{ij})^2=n^2$; writing $\mathrm{OPT}$ for the maximum,
\begin{equation*}
n\ \le\ \mathrm{OPT}\ \le\ n^2 ,
\end{equation*}
and $\mathrm{OPT}$ is an integer. Take $\varepsilon=\frac1{2n^2}$. In the feasible-solution formulation a table $X$ with $\lVert X\rVert_F^2\ge(1+\varepsilon)^{-1}\mathrm{OPT}$ satisfies
\begin{equation*}
\mathrm{OPT}-\lVert X\rVert_F^2\ \le\ \mathrm{OPT}\cdot\frac{\varepsilon}{1+\varepsilon}\ \le\ \varepsilon\,\mathrm{OPT}\ \le\ \frac{n^2}{2n^2}=\frac12\ <\ 1 ,
\end{equation*}
and both quantities are integers, so $\lVert X\rVert_F^2=\mathrm{OPT}$. In the one-sided value formulation, $(1-\varepsilon)\mathrm{OPT}\le\widetilde A\le\mathrm{OPT}$ with the same $\varepsilon$ gives $\lceil\widetilde A\rceil=\mathrm{OPT}$, the ceiling being the correct direction for a maximization problem, and two-sidedly, rounding to the nearest integer.

\textbf{Restriction to the hard instances.} In both cases the accuracy used is $\varepsilon=\frac1{2n^2}$, and $n=\sum_i\lambda_i$ is computable from the input, so $\varepsilon$ is available to the algorithm. It remains to see that $1/\varepsilon$ is polynomial in $L$ where it needs to be. An FPTAS is required to meet its guarantee on \emph{every} input, so it may be run on any subfamily; take the family produced by the reduction of Theorem~\ref{thm:8.1} from 3-Partition instances whose integers are bounded by a polynomial in the number of items. On that family $C\le\mathrm{poly}(m)$ and hence
\begin{equation*}
n=mC\le\mathrm{poly}(m)\le\mathrm{poly}(L),
\end{equation*}
since (P1) writes out all $m+3m$ parts. Therefore $1/\varepsilon=2n^2\le\mathrm{poly}(L)$, running the hypothetical FPTAS at this $\varepsilon$ takes $\mathrm{poly}(L,1/\varepsilon)=\mathrm{poly}(L)$ time, and by the rounding above it returns the exact optimum. That decides CED, respectively CRE, in polynomial time on a family on which Theorem~\ref{thm:8.1} shows it to be NP-hard, so $\mathrm P=\mathrm{NP}$.
\end{proof}

The last step is what strong NP-hardness buys. The obstruction to running the argument on an arbitrary binary instance is real: there $L$ may be as small as $O(\log n)$, making $1/\varepsilon=2n^2$ exponential in $L$. But the conclusion does not need arbitrary instances. Since the hypothesised FPTAS must work everywhere, it suffices to defeat it on the polynomially bounded subfamily, where $n\le\mathrm{poly}(L)$ and the accuracy is affordable. Corollary~\ref{cor:8.3} therefore applies to the standard binary encoding of (P1) and needs no unary hypothesis.

Corollary~\ref{cor:8.3} excludes an FPTAS only. It does not exclude a PTAS and gives no APX-hardness: the reduction of \S\ref{sec:8.1} has zero gap and supports no gap amplification. The constant-factor algorithm of Corollary~\ref{cor:3.12} is consistent with it, a constant-factor approximation not being an approximation scheme.

We also record a nearby tractable framework that does not apply. Integer maximization of $c(w_1x,\dots,w_dx)$ with $c$ convex and $d$ \textbf{fixed} is polynomially solvable by Graver-basis methods \cite{ref33}. Our objective $\sum_{ij}x_{ij}^2$ is convex on all cells, but its number of arguments grows with the instance, so Theorem~\ref{thm:8.1} is consistent with that line.

\subsection{The nearest/farthest dichotomy}\label{sec:8.4}

Reversing the optimization in the definition gives the \textbf{farthest alignment}
\begin{equation*}
q^{\max}(\lambda,\mu)=\max_{\sigma\in S_n}\bigl|E(G_\lambda)\triangle\sigma E(G_\mu)\bigr| .
\end{equation*}

\begin{proposition}\label{prop:8.4}
$q^{\max}$ is computable exactly in time polynomial in the binary input length:
\begin{equation*}
q^{\max}(\lambda,\mu)=\tfrac12\bigl(\lVert\lambda\rVert_2^2+\lVert\mu\rVert_2^2\bigr)-\min_{X\in T(\lambda,\mu)}\lVert X\rVert_F^2 ,
\end{equation*}
the minimum being a minimum-cost flow with separable convex arc costs on a network with $k_\lambda+k_\mu+2$ nodes, $k_\lambda k_\mu+k_\lambda+k_\mu$ arcs and every capacity at most $n$, where $k_\lambda=k(\lambda)$ and $k_\mu=k(\mu)$.
\end{proposition}

\begin{proof}
The identity is the derivation of Theorem~\ref{thm:2.3} with $\min_\sigma$ replaced by $\max_\sigma$. Now $T(\lambda,\mu)$ is the set of integral feasible flows on the bipartite transportation network with row sums $\lambda$ and column sums $\mu$, the flow on arc $(i,j)$ being $x_{ij}$, and the objective $\sum_{ij}x_{ij}^2$ is separable across arcs and convex on each. Minoux \cite{ref24} solves minimum-cost flow with separable convex arc costs on integer networks in time polynomial in the size of the network and in the logarithm of the largest capacity; here there are $k_\lambda+k_\mu+2$ nodes, $k_\lambda k_\mu+k_\lambda+k_\mu$ arcs, and every capacity is at most $n$, so the running time is polynomial in the input length under either encoding. Only polynomiality is used; we do not restate the operation count of \cite{ref24}, whose arithmetic model we have not re-derived.
\end{proof}

The transportation representation therefore classifies both directions at once:
\begin{equation*}
\begin{aligned}
\text{deciding whether }q^*(\lambda,\mu)\le Q&&&\textbf{strongly NP-complete},\\
\text{computing }q^*(\lambda,\mu)\ \ (\Leftrightarrow\ \max\nolimits_X\lVert X\rVert_F^2)&&&\textbf{strongly NP-hard},\\
\text{computing }q^{\max}(\lambda,\mu)\ \ (\Leftrightarrow\ \min\nolimits_X\lVert X\rVert_F^2)&&&\textbf{polynomial-time solvable}.
\end{aligned}
\end{equation*}
Proposition~\ref{prop:8.4} computes a \emph{minimum} of transportation energy, which on the graph side is a \emph{maximum} of edit distance; it is not an algorithm for $q^*$. The same applies to the other tractable cases of quadratic transportation: the strongly polynomial algorithm of Cosares--Hochbaum for a fixed number of sources \cite{ref25}, and the special non-separable model of Hochbaum--Shamir--Shanthikumar \cite{ref26}. Both are also on the minimization side.

\section{Conclusions and open problems}\label{sec:9}

For cluster graphs the Euclidean distortion of the class is determined up to universal constants, and it is attained by an explicit map. The picture closes on one diagram: the transportation representation of Theorem~\ref{thm:2.3} puts the metric within a factor $3$ of the $\ell_1$ norm of a Ferrers staircase difference, by Theorem~\ref{thm:3.1}; the difference is arithmetically rigid, every jump at position $s$ being a multiple of $s$; that rigidity forces a lower bound on the critical dyadic energy of the difference, valid on the whole lattice of closed quantized sequences; and the resulting map $F_n$, of dimension below $4n$ and computable in $O(n)$ time from a single partition, has distortion $\Theta(n^{1/4})$, which a Hamming cube of partitions and Enflo's theorem show to be optimal for every Hilbert embedding.

Read combinatorially rather than geometrically, the same representation produces two explicit $\ell_1$ models. The vertex-mass metric is equivalent to $q^*$ within the optimal factor $3$; the block-energy metric within the optimal factor $2$, whence $c_1(\mathcal K_n)\le2$ together with an $O(n\log n)$-time alignment of cost strictly below $2q^*$ carrying a two-sided certificate. Read computationally, it shows that the decision problem is strongly NP-complete, that evaluation is strongly NP-hard with no FPTAS, and that the farthest alignment is polynomial.

The mechanism behind the embedding is worth stating separately from the application, because it is not an analytic estimate. A realizable staircase difference can only jump at position $s$ by a multiple of $s$, so oscillation is priced by where it happens; cheap wide wobble is confined to small positions, and at large positions every change of direction is spike-scale. Theorem~\ref{thm:6.2} is the quantitative form of that sentence, and it holds for every closed quantized sequence, with no reference to partitions. Its most portable component is the folded-remainder capital lemma (Proposition~\ref{prop:6.9}), which compresses all of the mass below the half-line $|v(s)|=\frac{s-1}2$ onto the mass above it, at the cost of a factor $8$.

Three features of the picture could have come out otherwise. The leading constant in Proposition~\ref{prop:4.4} is pinned rather than estimated, because three separate inequalities saturate on one family. The two $\ell_1$ models are genuinely different maps with different optimal constants, $3$ and $2$, rather than reparametrizations of one another. And the direction dichotomy of \S\ref{sec:8.4} is easy to misread in the transportation literature: the known polynomial algorithms minimize transportation energy and therefore compute the farthest alignment, not the nearest.

\textbf{Open problems.}

\begin{enumerate}
\item \textbf{The sharp constant.} Is Conjecture~\ref{conj:7.2} true, that is, is the chirp family asymptotically extremal with $\inf\kappa^2\to\frac23$? What Theorem~\ref{thm:6.2} proves is $\inf\kappa^2\ge1/(63504\sqrt2)$, some $6\cdot10^4$ below. Reducing the constant to anything of moderate size would already make Theorem~\ref{thm:5.11} a usable quantitative statement rather than an order; the four modules of \S\ref{sec:6} each contribute, and none of them has been optimized.
\item \textbf{The value of $c_1(\mathcal K_n)$.} Corollary~\ref{cor:3.13} gives $c_1(\mathcal K_n)\le2$, and Proposition~\ref{prop:3.11} shows that $2$ is optimal for the block-energy map, as $3$ is for the sorted-degree map (Theorem~\ref{thm:3.1}(F)). What is the infimum over all embeddings into $\ell_1$? No lower bound above $1$ is known to us.
\item \textbf{PTAS or APX-hardness.} Corollary~\ref{cor:8.3} excludes an FPTAS but not a PTAS, and \S\ref{sec:8.1} has zero gap, so neither direction has evidence. The distance between the ratio $2$ of Corollary~\ref{cor:3.12} and the best possible ratio is also unknown.
\end{enumerate}

\textbf{Further directions.} Four questions are natural and we record them without elaboration. The proof of Theorem~\ref{thm:6.2} uses the divisibility only through the implication ``$\Delta v(s)\ne0\Rightarrow|\Delta v(s)|\ge s$''; the same four modules survive on any such lattice where jumps at position $s$ are constrained to $w(s)\mathbb Z$ for a non-decreasing weight $w$, of which the partition metric is the case $w(s)=s$, and identifying the resulting exponent as a function of $w$ would be a genuine generalization. Conjecture~\ref{conj:7.2} asks only for one constant; one may ask instead for the whole extremal profile, the least critical energy compatible with prescribed $\lVert v\rVert_1$ and $\mathrm{TV}(v)$, of which Proposition~\ref{prop:7.1} computes one point. Theorem~\ref{thm:5.8} rules out one specific coordinate map at the order $n^{1/4}\sqrt{\log n}$; whether \emph{every} map linear in $u_\lambda$ is stuck there, which would show that the weighting in $F_n$ is not merely convenient but necessary, we do not know. Finally, the corresponding questions for cographs and trees are open; the methods here rest on the correspondence between block structure and contingency tables, which those classes do not have.

\section*{Declarations}

\textbf{Data and code availability.} This paper reports no experimental data. The enumeration code used for the cross-checks of Appendix B is provided as supplementary material, with running instructions and version information; it is used nowhere in any proof.

\textbf{Funding.} This research received no external funding.

\textbf{Competing interests.} The authors declare that they have no competing interests.

\textbf{Declaration of generative AI in the writing process.} During the preparation of this work the authors used a large language model in order to reorganize the manuscript and to improve its English expression. After using this tool the authors reviewed and edited the content as needed and take full responsibility for the content of the publication. No mathematical statement, proof, constant, numerical value or reference in this paper was generated without author verification. Spelling and basic grammar checking alone would require no declaration; substantive assistance with organization or expression does.


\appendix
\section{An exact Haar identity for the dyadic energy}\label{sec:A}

This appendix identifies the functional $\lVert F^{(0)}(\lambda)-F^{(0)}(\mu)\rVert_2^2$ exactly in Haar coordinates. Nothing in the body depends on it; we include it because it explains, in one line, why a spike costs $\sqrt{\log n}$ and where the mean term of a realizable difference comes from.

For a dyadic interval $I$ of level $m\ge1$ with left and right halves $I_L,I_R$, let $h_I=\mathbf 1_{I_L}-\mathbf 1_{I_R}$ be the unnormalized Haar function, and put $H_m(w)=\sum_{\ell(I)=m}\langle w,h_I\rangle^2$.

\begin{proposition}\label{prop:A.1}
For every $w\in\mathbb R^N$ with $N=2^J$,
\begin{equation*}
\sum_{\ell=0}^JE_\ell(w)\;=\;\bigl(2-2^{-J}\bigr)\,\langle w,\mathbf 1\rangle^2\;+\;\sum_{m=1}^J\bigl(1-2^{-m}\bigr)\,H_m(w).
\end{equation*}
\end{proposition}

\begin{proof}
For each level-$(\ell+1)$ interval $I$ the parallelogram identity $a^2+b^2=\frac12\bigl((a+b)^2+(a-b)^2\bigr)$ with $a=\langle w,\mathbf 1_{I_L}\rangle$, $b=\langle w,\mathbf 1_{I_R}\rangle$ gives $\langle w,\mathbf 1_{I_L}\rangle^2+\langle w,\mathbf 1_{I_R}\rangle^2=\frac12\langle w,\mathbf 1_I\rangle^2+\frac12\langle w,h_I\rangle^2$. Summing over the level-$(\ell+1)$ intervals, $E_\ell=\frac12E_{\ell+1}+\frac12H_{\ell+1}$ for $0\le\ell\le J-1$. Unrolling downwards from $E_J=\langle w,\mathbf 1\rangle^2$ gives $E_\ell=2^{\ell-J}E_J+\sum_{m=\ell+1}^J2^{\ell-m}H_m$. Summing over $0\le\ell\le J$, the coefficient of $E_J$ is $\sum_{\ell=0}^J2^{\ell-J}=2-2^{-J}$ and that of $H_m$ is $\sum_{\ell=0}^{m-1}2^{\ell-m}=1-2^{-m}$.
\end{proof}

\emph{Check.} For $w=\delta_s$ the left side is $J+1$, and the right side is $(2-2^{-J})+\sum_{m\ge1}(1-2^{-m})=J+1$.

\begin{corollary}\label{cor:A.2}
\emph{(i)} $\ \frac32\langle w,\mathbf 1\rangle^2+\frac12\sum_mH_m(w)\ \le\ \sum_\ell E_\ell(w)\ \le\ 2\langle w,\mathbf 1\rangle^2+\sum_mH_m(w)$ for $J\ge1$.

\emph{(ii)} For a realizable difference $w=v_{\lambda,\mu}$ the mean term is the squared \textbf{energy gap}:
\begin{equation*}
\langle v,\mathbf 1\rangle=\sum_su_\lambda(s)-\sum_su_\mu(s)=\lVert\lambda\rVert_2^2-\lVert\mu\rVert_2^2 .
\end{equation*}
\end{corollary}

\begin{proof}
(i) Bound the coefficients of Proposition A.1 by their extreme values. (ii) $\sum_su_\lambda(s)=\sum_i\sum_{s\le\lambda_i}\lambda_i=\sum_i\lambda_i^2$.
\end{proof}

\begin{remark}[what the identity says]\label{rem:A.3}
In normalized Haar coordinates $c_I=2^{-m/2}\langle w,h_I\rangle$ one has $H_m=2^m\sum_{\ell(I)=m}c_I^2$, so up to the bounded factors of Proposition A.1 the map $F^{(0)}$ is the square function of the homogeneous norm $\bigl(\langle w,\mathbf 1\rangle^2+\sum_m2^m\lVert c_{(m)}\rVert_2^2\bigr)^{1/2}$, of smoothness index $-\frac12$. A point mass has logarithmically divergent mass in that norm across scales, which is the $\sqrt{\log n}$ of the spike of Lemma~\ref{lem:5.6} seen analytically, and the breathing train of Lemma~\ref{lem:5.7} is the corresponding oscillation-at-fixed-scale extremal. The weight $2^{-\ell/2}$ of \S\ref{sec:6} moves the index to $-\frac14$. We use this reading only as orientation; as \S\ref{sec:6.1} records, the substance of \S\ref{sec:6} is arithmetic rather than analytic, and no statement in the body depends on this appendix.
\end{remark}

\section{Computational cross-checks}\label{sec:B}

Nothing here carries any part of the burden of proof: every theorem in the body is a self-contained universally quantified argument, and the enumerations below were run only as cross-checks against transcription errors. All are exhaustive over the stated ranges unless marked otherwise. Exact integer arithmetic is used for $q^*$, $B$ and $\operatorname{cost}$; the dyadic scans use IEEE doubles, with all intermediate values below $2^{47}$ and hence exactly represented. The values of $q^*$ were produced by two independently written solvers and agree throughout.

\textbf{B.1 The two $\ell_1$ models.}

\begin{small}
\begin{longtable}{>{\raggedright\arraybackslash}p{0.34\linewidth}>{\raggedright\arraybackslash}p{0.24\linewidth}>{\raggedright\arraybackslash}p{0.34\linewidth}}
\toprule
\textbf{Item checked} & \textbf{Range} & \textbf{Outcome} \\
\midrule\endfirsthead
\toprule
\textbf{Item checked} & \textbf{Range} & \textbf{Outcome} \\
\midrule\endhead
the two single-step formulas of Theorem~\ref{thm:3.1}(A) & all legal single steps, $n\le18$; 71244 steps & no counterexample \\
the two inequalities of Theorem~\ref{thm:3.1}(B) & all pairs, $n\le18$; 180124 pairs & no counterexample \\
the energy sandwich $q^*\le B\le2q^*-1$ of Corollary~\ref{cor:3.10} & all pairs, $n\le18$; 180124 pairs & no violation; $B=q^*$ on 70840 pairs \\
the per-$n$ maximum of $B/q^*$ & $2\le n\le18$ & equals $2-1/\lfloor n/2\rfloor$ throughout, always attained by $\bigl((2k),(k,k)\bigr)$, padded with a part $1$ for odd $n$ \\
the per-table bound of Theorem~\ref{thm:3.9}, over \textbf{every} table & all pairs and all $X\in T(\lambda,\mu)$, $n\le8$; 108592 tables & no violation; equality in 45346 cases; $\min_X\operatorname{cost}=q^*$ throughout \\
the witness family of Proposition~\ref{prop:3.11} & $k\le60$ & $q^*=k^2$, $B=2k^2-k$, Theorem~\ref{thm:3.9} an equality \\
the reduction identity of Lemma~\ref{lem:3.2} & $n\le9$; 1716 combinations & no counterexample \\
exhaustiveness of the endpoints in Lemma~\ref{lem:3.3} & $p,q\le13$; 182 combinations & $\max Q$ agrees with the endpoint value throughout \\
the residue examples of \S\ref{sec:3.3} & $(a,b,n)=(5,3,15),(7,4,28),(8,3,24)$ & true maxima $37,94,64$ against row-wise bounds $39,100,66$ \\
\bottomrule
\end{longtable}
\end{small}

\textbf{B.2 The sorted-degree coordinate.}

\begin{small}
\begin{longtable}{>{\raggedright\arraybackslash}p{0.34\linewidth}>{\raggedright\arraybackslash}p{0.24\linewidth}>{\raggedright\arraybackslash}p{0.34\linewidth}}
\toprule
\textbf{Item checked} & \textbf{Range} & \textbf{Outcome} \\
\midrule\endfirsthead
\toprule
\textbf{Item checked} & \textbf{Range} & \textbf{Outcome} \\
\midrule\endhead
the cancellation of Proposition~\ref{prop:4.4}(A) & $n\le9$; 861 pairs, exact rational comparison & identity throughout \\
$L^+=\sqrt2$ of Lemma~\ref{lem:4.5} & $n\le9$ & attained for each $n$, never exceeded \\
the padding invariance of Lemma~\ref{lem:4.7} & $k=2,3$, $n\le13$; 10 configurations & $q^*$ and $\lVert z\rVert_2$ unchanged \\
the sandwich of Proposition~\ref{prop:4.4}(F) & $n=6,\dots,10^5$ & holds throughout \\
\bottomrule
\end{longtable}
\end{small}

\textbf{B.3 The multiscale family.} Here $\kappa=\sqrt{\mathcal E(v)}\,n^{1/4}/\lVert v\rVert_1$ as in \S\ref{sec:7}, and $\kappa_\gamma=\lVert F^{(\gamma)}(\lambda)-F^{(\gamma)}(\mu)\rVert_2\,n^{1/4}/\lVert v\rVert_1$ is the same ratio at a general weight, so that $\kappa=\kappa_{1/4}$.

\begin{small}
\begin{longtable}{>{\raggedright\arraybackslash}p{0.34\linewidth}>{\raggedright\arraybackslash}p{0.24\linewidth}>{\raggedright\arraybackslash}p{0.34\linewidth}}
\toprule
\textbf{Item checked} & \textbf{Range} & \textbf{Outcome} \\
\midrule\endfirsthead
\toprule
\textbf{Item checked} & \textbf{Range} & \textbf{Outcome} \\
\midrule\endhead
quantization and the budget $\mathrm{TV}\le2n$ (Lemma~\ref{lem:2.6}), and the three counts of Lemma C.1 & all pairs $4\le n\le18$; 800 sampled pairs at $n=60,120$ & no violation \\
the per-run bound of Proposition~\ref{prop:6.6}, as $7\mathcal E/\sum_R\mu_R^2W_R^{-1/2}\ge1$ & as above, plus all named families & minimum $5.0$ (chirp, $m=102$): a margin of $5\times$ over the stated constant \\
the octave bound of Proposition~\ref{prop:6.14}, as $14\bigl(\sum_i\sqrt{\tau_i}\bigr)\mathcal E/\lVert v\rVert_1^2\ge1$ & as above & minimum $24.6$: a margin of $25\times$ \\
the scaling of Theorem~\ref{thm:5.8}, breathing pairs & $n=2^{10},\dots,2^{18}$ & $\lVert F^{(0)}v\rVert_2\,n^{1/4}/\lVert v\rVert_1\in[2.83,3.37]$, bounded oscillation with no logarithmic drift; the level-$0$ share of the energy rises to $0.9993$ \\
the scaling of Theorem~\ref{thm:5.8}, spike pairs & $n=2^{10},\dots,2^{18}$ & $\bigl(\lVert F^{(0)}w\rVert_2/\lVert w\rVert_1\bigr)/\sqrt{J+1}\to0.61$, confirming expansion of order $\sqrt{\log n}$ \\
$\min_{\lambda\ne\mu\vdash n}\kappa$, exhaustive (up to $1.57\times10^7$ pairs per $n$) & $6\le n\le30$ & values in $[0.887,1.159]$, not monotone in $n$; every minimizer is an interleaved doubled ladder with alternating ownership, and the record lows fall at $n=12,20,30$, where a ladder fits exactly: $(6,3,3)$ against $(5,5,1,1)$, $(8,5,5,1,1)$ against $(7,7,3,3)$, and $(10,7,7,3,3)$ against $(9,9,5,5,1,1)$. The first and third are the chirp pairs of \S\ref{sec:7} at $m=6,10$; the second is the $m=8$ ladder, outside the residue class $m\equiv2\pmod4$ used there. At the remaining $n$ the minimizer is a truncated or unit-padded variant of the same ladder \\
chirp calibration against $\sqrt{2/3}=0.8165$ (Proposition~\ref{prop:7.1}) & $m\approx2\sqrt n$, $n=2^{10},\dots,2^{18}$ & $\kappa$ falls from $0.8303$ to $0.8173$; energy shares $E_0:E_1:E_{\ge2}=1-O(10^{-4}):O(10^{-4}):O(10^{-5})$, matching Proposition~\ref{prop:7.1}(b) \\
adversarial families at scale: dipole trains of width $W\in\{2,\dots,256\}$; multiscale cascades over up to $7$ geometric scales & $n=2^{16}$, and $n=2^{12},\dots,2^{18}$ at fixed $W$ & $\kappa\in[1.98,2.83]$ and flat in $n$ at fixed $W$; cascades give $\kappa\ge1.59$, non-monotone in the number of scales \\
the upper Lipschitz bound $\lVert F^{(1/4)}v\rVert_2\le1.848\lVert v\rVert_1$ of Proposition~\ref{prop:5.2} & 300 random pairs plus all named families & maximum observed ratio $1.108$ \\
criticality of the exponent (Proposition~\ref{prop:5.9}), on $(n)$ against $(1^n)$ & $n=2^{10},\dots,2^{18}$ & the level-$0$-only functional decays like $n^{-1/4}$; $\kappa_{1/4}\to1.8465\approx\sqrt{2+\sqrt2}$, flat; $\kappa_{0.35}\cdot n^{0.1}$ flat, so the measured decay exponent equals $\gamma-\frac14$ \\
\bottomrule
\end{longtable}
\end{small}

\textbf{B.4 The quantized inverse-energy theorem.} Every inequality of \S\ref{sec:6} is checked
separately, on three sources: every pair of partitions of $n$, which is the realizable
family; random closed quantized sequences, which are not staircase differences and so test
the theorem in the generality it is stated in; and an exhaustive enumeration of the closed
quantized sequences with $|v(s)|\le2s$, which is the regime in which the extremal ratio of
Proposition~\ref{prop:6.9} lives. The thick regime $L\ge252\,\mathrm{TV}(v)$ is out of reach at these sizes -- the support
of $v$ lies in $[2,\mathrm{TV}(v)]$ and $\max_s|v(s)|\le\mathrm{TV}(v)/2$ by Lemma~\ref{lem:6.3}, so
$L\le\mathrm{TV}(v)^2/2$ and the regime needs $n\ge252$ -- and it is therefore checked
separately on the classical pair.

\begin{small}
\begin{longtable}{>{\raggedright\arraybackslash}p{0.34\linewidth}>{\raggedright\arraybackslash}p{0.24\linewidth}>{\raggedright\arraybackslash}p{0.34\linewidth}}
\toprule
\textbf{Item checked} & \textbf{Range} & \textbf{Outcome} \\
\midrule\endfirsthead
\toprule
\textbf{Item checked} & \textbf{Range} & \textbf{Outcome} \\
\midrule\endhead
every inequality of \S\ref{sec:6.2}--\S\ref{sec:6.6} simultaneously & all pairs $\lambda\ne\mu\vdash n$, $2\le n\le14$; 20283 pairs & no violation \\
the same, on sequences outside the realizable family & random closed quantized sequences, $N\in\{4,8,16,32,64\}$; 1995 sequences & no violation \\
the Whitney tiling of Lemma~\ref{lem:6.4}, over every subinterval of every run & as above & partition and two-per-scale hold throughout; $\sup\sum_K\lvert K\rvert^{1/2}\lvert I\rvert^{-1/2}=2.856$ against the proved $6.829$ \\
the capital lemma $M_{\rm low}\le8M_{\rm cap}$ of Proposition~\ref{prop:6.9}, exhaustively & all closed quantized sequences with $\lvert v(s)\rvert\le2s$, $N\le13$; 53808492 sequences & no violation; the maximum of $M_{\rm low}/M_{\rm cap}$ equals $(N-4)/(N+2)$ at every $N$ in range, hence $0.600$ at $N=13$, against the proved $8$ \\
Theorem~\ref{thm:6.2} in the thick regime & $(n)$ against $(1^{\,n})$, $n=600,\,2000,\,20000$ & holds with $63504\,\mathcal E\sqrt{\mathrm{TV}}/\lVert v\rVert_1^2\approx3.0\times10^5$ throughout \\
\bottomrule
\end{longtable}
\end{small}

The margins over the constants actually proved, minimized over all of the above, are:
$7\mathcal E/\sum_R\mu_R^2W_R^{-1/2}\ge6.71$ for Proposition~\ref{prop:6.6};
$14P\sqrt T/M_{\rm cap}^{\rm clean}\ge20.8$ for Lemma~\ref{lem:6.8}, over $P/\sqrt T\in\{1,1.3,2,4\}$;
$35\,\mathcal E/L^{3/2}\ge40.7$ for Lemma~\ref{lem:6.13}; and
$63504\,\mathcal E\sqrt{\mathrm{TV}(v)}/L^2\ge8.5\times10^4$ for Theorem~\ref{thm:6.2} itself.
The last figure is the numerical form of the discussion in \S\ref{sec:6.7}: the constant proved is
some five orders of magnitude smaller than anything the families reach, and the loss is in
the bookkeeping rather than in the statement.

The enumeration code, with running instructions and version information, is provided as supplementary material; without it these entries are not independently verifiable, and they are relied upon nowhere in the paper.

\section{Two deferred proofs}\label{sec:C}

Neither proof below is used by the main line. The first supplies Proposition~\ref{prop:6.14}, a bound with a far better constant than Theorem~\ref{thm:6.2} but blind on flat octave spectra; it was the estimate obtained in the first version of this work, and it is recorded because it remains the sharper of the two off that configuration. The second is a spectral restatement of Corollary~\ref{cor:3.10}, of independent interest but used nowhere.

\subsection{The octave bound}\label{sec:C.1}

Keep the notation of \S\ref{sec:6}. For $0\le i\le J$ put
\begin{equation*}
\text{octave }i=\bigl[2^i,\,2^{i+1}\bigr)\cap[N],\qquad \tau_i=\sum_{s\in[2^i,2^{i+1})}\bigl|\Delta v(s)\bigr| ,
\end{equation*}
so that $\sum_i\tau_i=\mathrm{TV}(v)$, and assign each run $R$ to the octave containing its right endpoint $b_R$.

\begin{lemma}[octave counting]\label{lem:C.1}
Let $K_i$ be the number of runs assigned to octave $i$ and $\Omega_i$ the total of their widths. Then
\begin{equation*}
\text{(1) } K_i\le\tau_i\,2^{-i},\qquad \text{(2) } \Omega_i\le2^{\,i+1},\qquad \text{(3) }\sum_{R\text{ in octave }i}\sqrt{W_R}\ \le\ \sqrt{2\tau_i} .
\end{equation*}
\end{lemma}

\begin{proof}
(1) A run $R$ is maximal, so $v(b_R+1)$ is zero or of the opposite sign to $v(b_R)\ne0$; hence $\Delta v(b_R)\ne0$ and, by (6.2), $|\Delta v(b_R)|\ge b_R\ge2^i$. Distinct runs have distinct right endpoints, all lying in octave $i$, so $\tau_i\ge K_i2^i$.

(2) Every run assigned to octave $i$ ends before $2^{i+1}$, hence is contained in $[1,2^{i+1})$, and runs are pairwise disjoint.

(3) By Cauchy--Schwarz and then (1) and (2), $\sum\sqrt{W_R}\le\sqrt{K_i\Omega_i}\le\sqrt{\tau_i2^{-i}\cdot2^{i+1}}=\sqrt{2\tau_i}$.
\end{proof}

\begin{proof}[Proof of Proposition~\ref{prop:6.14}]
Write $\Lambda_i=\sum_{R\text{ in octave }i}\mu_R$ and $Q_i=\sum_{R\text{ in octave }i}\mu_R^2/\sqrt{W_R}$. Within octave $i$, Cauchy--Schwarz and Lemma C.1(3) give
\begin{equation*}
\Lambda_i^2\ \le\ Q_i\sum_{R\text{ in octave }i}\sqrt{W_R}\ \le\ Q_i\sqrt{2\tau_i},\qquad\text{that is}\qquad Q_i\ \ge\ \frac{\Lambda_i^2}{\sqrt{2\tau_i}} .
\end{equation*}
Across octaves, Cauchy--Schwarz again, followed by this bound and then Proposition~\ref{prop:6.6} applied to the family of all runs:
\begin{equation*}
\lVert v\rVert_1^2=\Bigl(\sum_i\Lambda_i\Bigr)^2\le\Bigl(\sum_i\frac{\Lambda_i^2}{\sqrt{\tau_i}}\Bigr)\Bigl(\sum_i\sqrt{\tau_i}\Bigr)\le\sqrt2\Bigl(\sum_iQ_i\Bigr)\Bigl(\sum_i\sqrt{\tau_i}\Bigr)\le7\sqrt2\;\mathcal E(v)\sum_i\sqrt{\tau_i},
\end{equation*}
where the sums range over octaves with $\tau_i>0$, an octave with $\tau_i=0$ containing no run by Lemma C.1(1). Since $7\sqrt2<14$, this is (6.21). The two consequences follow by substituting the hypothesis, respectively by Cauchy--Schwarz over the at most $J+1$ octaves.
\end{proof}

Applied to a realizable pair with $\mathrm{TV}(v)\le2n$ (Lemma~\ref{lem:2.6}), the last form gives $\mathcal E(v)\ge\lVert v\rVert_1^2/\bigl(14\sqrt{2n(\log_2n+2)}\bigr)$ and hence $\rho(F^{(1/4)})\le25\,n^{1/4}(\log_2n+2)^{1/4}$, which is the bound Theorem~\ref{thm:5.11} improves to $O(n^{1/4})$.

\subsection{A spectral reformulation of the block-energy model}\label{sec:C.2}

The block-energy vector is a spectrum, which gives Corollary~\ref{cor:3.10} a reading in which the combinatorial alignment is relaxed to a unitary one. Let $A_\lambda$ be the adjacency matrix of $G_\lambda$, let $Q_\lambda:=\frac12(A_\lambda^2+A_\lambda)$ be its \textbf{clique-energy operator}, and write $\lVert\cdot\rVert_{S_1}$ for the trace norm.

\begin{proposition}\label{prop:C.2}
$\operatorname{spec}(Q_\lambda)=\{g(\lambda_1),\dots,g(\lambda_{k_\lambda})\}\cup\{0^{\,n-k_\lambda}\}$, and
\begin{equation*}
\min_{U\in\mathrm U(n)}\bigl\lVert Q_\lambda-UQ_\mu U^*\bigr\rVert_{S_1}\;=\;B(\lambda,\mu) .
\end{equation*}
Consequently $q^*(\lambda,\mu)\le\min_U\lVert Q_\lambda-UQ_\mu U^*\rVert_{S_1}\le2\,q^*(\lambda,\mu)$ by Corollary~\ref{cor:3.10}.
\end{proposition}

\begin{proof}
$A_\lambda=\bigoplus_i(J_{\lambda_i}-I)$, and $J_t-I$ has eigenvalues $t-1$ once and $-1$ with multiplicity $t-1$. The polynomial $x\mapsto\frac{x^2+x}2$ sends $t-1\mapsto g(t)$ and $-1\mapsto0$, giving the spectrum. For Hermitian $P,P'$, Mirsky's inequality \cite{ref32} gives $\lVert P-P'\rVert_{S_1}\ge\lVert\operatorname{spec}^\downarrow(P)-\operatorname{spec}^\downarrow(P')\rVert_1$; since $UQ_\mu U^*$ has the spectrum of $Q_\mu$, the minimum is at least the sorted spectral $\ell_1$ distance, which is $B(\lambda,\mu)$ because both spectra sorted descending are the padded vectors $e(\cdot)$. Choosing $U$ to co-diagonalize in matching order attains it.
\end{proof}

Relaxing the combinatorial alignment to a unitary one therefore changes the distance by at most a factor $2$, and the relaxed quantity is available in closed form in $O(n\log n)$ time, whereas $q^*$ is strongly NP-hard to compute by Theorem~\ref{thm:8.1}. Proposition C.2 is used nowhere else. It is recorded because it explains why $2$, rather than some smaller constant, is what a spectral method can see: the trace-norm minimum is blind to everything about a partition except its multiset of block energies.

\end{document}